\documentclass[aps,10pt,twocolumn,showpacs,pra,citeautoscript,amsmath,amssymb,floatfix,nofootinbib]{revtex4-1}
\usepackage[utf8]{inputenc}
\usepackage[T1]{fontenc}
\usepackage{amsthm}
\usepackage{braket}
\usepackage{color}
\usepackage{graphicx}
\usepackage{verbatim}

\definecolor{dullmagenta}{rgb}{0.4,0,0.4}   
\definecolor{darkblue}{rgb}{0,0,0.4}
\usepackage[bookmarks,colorlinks,breaklinks]{hyperref}  
\hypersetup{linkcolor=red,citecolor=blue,filecolor=dullmagenta,urlcolor=darkblue} 

\usepackage{pgfplots}
\pgfplotsset{compat=1.14}

\DeclareMathOperator{\sign}{sign}

\DeclareMathOperator{\Span}{span}

\newcommand{\vertiii}[1]{{\left\vert\kern-0.25ex\left\vert\kern-0.25ex\left\vert #1 
    \right\vert\kern-0.25ex\right\vert\kern-0.25ex\right\vert}}
    
\allowdisplaybreaks

\newtheorem{corollary}{Corollary}
\newtheorem{lemma}{Lemma}

\newtheorem{proposition}{Proposition}
\newtheorem{theorem}{Theorem}

\usepackage[normalem]{ulem} 

\newcommand{\ketbra}[2]{\ket{#1}\!\!\bra{#2}}

\makeatletter
\newcommand*{\defeq}{\mathrel{\rlap{%
                     \raisebox{0.3ex}{$\m@th\cdot$}}%
                     \raisebox{-0.3ex}{$\m@th\cdot$}}%
                     =}
\makeatother

\newcommand\mpwS[1]{{\let\helpcmd\sout\parhelp#1\par\relax\relax}}
\long\def\parhelp#1\par#2\relax{%
	\helpcmd{#1}\ifx\relax#2\else\par\parhelp#2\relax\fi%
}

\usepackage{chapterbib}

\begin{document}

\title{Approximate quantum non-demolition measurements}
\author{Sami Boulebnane}
\author{Mischa P.\ Woods}
\author{Joseph M.\ Renes}
    \affiliation{%
     Institute for Theoretical Physics, ETH Zürich, Switzerland
    }%

\begin{abstract}
With the advent of gravitational wave detectors employing squeezed light, quantum waveform estimation---estimating a time-dependent signal by means of a quantum-mechanical probe---is of increasing importance.
As is well known, backaction of quantum measurement limits the precision with which the waveform can be estimated, though these limits can in principle be overcome by ``quantum nondemolition'' (QND) measurement setups found in the literature. Strictly speaking, however, their implementation would require infinite energy, as their mathematical description involves Hamiltonians unbounded from below. 
This raises the question of how well one may approximate nondemolition setups with finite energy or finite-dimensional realizations. 
Here we consider a finite-dimensional waveform estimation setup based on the ``quasi-ideal clock'' and show that the estimation errors due to approximating the QND condition decrease slowly, as a power law, with increasing dimension. 
As a result, we find that good QND approximations require large energy or dimensionality.  
We argue that this result can be expected to also hold for setups based on truncated oscillators or spin systems.
\end{abstract}

\maketitle


The general problem of waveform estimation is to estimate a classical time-dependent signal $x(t)$ by coupling it to a probe system and repeatedly measuring the probe.
The difficulty in using quantum probe systems is that measurement causes back-action on the probe, limiting the overall precision of the scheme~\cite{braginsky_quantum_1992}. 
This is relevant not only to very small probes, e.g.\ optomechanical systems~\cite{kippenberg_cavity_2008,cripe_measurement_2019}, but also very large, such as LIGO~\cite{ligo_scientific_collaboration_enhanced_2013}, which has recently begun an observing run employing squeezed light~\cite{ligopr} as suggested by Caves~\cite{caves_quantum-mechanical_1981}.

To circumvent these limitations, a specific class of measurements---known as ``quantum nondemolition''---was identified~\cite{braginskii_quantum-mechanical_1975} and explored, particularly with application to gravitational wave detectors (see, e.g.\ \cite{braginsky_quantum_1980,caves_measurement_1980} and references therein). 
An observable $\hat O(t)$ (regarded in the Heisenberg picture) is quantum nondemolition if $[\hat O(t),\hat O(t')]=0$ for all $t$ and $t'$. 
When this condition only holds at discrete times, the observable is termed stroboscopic, otherwise continuous.

A static observable is a simple case of a QND observable in which $\hat O(t)=\hat O(t')$, either for all times (continuous) or periodically spaced times (stroboscopic). 
A prominent example useful for metrology is the ``back-action evading measurement'' of the co-rotating position quadrature of a harmonic oscillator~\cite{thorne_quantum_1978,caves_measurement_1980}.
More recently, building on Koopman's formulation of classical mechanics in Hilbert space~\cite{koopman_hamiltonian_1931}, Tsang and Caves showed how appropriate coupling of several quantum systems enables one to construct a collection of continuous QND  observables which satisfy any desired classical equations of motion~\cite{tsang_evading_2012}. 
For instance, the center of mass $\frac12(\hat q_1+\hat q_2)$ and relative momentum $\hat p_1-\hat p_2$ of two uncoupled oscillators, one of mass $m$ and the other of \emph{negative} mass $-m$, are QND observables of the position and momentum of a classical oscillator~\cite{tsang_coherent_2010}.

Unfortunately, as in this example, the mathematical description of non-static QND observables relies on unphysical Hamiltonians whose implementation would require infinite energy~\cite{sami}.
The mathematical issues are similar to those first raised by Pauli, of whether or not a time observable can exist in quantum theory~\cite{pauli_quantentheorie_1926,galapon_paulis_2002}. 
Of course, one need only approximate the QND condition by finite-dimensional or finite-energy truncations. 
For instance, to emulate the negative mass oscillator, one can employ symmetric red and blue sidebands of a carrier frequency~\cite{tsang_evading_2012} or use spin systems~\cite{julsgaard_experimental_2001}. 
The latter have already been applied to magnetometry~\cite{wasilewski_quantum_2010} and position measurement~\cite{moller_quantum_2017}; possible applications to LIGO were recently examined by Khalili and Polzik~\cite{khalili_overcoming_2018}. 
The question then becomes how the approximation limits the estimation scheme. 
These limitations may bear (among others) upon the properties of the waveform $x\left(t\right)$ to be estimated, or the frequency at which one is allowed to perform measurements. 
It is an interesting question of principle, and potentially of practical relevance in the near future, to determine how stringent these restrictions are for a given dimension or energy constraint. 
For instance, does the approximate quantum nondemolition setup approach the exact one exponentially fast as a function of dimension/energy, or with a slower convergence? 
One may regard the approximation as especially forgiving in the former case, as only a very modest investment would be needed to obtain excellent performance.

Another possible approximate QND system that we explore here involves the ``quasi-ideal clock'' of \cite{woods_autonomous_2019}, a finite-dimensional approximation of an idealized clock governed by the (unbounded) Hamiltonian $\hat H=v\hat p$, for some constant velocity $v$.
The dynamics of the idealized clock are just pure translation, $\hat q(t)=\hat q(0)+vt$ and $\hat p(t)=\hat p(0)$, as for the case of a classical free particle, and so its position records the time. 
Indeed, the idealized clock can be viewed as an instance of Tsang and Caves's construction, since the center of mass and relative momentum of two free (quantum) particles with opposite masses also satisfy the equations of motion of the classical free particle (with $v=(\hat p_1-\hat p_2)/m$ in this case). 

The quasi-ideal clock is a particularly simple QND system, as its free dynamics approximate the idealized case exponentially well in the dimension $d$~\cite{woods_autonomous_2019}. 
Nevertheless, as we show in this Letter, when subject to repeated measurement for waveform estimation, the quasi-ideal clock only approximates the idealized nondemolition setting with errors decaying as a power law in the dimension.
Namely, backaction limits the minimum achievable measurement precision as well as the minimum detectable signal strength to scale as $d^{-\frac 14}$ and $d^{-\frac 12}$, respectively, which translates into an energy scaling of $E^{-\frac 14}$ and $E^{-\frac12}$.  
We also discuss the other options and argue that they may be expected to have a similar power law scaling. 

\emph{Description of quantum measurement}.---%
For context, we begin by reviewing the general framework of quantum measurement, following \cite[Chapter 5]{braginsky_quantum_1992}. 
Given a quantum system prepared in a state $\ket{\Psi}$, 
consider $n$ sequential measurements corresponding to the observables $\hat{q}_1, \ldots, \hat{q}_n$. 
The (non-normalized) post-measurement state, given that one observed $\widetilde{q}_1, \ldots, \widetilde{q}_n$, is given by $\ket{\Psi'} = \hat{\Omega}_n\left(\widetilde{q}_n\right)\cdots\hat{\Omega}_1\left(\widetilde{q}_1\right)\ket{\Psi}$, where $\hat{\Omega}_j(\widetilde{q}_j)$ is the Kraus operator corresponding to outcome $\widetilde q_j$ of the $j$th measurement. 
The joint probability distribution of the outcomes $\widetilde{q}_1, \ldots, \widetilde{q}_n$ is obtained from the norm squared of the latter. Typically, the Kraus operator $\hat{\Omega}_j\left(\widetilde{q}_j\right)$ is constructed by ``smearing'' the projector onto the eigenspace of $\hat{q}_j$ associated to the eigenvalue $\widetilde{q}_j$, e.g.\ by a Gaussian. 
Physically, this corresponds to making an imprecise measurement of $\hat{q}_j$.
We denote the imprecision by $\sigma_{\text{m}}$, and assume it is the same for all $j$. 
 An interesting particular case is when the observables $\hat{q}_1, \ldots, \hat{q}_n$ are given by a Heisenberg picture operator $\hat{q}$ evaluated at different times:\, $\hat{q}_j\! \defeq \hat{q}(t_j)$ ($t_1 \leq \ldots \leq t_n$), which amounts to considering a measurement of a fixed observable $\hat{q}$ at times $t_1, \ldots, t_n$.

This gives a prescription for computing the joint probability distribution of the measurement outcomes for any series of measurements. 
The moments of this disribution enjoy particularly simple expressions when the measurements are \emph{linear}, meaning the corresponding observables commute up to a scalar: $\left[\Hat{q}_j, \Hat{q}_l\right] = ic_{jl}$, for $c_{jl}\in \mathbb C$. 
Linear measurements cover several elementary quantum systems including the harmonic oscillator and free particle, which are of special interest to metrology.
Here, measurement backaction does not show up in the first moments, as $\langle \widetilde q_j\rangle = \braket{\Psi|\Hat q_j|\Psi}$. 
It does however show up in the second moments. 
Letting $B_{jl}=\left\langle(\widetilde{q}_j - \left\langle\widetilde{q}_j\right\rangle)(\widetilde{q}_l - \left\langle\widetilde{q}_l\right\rangle)\right\rangle$ and $B_{jl}^{\text{init}}$ the symmetrized correlation evaluated using $\ket{\Psi}$, \cite[Chapter 5]{braginsky_quantum_1992} shows that 
\begin{align}
B_{jl}=B_{jl}^{\text{init}}+\delta_{jl}\sigma_{\text{m}}^2 + \sum_{1 \leq n \leq j, l}\frac{c_{jn}c_{ln}}{4\sigma_{\text{m}}^2}\label{eq:linear_measurements_correlations}\,.
\end{align}
The first term describes the contributions from the wavefunction, the second term the imprecision of measurement, and the third term the quantum backaction. 
Precise measurements make the third term large, and with it the variance of the measurement result.

\emph{The quasi-ideal clock}.---%
Consider an odd $d$-dimensional Hilbert space, whose basis elements $\ket{k}$ we label using the integers $\mathbb Z_d$, ranging from $-\frac{d-1}2$ to $\frac{d-1}2$.
The Hamiltonian of the quasi-ideal clock is simply $\hat H_d=\frac{2\pi}{\sqrt d}\sum_{k\in \mathbb Z_d} k \ketbra{k}{k}$.
The discrete Fourier transform of the energy eigenstates defines the ``time eigenstates'' $\ket{\theta_j} \defeq \frac{1}{\sqrt{d}}\sum_{k\in \mathbb Z_d} e^{-\frac{2\pi ijk}{d}}\ket{k}$, which are eigenvalues of the ``time operator'' $\hat{T}_d \defeq \sum_{j\in \mathbb Z_d} j \ketbra{\theta_j}{\theta_j}$.
The time eigenstates have the property that $\ket{\theta_j}$ is transformed to $\ket{\theta_{j+1}}$ under evolution by time $1/\sqrt d$, meaning the the system stroboscopically emulates the idealized case of pure translation~\cite{salecker_quantum_1958}.
Remarkably, this feature persists for all evolution times, up to an \emph{exponentially} small error (as a function of the dimension or energy), provided the state is restricted to the ``quasi-ideal states'' of \cite{woods_autonomous_2019} (and not necessarily otherwise~\cite{peres_measurement_1980}). 
Essentially, these states consist of a Gaussian superposition of time eigenstates, with mean energy $E\propto d$ above the ground state and width growing as $d^\lambda$ for some $\lambda\in (0,1)$.

Consider, now, the quasi-ideal clock coupled to a classical waveform through the time-dependent Hamiltonian $(1 + x(t))\hat{H}_d$.
We eschew the question of how to engineer such a coupling, as our focus is on in principle limits. 
For the idealized clock, $\hat q(t)=\hat q(0)+t+\int_{0}^t\!\textrm{d}\tau\,x(\tau)$, so the waveform can in principle be reconstructed from $\widetilde q(t)$. 
The finite-dimensional analog of the position operator $\hat q$ is the rescaled time operator $\hat \xi_{j}\defeq \hat T_d(t_j)/\sqrt{d}$, as evolution by time $1/\sqrt d$ advances the clock value by precisely this amount when $x=0$. 
Hence, setting $t_j=j/\sqrt{d}$ for integer $j$, one finds that $\hat \xi_{j}-\hat \xi_{j-1}$ furnishes an estimate of $\int_{t_{j-1}}^{t_j}\!\textrm{d}\tau\,(1+x(\tau))$. 
Observe that the resulting measurements are not linear.

We take the initial state to be a quasi-ideal state of variance $\frac{\sigma_{\text{s}}^2 d}{4\pi}$, where $\frac1{ d}\ll\sigma_{\text{s}}^2\ll  d$, and the measurement precision to be given by $\sigma_{\text{m}}$. 
The evolution of the clock state between $t_{j-1}$ and $t_j$ is given by the unitary $e^{-i \hat H_d \Delta t_j/\sqrt{d}}$, where 
$\Delta t_j=1+\int_{j-1}^j \!\mathrm{d}\tau\, x(\tau/\sqrt{d})$.
Equivalently, one may regard this as a measurement of a \textit{freely} evolving quasi-ideal clock at successive time intervals $\frac{\Delta t_1}{\sqrt{d}}, \ldots, \frac{\Delta t_n}{\sqrt{d}}$.

\emph{Backaction scaling}.---It would be desirable to compute the lowest-order moments $\langle\widetilde{\xi}_j\rangle$ and $\langle \widetilde{\xi}_j\widetilde{\xi}_l\rangle$, but this is technically awkward due to periodic boundary conditions. 
Instead, for integers $\ell$ and $m$, investigating $\langle e^{\frac{2\pi i\ell \widetilde{\xi}_j}{\sqrt{d}}}e^{\frac{2\pi im\widetilde{\xi}_k}{\sqrt{d}}}\rangle$ leads to a much more tractable problem and nonetheless allows for a nice analogy with the theory of linear measurements.  
These quantities carry information about both the expected values of the measurements and their correlations.
For illustration, the random variable $X\sim\mathcal{N}(\mu, \sigma^2)$ with $\alpha \in \mathbb{R}$, yields $\langle e^{i\alpha X}\rangle=e^{i\alpha\mu}e^{-\frac{1}{2}\alpha^2\sigma^2}$. 
In general, the phase of $\langle e^{\frac{2\pi i\widetilde{\xi}_j}{\sqrt{d}}}\rangle$ carries information about the expectation of $\widetilde{\xi}_j$ (provided the latter is symmetrically distributed around its mean), while the modulus carries information about its dispersion.

Let us sketch the result for $\langle e^{-\frac{2\pi i\widetilde{\xi}_n}{\sqrt{d}}}\rangle$,  corresponding to $\ell = -1, m = 0$, reserving details for \S\ref{sec:IId1}. 
Similarly to the case of linear measurements, the expectation value can be divided into three contributions, as
\begin{align}
\label{eq:pseudo_variance_formula}
     \big\langle e^{-\frac{2\pi i\widetilde{\xi}_n}{\sqrt{d}}}\big\rangle & = e^{-\frac{2\pi i}{d}\sum_{j=1}^n \Delta t_j}C_1C_2 C_3\,,
\end{align}
where $C_1$ depends only on $\sigma_{\text{s}}$, $C_2$ only on $\sigma_{\text{m}}$, and $C_3$ on both  $\sigma_{\text{s}}$ and $\sigma_{\text{m}}$ as well as $\{\Delta t_j\}_j$. 
The first contribution, consisting of $C_1$ and the phase factor, is roughly analogous to the contribution of the wavefunction to the first and second moments in the linear case: It turns out that $C_1$ behaves as 
$e^{-\frac{\pi\sigma_{\text{s}}^2}{2d}}$ for large $d$, and $e^{-\frac{2\pi i}{d}\sum_{j=1}^n\Delta t_j - \frac{\pi\sigma_{\text{s}}^2}{2d}}$ is exactly the classical expectation of $e^{-\frac{2\pi iX}{\sqrt{d}}}$ when $X \sim \mathcal{N}\left(\frac{1}{\sqrt{d}}\sum_{j=1}^n\Delta t_j, \frac{\sigma_{\text{s}}^2}{4\pi}\right)$. 
Meanwhile, the factor $C_2$ is essentially $e^{-\frac{\pi\sigma_{\text{m}}^2}{2d}}$, provided $\sigma_{\text m}^2 \ll d$. 
Therefore, it becomes trivial, i.e.\ 1, when $\sigma_{\text{m}} \rightarrow 0$. 
For this reason one may regard it as analogous to the second term in \eqref{eq:linear_measurements_correlations}.
Since we have identified the analogs of all the factors appearing in \eqref{eq:linear_measurements_correlations} \textit{except} the term coming from the quantum backaction, one may by default regard the $C_3$ contribution as analogous to the backaction. 
There is also a more positive argument supporting this conclusion, as one can show that $C_3=1$ if all the $\Delta t_j$ are integers. 
This corresponds to commuting $\hat T_d$ operators, i.e.\ the case of no backaction. 

It turns out that $C_3$ has a simple form related to a random walk. 
Here we give the general picture; the precise details are spelled out in \S\ref{sec:measured_quasi_ideal_clock} 
 and hold for all $\ell$ and $m$. 
The walk takes place on the discrete ring $\mathbb Z_d$, and the step size varies according to a roughly Gaussian distribution of zero mean and variance $\frac{d}{4\pi\sigma_\text{m}^2}$. 
Nontrivial contributions to $C_3$ occur whenever the walk lands on $\frac{d-1}2$. 
Calling $Z_j$ the position at step $j$, we have 
\begin{align}
\label{eq:pseudo_variance_backaction}
C_3=\sum_{z_1\in \mathbb Z_d}P(z_1)\mathbb E^{z_1}\prod_{j=1}^n1-(1-e^{2\pi i\Delta t_j})\mathbf 1_{Z_j=\frac{d-1}2}\,,
\end{align}
where $\mathbb{E}^{z_1}$ denotes the expectation for a walk starting at $z_1$ and the distribution $P$ depends on $\sigma_{\text{s}}$.

To proceed further, we must specify the $\Delta t_j$ (or equivalently, $x(t)$). 
As previously mentioned, the case of integer $\Delta t_j$ gives $C_3=1$. 
Heuristically, the half-integer choice, e.g.\ $\Delta t_j=\tfrac 12$, can be expected to generate the largest backaction on the system, as this is the ``furthest away'' (for a comparable spacing of measurements) from the case of no backaction. 
Now let the number of measurements scale with $d$ as $n = 2t\sqrt{d}$,  so that the total measurement time is fixed at $\sum_{j=1}^n \frac{\Delta t_j}{\sqrt{d}} = t$, independently of $d$. 
The behavior of $C_3$ in terms of the scaling of $\sigma_{\text{m}}^2$ with respect to $d$ is worked out in detail in \S\ref{sec:measured_quasi_ideal_clock_derivation_formulae}, 
but the results can be appreciated from the form of \eqref{eq:pseudo_variance_backaction}. 
The walk will spend a significant amount of time on the last position when the variance after $n$ steps is the size of the ring: $n\frac d{\sigma_{\text{m}}^2}\gg d^2$, which for the given choice of $n$ gives the condition $\sigma_{\text{m}}^2\ll 1/\sqrt d$.
The detailed calculation shows that in this case $C_3$ is bounded away from $1$ by $2t/\sqrt d$, while $C_3\approx 1$ up to a deviation exponentially small in $d$ when $\sigma_{\text{m}}^2\gg 1/\sqrt d$. 
A numerical simulation of the difference of the two cases is illustrated in Figure~\ref{fig:simulations}.
\begin{figure}
\centering
\begin{tikzpicture}
\begin{loglogaxis}[
    xmax=10^4,
    xmin=500,
    width=0.5*\textwidth,
    height=0.25*\textwidth,
    xlabel=$d$,
    xlabel near ticks,
    title={$\sigma_{\text{m}}^2\propto d^{-0.12}$, $\sigma_{\text{s}}^2\propto d^{-0.5}$},
    title style={at={(0.5,0.92)}},
    xlabel shift={-12pt},
    legend style={legend cell align=left}]
\addplot[only marks,mark=+] table[x index=0,y index=2] {plot1.dat};
\addlegendentry{$1-C_1C_2C_3$}
\addplot[only marks,mark=x] table[x index=0,y index=1] {plot1.dat};
\addlegendentry{$1-C_1C_2$}
\end{loglogaxis}
\end{tikzpicture}
\vskip 10pt
\begin{tikzpicture}
\begin{loglogaxis}[
    xmax=10^4,
    xmin=500,
    xminorticks=true,
    yminorticks=true,
    ytick={0.000001,0.00001,0.0001,0.001,0.01,0.1},
    yticklabels={,$10^{-5}$, ,$10^{-3}$, ,$10^{-1}$},
    xlabel near ticks,
    xlabel shift={-12pt},
    width=0.5*\textwidth,
    height=0.25*\textwidth,
    xlabel=$d$,
    title style={at={(0.5,0.92)}},
    title={$\sigma_{\text{m}}^2\propto d^{-0.65}$, $\sigma_{\text{s}}^2\propto d^{\,0}$},
    legend style={legend cell align=left, 
        at={(0.977,0.5)},
        anchor=east}]
\addplot[only marks,mark=+] table[x index=0,y index=2] {plot2.dat};
\addlegendentry{$1-C_1C_2C_3$}
\addplot[only marks,mark=x] table[x index=0,y index=1] {plot2.dat};
\addlegendentry{$1-C_1C_2$}
\end{loglogaxis}
\end{tikzpicture}
\caption{\label{fig:simulations} Scaling of $C_1, C_2, C_3$ with dimension $d$ for different scalings of $\sigma_{\text{s}}$ and $\sigma_{\text{m}}$. Monte Carlo simulations of $5000$ samples, with $t = 1$ fixed.}
    
\end{figure}
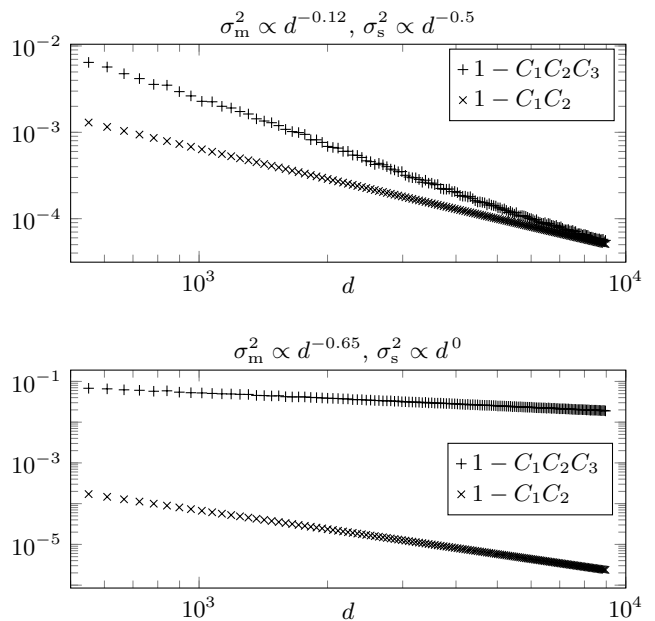

Combining the scaling behavior of $C_2$ and $C_3$, it is apparent that one cannot achieve a variance smaller than $1/\sqrt d$ on the measurement of $\widetilde{\xi}_n$.
Due to the form of the Hamiltonian, this is essentially $1/\sqrt E$, for $E$ the mean energy of the quasi-ideal clock. 
This scaling can be approached by taking $\sigma_{\text{m}}^2 \propto d^{-\frac12 + \varepsilon}$, with $\varepsilon > 0$ small. 
If, however, $\sigma_{\text{m}}^2 \ll 1/\sqrt d$, the variance on $\widetilde{\xi}_n$ is at least $\sqrt d$. 
Although these results were derived for the case of half-integer $\Delta t_j$, they generalize in a straightforward way to the case of a random waveform consisting of  white noise of variance $\sigma^2$: $x(t) = \sigma\frac{\mathrm{d}W(t)}{\mathrm{d}t}$~, as shown in \S\ref{sec:IId3} 
(a standard way of benchmarking a statistical estimator; see the Cramér-Rao bound in~\cite{van_trees_detection_2001}). 
In case the waveform is completely general, it is still possible to show that one may achieve a variance as small as $d^{-\frac12 + \varepsilon}$ (for all fixed $\varepsilon > 0$) on the measurement of $\widetilde{\xi}_n$, but we have no clear proof that this scaling is optimal. 

\emph{Waveform estimation}.---%
It remains to be seen whether one can perform efficient waveform estimation given the above constraints on $\widetilde{\xi}_j$.
Returning to the case of continuous $x$, given that typical errors of $\widetilde{\xi}_j - \widetilde{\xi}_{j - 1}$ will be roughly of size $d^{-\frac 14}$ and from this quantity we aim to estimate $\int_{\frac{j - 1}{\sqrt{d}}}^{\frac{j}{\sqrt{d}}}\!\mathrm{d}\tau\,\left(1 + x(\tau)\right)$, the magnitude $|x|$ of the waveform must satisfy $ |x|\frac{1}{\sqrt{d}} \gg d^{-\frac14}$ for the error not to overwhelm the expectation.
This is unsatisfying, as it means the smallest detectable signal strength \emph{increases} with increasing dimension as $d^{\frac14}$. 
The difficulty is that the measurements are too sharp for how frequently they are occurring, and we should either contemplate weaker measurements at the same frequency or less frequent measurements. 
Suppose that the number of measurements $n$ in fixed time $t$ scales as $d^{\gamma}$ for $0\leq \gamma\leq \tfrac12$. The condition on the variance is now $\sigma_{\text{m}}^2\gg d^{\gamma-1}$, so that the condition on the waveform becomes $|x| \gg d^{\frac{3}{2}\gamma - \frac{1}{2}}$. 
As a result, provided $\gamma < \frac{1}{3}$, the smallest detectable signal is allowed to vanish as $d \to \infty$, though strictly slower than $d^{-\frac12}$.
Of course, measuring less frequently impacts the useful bandwidth of the procedure, and the maximum detectable frequency $f_{\max}$ is bounded by $f_{\max} < d^{\gamma}$. 
For example, choosing $\gamma=\frac16$ and working in terms of $E$ gives $|x|\ll 1/\sqrt E$ and $f_{\max}<E^{\frac13}$. 


Improved scalings can be obtained by more sophisticated estimation procedures, but likely not an exponential improvement. 
For instance, one might like to employ smoothing, as formulated in the quantum case by Tsang~\cite{tsang_time-symmetric_2009,tsang_optimal_2009,tsang_optimal_2010}, or indeed any other estimation technique designed for continuous signals. 
To do so requires a formulation of the measurement process in a suitable limit 
as a continuous-time process.
Note that our setup is outside the usual limiting procedure of ever weaker measurements made ever more often~\cite{barchielli_statistics_1983,caves_quantum-mechanical_1987,braginsky_quantum_1992,wiseman_quantum_2009}.
In contrast, here we call for stronger measurements made ever more often on an ever larger system, for particular scaling of the former two as a function of the latter. 
Using our results and techniques from \cite[sections 7 and 13]{billingsley_convergence_1999}, it can be shown that as $d\to \infty$, the measurement process converges weakly to a well-defined continuous-time process in certain cases: to a deterministic motion if the measurement is moderately sharp ($1/\sqrt{d}	 \ll \sigma_{\text{m}}^2 \ll 1$) and to a Cauchy process with drift for sharp measurements ($\sigma_{\text{m}} = 0$).
The former case is precisely the behavior we expect for an idealized clock, namely zero-error, which reinforces our conclusions above.
A different limit procedure is needed to construct continuous time estimators for finite $d$, but studying the speed of convergence of this limit may be useful in this regard. 
Finally, to underscore the relative crude nature of our estimator, we note that estimating the position of a forced oscillator, as considered e.g.\ in~\cite{tsang_fundamental_2011}, using finite differences does not seem to work, as their variance diverges in the limit, as detailed in~\S\ref{sec:measured_oscillator}.

\emph{Other approximate QND systems}.---%
The oscillators in the QND construction of~\cite{tsang_evading_2012} can be approximated by truncated oscillators or by spin systems via the Holstein-Primakoff approximation. 
Both have free evolution that well-approximates the idealized case. 
Indeed, the time evolution of spin coherent states exactly emulates the idealized case, since a rotation of a spin coherent state by $J_z$ produces another spin coherent state. 
The former behaves similarly to the quasi-ideal clock in that the free evolution is exponentially good, provided the wavefunction is taken much wider than $1/\sqrt d$ but much narrower than $\sqrt d$, where $d$ is the truncated dimension, and the energy scales as the dimension.   
However, it appears from numerical investigation that the accuracy of waveform estimation scales not as favorably with the dimension. 
Coupled with the fact that two oscillators are needed for the QND setup of~\cite{tsang_evading_2012}, it is not unreasonable to expect a worse scaling in estimator accuracy with energy.

\emph{Conclusion}.---%
The polynomial scaling of the error in waveform estimation in dimension or energy of the quasi-ideal clock echoes similar error scalings when it is used for timekeeping~\cite{woods_quantum_2018} or covariant quantum error correction~\cite{woods_continuous_2019}. 
Indeed, in \cite{woods_continuous_2019} the bound achieved via the quasi-ideal clock is proven to be the optimally achievable rate permissible by quantum mechanics~\cite{faist_continuous_2019}---suggesting that the scaling derived in this Letter may also be optimal.
Its simple structure enables relatively straightforward mathematical analysis, compared with the double oscillator systems; though in light of their practical application~\cite{tsang_coherent_2010,tsang_evading_2012,khalili_overcoming_2018}, it would be interesting and useful to more thoroughly characterize the error scaling in those cases. 
To enable a more sophisticated error analysis, it would also be useful to formulate a continuous limit. 
Perhaps, unlike our considerations above, one can fix $d$ and scale $\sigma_{\text{m}}$ and $n$ to obtain a useful limit. 

\emph{Acknowledgments}.---%
We acknowledge useful discussions with Carlton M.\ Caves. 
This work was supported by the Swiss National Science Foundation (SNSF) via the National Centre of Competence in Research ``QSIT''. M.P.W. acknowledges funding from the SNSF (AMBIZIONE Fellowship PZ00P2\_179914).

\bibliography{bibliography}

\onecolumngrid

\newpage

\appendix

\section*{Appendix}
Here we present detailed proofs and more extensive discussions of the main results of the paper. In section \ref{sec:measured_oscillator}, we treat the repeatedly measured quantum harmonic oscillator coupled to a classical time-dependent force; we show that owing to quantum backaction, one may not precisely estimate the latter by means of continuous measurements (hence the motivation behind non-demolition measurement). We then establish (section \ref{sec:measured_quasi_ideal_clock}) the main result of the paper, involving a waveform estimation setup consisting of a continuously measured quasi-ideal clock coupled to the time-dependent signal. Before moving to the main matter, we prove some instructive preliminary estimates on the freely evolving (i.e.\ without measurement) quasi-ideal clock. Finally, section \ref{sec:calculus} gathers all the mathematical results used in the body of this document.
\section{Measured quantum harmonic oscillator coupled to a time-dependent classical force}
\label{sec:measured_oscillator}

The case of the driven harmonic oscillator, described by the time-dependent Hamiltonian
\begin{align}
    \Hat{H}(t) & = \frac{1}{2}\Hat{q}^2 + \frac{1}{2}\Hat{p}^2 - x(t)\hat{q}\,,
\end{align}
is of special interest to metrology. In this setup, generally speaking, the purpose is to infer something on the time-dependent force $x$ by monitoring the position of the oscillator (or more generally any quadrature, the angle of the quadrature being possibly time-dependent). The Heisenberg-picture $\hat{q}$ and $\hat{p}$ operators for this system are given by:
\begin{align}
    \hat{q}(t) & = \cos(t)\hat{q} + \sin(t)\hat{p} + \int_0^t\!\mathrm{d}t'\,x\left(t'\right)\sin\left(t - t'\right)\,,\label{eq:q had 1}\\
    \hat{p}(t) & = \cos(t)\hat{p} - \sin(t)\hat{q} + \int_0^t\!\mathrm{d}t'\,x\left(t'\right)\cos\left(t - t'\right)\,.\label{eq:q had 2}
\end{align}
The two-times commutators of these observables are manifestly independent from the classical force. In particular, for the position
\begin{align}
    \left[\hat{q}\left(t'\right), \hat{q}\left(t\right)\right] & = i\sin(t - t')\,.
\end{align}
Now, it is interesting to figure out what one may learn about $x$ from measuring the position $q$. Since
\begin{align}
    \frac{\mathrm{d}^2\hat{q}(t)}{\mathrm{d}t^2} + \hat{q}(t) & = x(t)\,,
\end{align}
it may be tempting (at least conceptually) to think of continuously measuring $q$, yielding some continuous time series $\widetilde{q}(t)$ and then use\footnote{In the following we use primes to indicate derivatives w.r.t. $t$.} $\widetilde{q}''(t) + \widetilde{q}(t)$ as an estimator for $x(t)$. Since the formalism described up to now deals with discrete instead of continuous measurements, one should start from a discrete setting and approach the continuum one by a limiting process where the scaling of the free parameters (essentially the $\Delta q$) is to be specified. Let then $\tau > 0$ denote a ``unit time step''. This means that one measures the positions at times $t_1 = 0, t_2 = \tau, \ldots, t_n = (n - 1)\tau$ and that we approximate the second derivative of the idealized estimator $y''(t) + y(t)$ by a finite difference. In other words, our estimate for $x(t_j)$ reads
\begin{align}
    & \widetilde{q}_j + \frac{\widetilde{q}_{j + 1} + \widetilde{q}_{j - 1} - 2\widetilde{q}_j}{\tau^2}\, .
\end{align}
One may now compute the variance of this estimator:
\begin{align}
    & \left\langle\left(\widetilde{q}_j + \frac{\widetilde{q}_{j + 1} + \widetilde{q}_{j - 1} - 2\widetilde{q}_j}{\tau^2} - \left\langle\widetilde{q}_j + \frac{\widetilde{q}_{j + 1} + \widetilde{q}_{j - 1} - 2\widetilde{q}_j}{\tau^2}\right\rangle\right)^2\right\rangle\nonumber\\
    & = \frac{1}{\tau^4}\left\langle\left[(\tau^2 - 2)\left(\widetilde{q}_j - \langle\widetilde{q}_j\rangle\right) + \left(\widetilde{q}_{j + 1} - \langle\widetilde{q}_{j + 1}\rangle\right) + \left(\widetilde{q}_{j - 1} - \langle\widetilde{q}_{j - 1}\rangle\right)\right]^2\right\rangle\\
    & = \frac{1}{\tau^4}\left\langle(\tau^2 - 2)^2\left(\widetilde{q}_j - \langle\widetilde{q}_j\rangle\right)^2 + \left(\widetilde{q}_{j + 1} - \langle\widetilde{q}_{j + 1}\rangle\right)^2 + \left(\widetilde{q}_{j - 1} - \langle\widetilde{q}_{j - 1}\rangle\right)^2\right.\nonumber\\
    & \left.\hspace{0.1\textwidth} + 2(\tau^2 - 2)\left(\widetilde{q}_j - \langle\widetilde{q}_j\rangle\right)\left(\widetilde{q}_{j + 1} - \langle\widetilde{q}_{j + 1}\rangle\right) + 2\left(\widetilde{q}_j - \langle\widetilde{q}_j\rangle\right)\left(\widetilde{q}_{j - 1} - \langle\widetilde{q}_{j - 1}\rangle\right)\right.\nonumber\\
    & \left.\hspace{0.1\textwidth} + 2\left(\widetilde{q}_{j + 1} - \langle\widetilde{q}_{j + 1}\rangle\right)\left(\widetilde{q}_{j - 1} - \langle\widetilde{q}_{j - 1}\rangle\right)\right. \Big\rangle\\
    & = \frac{1}{\tau^4}\Bigg(\left.(\tau^2 - 2)^2\braket{\Psi|\left(\hat{q}_j - \langle\hat{q}_j\rangle\right)^2|\Psi} + \braket{\Psi|\left(\hat{q}_{j + 1} - \langle\hat{q}_{j + 1}\rangle\right)^2|\Psi} + \braket{\Psi|\left(\hat{q}_{j - 1} - \langle\hat{q}_{j - 1}\rangle\right)^2|\Psi}\right.\nonumber\\
    & \left.\hspace{0.1\textwidth} + 2(\tau^2 - 2)\braket{\Psi|\frac{1}{2}(\hat{q}_j\hat{q}_{j + 1} + \hat{q}_{j + 1}\hat{q}_j)|\Psi} + 2(\tau^2 - 2)\braket{\Psi|\frac{1}{2}(\hat{q}_j\hat{q}_{j - 1} + \hat{q}_{j - 1}\hat{q}_j)|\Psi}\right.\nonumber\\
    & \left.\hspace{0.1\textwidth} + 2\braket{\Psi|\frac{1}{2}(\hat{q}_{j - 1}\hat{q}_{j + 1} + \hat{q}_{j + 1}\hat{q}_{j - 1})|\Psi}\right.\nonumber\\
    & \left. \hspace{0.1\textwidth} + (\tau^2 - 2)^2(\Delta q_j)^2 + (\Delta q_{j + 1})^2 + (\Delta q_{j - 1})^2\right.\nonumber\\
    & \left. \hspace{0.1\textwidth} + (\tau^2 - 2)^2\sum_{1 \leq n \leq j - 1}k_{j, n}^2\sigma_n^2 + \sum_{1 \leq n \leq j - 2}k_{j - 1, n}^2\sigma_n^2 + \sum_{1 \leq n \leq j}k_{j + 1, n}^2\sigma_n^2\right.\nonumber\\
    & \left. \hspace{0.1\textwidth} + 2(\tau^2 - 2)\sum_{1 \leq n \leq j - 1}k_{j, n}k_{j - 1, n}\sigma_n^2 +  2(\tau^2 - 2)\sum_{1 \leq n \leq j}k_{j, n}k_{j + 1, n}\sigma_n^2\right.\nonumber\\
    & \left.\hspace{0.1\textwidth} + 2\sum_{1 \leq n \leq j - 1}k_{j - 1, n}k_{j + 1, n}\sigma_n^2 \right) \,.
\end{align}
We will now focus on the contribution of the backaction terms. Furthermore, we will assume that the Kraus operator is Gaussian, so that $\Delta q_j$ is minimized given $\sigma_j$ and therefore equals $\frac{1}{4\sigma_j^2}$. This leads to:
\begin{align}
\label{eq:qho_force_estimator_all_expanded}
    & \frac{1}{\tau^4}\left.\Bigg((\tau^2 - 2)^2(\Delta q_j)^2 + (\Delta q_{j + 1})^2 + (\Delta q_{j - 1})^2\right.\nonumber\\
    & \left. \hspace{0.1\textwidth} + (\tau^2 - 2)^2\sum_{1 \leq n \leq j - 1}k_{j, n}^2\sigma_n^2 + \sum_{1 \leq n \leq j - 2}k_{j - 1, n}^2\sigma_n^2 + \sum_{1 \leq n \leq j}k_{j + 1, n}^2\sigma_n^2\right.\nonumber\\
    & \left. \hspace{0.1\textwidth} + 2(\tau^2 - 2)\sum_{1 \leq n \leq j - 1}k_{j, n}k_{j - 1, n}\sigma_n^2 +  2(\tau^2 - 2)\sum_{1 \leq n \leq j}k_{j, n}k_{j + 1, n}\sigma_n^2\right.\nonumber\\
    & \left.\hspace{0.1\textwidth} + 2\sum_{1 \leq n \leq j - 1}k_{j - 1, n}k_{j + 1, n}\sigma_n^2 \right)\nonumber\\
    & = \frac{1}{\tau^4}\left(\frac{1}{4}\left(\frac{(\tau^2 - 2)^2}{\sigma_j^2} + \frac{1}{\sigma_{j + 1}^2} + \frac{1}{\sigma_{j - 1}^2}\right) + \sum_{1 \leq n \leq j - 2}\left((\tau^2 - 2)k_{j, n} + k_{j - 1, n} + k_{j + 1, n}\right)^2\sigma_n^2\right.\nonumber\\
    & \left.\hspace{0.05\textwidth} + \left((\tau^2 - 2)k_{j, j - 1} + k_{j + 1, j - 1}\right)^2\sigma_{j - 1}^2 + k_{j + 1, j}^2\sigma_j^2\right.\Bigg).
\end{align}
One now rewrites the terms in $k$ explicitly to exhibit their scaling in $\tau$:
\begin{align}
    & (\tau^2 - 2)k_{j, n} + k_{j - 1, n} + k_{j + 1, n}\nonumber\\
    & = (\tau^2 - 2)\sin(t_n - t_j) + \sin(t_n - t_{j - 1}) + \sin(t_n - t_{j + 1})\\
    & = (\tau^2 - 2)\sin(t_n - t_j) + \sin(t_n - t_j - \tau) + \sin(t_n - t_j + \tau)\\
    & = (2\cos(\tau) - 2 + \tau^2)\sin(t_n - t_j)\\
    & = \left(2\frac{\tau^4}{4!} - 2\frac{\tau^6}{6!} + \ldots\right)\sin(t_n - t_j) \,, \\
   {}\nonumber\\
    & (\tau^2 - 2)k_{j, j - 1} + k_{j + 1, j - 1}\nonumber\\
    & = (\tau^2 - 2)\sin(t_{j - 1} - t_j) + \sin(t_{j - 1} - t_{j + 1})\\
    & = -(\tau^2 - 2)\sin(\tau) - \sin(2\tau)\\
    & = -\sin(\tau)\left(2\cos(\tau) - 2 + \tau^2\right)\\
    & = -\left(\tau - \frac{\tau^3}{3!} + \ldots\right)\left(2\frac{\tau^4}{4!} - 2\frac{\tau^6}{6!} + \ldots\right)\, , \\
{}\nonumber\\
    k_{j + 1, j} & = \sin(t_j - t_{j + 1})\\
    & = -\sin(\tau) \,.
\end{align}
But assuming $\tau \leq 1$ (which is reasonable since one wants $\tau \downarrow 0$ in the end) and using therefore $(\tau - 2)^2 \geq 1, \sin(\tau) \geq \frac{2}{\pi}\tau$,
\begin{align}
    \frac{1}{\tau^4}\left(\frac{(\tau^2 - 2)^2}{\sigma_j^2} + k_{j + 1, j}^2\sigma_j^2\right) & \geq \frac{1}{\tau^4}\left(\frac{1}{4\sigma_j^2} + \frac{4}{\pi^2}\sigma_j^2\tau^2\right)\\
    & = \frac{1}{\tau^3}\left(\frac{1}{4\sigma_j^2\tau} + \frac{4}{\pi^2}\sigma_j^2\tau\right)\\
    & \geq \frac{2}{\pi}\frac{1}{\tau^3}\,,
\end{align}
and this minimum is achieved only if $\sigma_j^2$ scales as $\frac{1}{\tau}$. One can check that the other terms in equation \eqref{eq:qho_force_estimator_all_expanded} do not grow more rapidly as $\tau \downarrow 0$.\footnote{Keeping in mind that in the limiting process, one shall also take the index of the measurement $n$ to scale as $\frac{1}{\tau}$, since one considers a measurement at a fixed time and $\tau$ is the interval between two consecutive measurements.} This means that for a given force $x$, one may certainly not let $\tau \downarrow 0$ if one wants to get any information at all about the force! To compute the expectation of the estimator (therefore for finite $\tau$), let us remark that using equations \eqref{eq:q had 1}, \eqref{eq:q had 2} we obtain 
\begin{align}
    & \hat{q}(t) + \frac{\hat{q}(t + \tau) + \hat{q}(t - \tau) - 2\hat{q}(t)}{\tau^2}\nonumber\\
    & = \frac{2\cos(\tau) - 2 + \tau^2}{\tau^2}\cos(t)\hat{q} + \frac{2\cos(\tau) - 2 + \tau^2}{\tau^2}\sin(t)\hat{p}\nonumber\\
    & \hspace{0.05\textwidth} + \frac{2\cos(\tau) - 2 + \tau^2}{\tau^2}\int_0^t\!\mathrm{d}t'\,x(t')\sin(t - t') + \frac{1}{\tau^2}\int_0^{\tau}\!\mathrm{d}u\,\sin(\tau - u)\left(x(t + u) + x(t - u)\right)\,.
\end{align}
Provided $x$ is regular enough and $\int_0^t\!\mathrm{d}t'\,x(t')\sin(t - t')$ is bounded by some constant when $t \geq 0$, the expectation of this quantity approximates $x(t)$ with an error $\mathcal{O}(\tau^2)$.

\section{Quasi-ideal clock: general properties and our setup}
\label{sec:quasi_ideal_clock}

In this part, which will lead to the proof of our main results, we will focus on the so-called ``quasi ideal clock''. This system, studied extensively in \cite{woods_autonomous_2019,woods_quantum_2018}, rests upon a discretization of the phase-space making extensive use the discrete Fourier transform. This allows for more tractable exact computations or estimates than with other discretization schemes such as the Holstein-Primakoff approximation or the mere truncation of the infinite-dimensional annihilation operator. An example of application \cite{woods_autonomous_2019} is the proof that the accuracy $R\in[1,\infty)$  \footnote{$R=\infty$ being a clock with zero error, $R=1$ being a useless clock.} of quantum clocks may scale with dimension $d$ as $d^{2 - \eta}$ for arbitrarily small fixed $\eta \in (0, 2)$ --- while $d$ is an upper bound for classical stochastic clocks. In this section, one will be concerned with the incorporation of measurement into the quasi-ideal clock --- keeping as close as possible to the general framework exposed in the main text and in greater detail in \cite[Chapter 5]{braginsky_quantum_1992}. As our setting and definitions slightly differ from those of the aforementioned papers\footnote{The difference lying most essentially in the use of Jacobi $\theta$ functions, which are more suitable for the boundary conditions of the problem, instead of Gaussians. This may be unimportant for simple estimates on the quasi-ideal clock, such as its approximation of the canonical commutation relations and the QND condition, but will prove crucial for the more involved analysis of measurement we will perform here.}, we will review them first. Then, we will apply the definitions and results hereby introduced to show that the quasi-ideal clock approximates the QND condition up to exponentially decaying terms in the dimension. Finally, we will move on to the analysis of the measurement; an interesting result will be the emergence of an error decreasing as a power law in $d$ (as opposed to exponentially without measurement) compared to the idealized infinite-dimensional system.

\subsection{General setting}\label{General setting}
The setting considered here is that of the quasi-ideal clock, described in detail in \cite{woods_autonomous_2019,woods_quantum_2018}. We recall in this paragraph the elements essential to the understanding of the following.

The $d$-dimensional --- take $d$ odd for simplicity --- Hilbert space of the quasi-ideal clock is $\Span\{ \ket{n} \}_{-\frac{d - 1}{2} \leq n \leq \frac{d - 1}{2}}$ (where the $\{\ket{n}\}$ form an orthonormal basis). The Hamiltonian governing the evolution of the system is:
\begin{align}
    \hat{H}_d & := \sum_{-\frac{d - 1}{2} \leq n \leq \frac{d - 1}{2}}n\ket{n}\,.
\end{align}
From the energy eigenstates $\ket{n},\, \left(-\frac{d - 1}{2} \leq n \leq \frac{d - 1}{2}\right)$, one may define what we will subsequently refer to as ``time eigenstates'' by way of the discrete Fourier transform:
\begin{align}
    \ket{\theta_k} & := \frac{1}{\sqrt{d}}\sum_{-\frac{d - 1}{2} \leq n \leq \frac{d - 1}{2}}e^{-\frac{2\pi ikn}{d}}\ket{n}, \qquad -\frac{d - 1}{2} \leq k \leq \frac{d - 1}{2}\,.
\end{align}
By unitarity of the discrete Fourier transform, the $\{\ket{\theta_k}\}$ form an orthonormal basis. One may define accordingly a ``time operator'' $\hat{t}_d$:
\begin{align}
    \hat{t}_d & := \sum_{-\frac{d - 1}{2} \leq k \leq \frac{d - 1}{2}}k\ket{\theta_k}\bra{\theta_k}
\end{align}
for which $\ket{\theta_k}$ is an eigenstate of eigenvalue $k$. Actually, in the following, it will be convenient to generalize the notation $\ket{\theta_k}$ (defined for the moment for $k$ an integer in $\left[-\frac{d - 1}{2}, \frac{d - 1}{2}\right]$) to $\ket{\theta_t}$ where $t \in \mathbf{R}$. Let us then simply define:
\begin{align}
    \ket{\theta_t} & := \frac{1}{\sqrt{d}}\sum_{-\frac{d - 1}{2} \leq n \leq \frac{d - 1}{2}}e^{-\frac{2\pi int}{d}}\ket{n}, \qquad t \in \mathbf{R}\,.
\end{align}
Note that it remains true for all $t \in \mathbf{R}$ and all sequences $(k_j)_{-\frac{d - 1}{2} \leq j \leq \frac{d - 1}{2}}$ modulo $d$,  that $\{\ket{\theta_{t + k_j}}\}_{-\frac{d - 1}{2} \leq j \leq \frac{d - 1}{2}}$ forms an orthonormal set.

We now describe the states on which we will study the dynamics of the clock. Here, we still follow \cite{woods_autonomous_2019} in essence, namely we use a Gaussian superposition of time eigenstates. However, instead of using actually Gaussian weights, we resort to Jacobi $\theta$ functions which, in essence, are periodized Gaussians. These Jacobi $\theta$ functions are more adapted to periodic systems and the use of the discrete Fourier transform. Some important properties of these functions are spelt out in section \ref{sec:jacobi_theta_functions}. We will consider an initial state $\ket{\Psi_0}$ of the form:
\begin{align}
\label{eq:quasi_ideal_states}
    \ket{\Psi_0} & = \sum_{-\frac{d - 1}{2} \leq k \leq \frac{d - 1}{2}}\psi(k)\ket{k}\,,\\
    \psi(k) & = \frac{\sqrt{2/d}}{\sqrt{\theta_3\left(0, \frac{i\xi^2}{2d}\right)\theta_3\left(0, \frac{id\xi^2}{2}\right) + \theta_3\left(\frac{1}{2}, \frac{i\xi^2}{2d}\right)\theta_3\left(\frac{d}{2}, \frac{id\xi^2}{2}\right)}}\theta_3\left(\frac{k}{d}, \frac{i\xi^2}{d}\right)\exp\left(\frac{2\pi in_0k}{d}\right)\\
    & =: \mathcal{N}\theta_3\left(\frac{k}{d}, \frac{i\xi^2}{d}\right)\exp\left(\frac{2\pi in_0k}{d}\right)\,,
\end{align}
where $n_0$ is an integer and $\xi > 0$ parametrizes the width of the state. This width will scale with $d$:
\begin{align}
\label{eq:quasi_ideal_state_scaling_xi}
    \xi^2 & := d^{\beta}\,, \qquad -1 < \beta < 1\,.
\end{align}
We also introduce --- following the notations of the aforementioned papers --- a parameter $\alpha_0 \in [0, 1]$ which measures the distance of $n_0$ from its extreme possible values $-\frac{d - 1}{2}$ and $\frac{d - 1}{2}$:
\begin{align}
\label{eq:definition_alpha_0}
    \alpha_0 & := 1 - \left|1 - n_0\left(\frac{2}{d - 1}\right)\right|
\end{align}
This state is normalized according to the results of section \ref{sec:jacobi_theta_functions}. Furthermore, the discrete Fourier transform $\widetilde{\psi}$ of $\psi$ is given by:
\begin{align}
\label{eq:quasi_ideal_states_ft}
    \widetilde{\psi}(p) & = \frac{1}{\sqrt{d}}\sum_{-\frac{d - 1}{2} \leq k \leq \frac{d - 1}{2}}e^{-\frac{2\pi ipk}{d}}\psi(k)\nonumber\\
    & = \frac{1}{\xi}\frac{\sqrt{2/d}}{\sqrt{\theta_3\left(0, \frac{i\xi^2}{2d}\right)\theta_3\left(0, \frac{id\xi^2}{2}\right) + \theta_3\left(\frac{1}{2}, \frac{i\xi^2}{2d}\right)\theta_3\left(\frac{d}{2}, \frac{id\xi^2}{2}\right)}}\theta_3\left(\frac{p - n_0}{d}, \frac{i}{\xi^2d}\right)\\
    & = \frac{\sqrt{2/d}}{\sqrt{\theta_3\left(0, \frac{i}{2\xi^2d}\right)\theta_3\left(0, \frac{id}{2\xi^2}\right) + \theta_3\left(\frac{1}{2}, \frac{i}{2\xi^2d}\right)\theta_3\left(\frac{d}{2}, \frac{id}{2\xi^2}\right)}}\theta_3\left(\frac{p - n_0}{d}, \frac{i}{\xi^2d}\right)\\
    & =: \mathcal{N}'\theta_3\left(\frac{p - n_0}{d}, \frac{i}{\xi^2d}\right)\,.
\end{align}
Using one of the two expressions of $\mathcal{N}'$ stated above, recalling that both $\frac{d}{\xi^2}$ and $\xi^2d$ go to infinity (following some power law in $d$) given the scaling for $\xi^2$ we chose and using the transformation property \ref{eq:theta3_modular_transformation}, one obtains:
\begin{align}
\label{eq:quasi_ideal_states_scaling_nprime}
    \mathcal{N}' & \sim \frac{2^{1/4}}{d^{3/4 + \beta/4}} = \frac{2^{1/4}}{d^{3/4}\xi^{1/2}}\,, \qquad d \to \infty\,.
\end{align}

The following lemma, combined with Poisson's summation formula, will be very useful to derive various estimates concerning the quasi-ideal clock.
\begin{lemma}
\label{lemma:truncation_theta_sum}
    Let $f$ denote a complex-valued function defined on $\mathbf{R}$ such that there exists $c, \beta \geq 0$ for which:
    \begin{align*}
        |f(x)| & \leq c|x|^{\beta}\,, \qquad x \in \mathbf{R}\,.
    \end{align*}
    Let $z \in \left(-\frac{1}{2}, \frac{1}{2}\right), \xi > 0$ and $d \in \mathbf{N}$ odd. Then the following estimate holds:
    \begin{align*}
        \sum_{-\frac{d - 1}{2} \leq k \leq \frac{d - 1}{2}}f(k)\theta_3\left(z + \frac{k}{d}, \frac{i\xi^2}{d}\right) & = \sqrt{\frac{\xi^2}{d}}\sum_{k \in \mathbf{Z}}f(k)\exp\left(-\frac{\pi}{\xi^2d}(dz + k)^2\right) + \varepsilon\,,\\
        |\varepsilon| & \leq 2^{2 - \beta}c\sqrt{\frac{\xi^2}{d}}e^{\frac{\beta^2\xi^2}{\pi d}}(d + 1)^{\beta}\frac{e^{-\frac{\pi d}{\xi^2}\left(\frac{1}{2} + \frac{1}{2d} - |z|\right)^2}}{1 - e^{-\frac{2\pi}{\xi^2}\left(\frac{1}{2} + \frac{1}{2d} - |z|\right)}}\,.
    \end{align*}
\begin{proof}
    \begin{align}
        & \sum_{-\frac{d - 1}{2} \leq k \leq \frac{d - 1}{2}}f(k)\theta_3\left(z + \frac{k}{d}, \frac{i\xi^2}{d}\right)\nonumber\\
        & = \sum_{-\frac{d - 1}{2} \leq k \leq \frac{d - 1}{2}}f(k)\sqrt{\frac{\xi^2}{d}}\sum_{p \in \mathbf{Z}}\exp\left(-\frac{\pi d}{\xi^2}\left(z + \frac{k}{d} + p\right)^2\right)\\
        %
        %
        & = \sqrt{\frac{\xi^2}{d}}\sum_{\substack{-\frac{d - 1}{2} \leq k \leq \frac{d - 1}{2}\\p \in \mathbf{Z}}}f(k + pd)\exp\left(-\frac{\pi d}{\xi^2}\left(z + \frac{k + pd}{d}\right)^2\right)\nonumber\\
        & \hspace{0.05\textwidth} + \sqrt{\frac{\xi^2}{d}}\sum_{\substack{-\frac{d - 1}{2} \leq k \leq \frac{d - 1}{2}\\p \in \mathbf{Z} - \{ 0 \}}}\left(f(k) - f(k + pd)\right)\exp\left(-\frac{\pi d}{\xi^2}\left(z + \frac{k + pd}{d}\right)^2\right)\\
        & = \sqrt{\frac{\xi^2}{d}}\sum_{m \in \mathbf{Z}}f(m)\exp\left(-\frac{\pi}{\xi^2d}\left(dz + m\right)^2\right)\nonumber\\
        & \hspace{0.05\textwidth} + \sqrt{\frac{\xi^2}{d}}\sum_{|m| \geq \frac{d + 1}{2}}\left(f\left(m - d\left\lfloor\frac{m + \frac{d - 1}{2}}{d}\right\rfloor\right) - f(m)\right)\exp\left(-\frac{\pi}{\xi^2d}(dz + m)^2\right)\,.\label{line:1st recorded}
    \end{align}
    Now, using the bound on $f$ as well as proposition \ref{prop:bound_tail_sum_gaussians}, one can upper bound the line \eqref{line:1st recorded} as 
    \begin{align}
    	& \sqrt{\frac{\xi^2}{d}}\sum_{|m| \geq \frac{d + 1}{2}}\left(f\left(m - d\left\lfloor\frac{m + \frac{d - 1}{2}}{d}\right\rfloor\right) - f(m)\right)\exp\left(-\frac{\pi}{\xi^2d}(dz + m)^2\right)\nonumber\\
	 & \leq 2c\sqrt{\frac{\xi^2}{d}}e^{\frac{\beta^2}{4\left(\frac{d + 1}{2}\right)^2\frac{\pi}{\xi^2d}}}\left[\left(\frac{d + 1}{2}\right)^{\beta}\frac{e^{-\frac{\pi}{\xi^2d}\left(dz + \frac{d + 1}{2}\right)^2}}{1 - e^{-\frac{2\pi}{\xi^2d}\left(dz + \frac{d + 1}{2}\right)}} + \left(\frac{d + 1}{2}\right)^{\beta}\frac{e^{-\frac{\pi}{\xi^2d}\left(-dz + \frac{d + 1}{2}\right)^2}}{1 - e^{-\frac{2\pi}{\xi^2d}\left(-dz + \frac{d + 1}{2}\right)}}\right]\\
        & \leq 2^{2 - \beta}c\sqrt{\frac{\xi^2}{d}}e^{\frac{\beta^2\xi^2}{\pi d}}(d + 1)^{\beta}\frac{e^{-\frac{\pi d}{\xi^2}\left(\frac{1}{2} + \frac{1}{2d} - |z|\right)^2}}{1 - e^{-\frac{2\pi}{\xi^2}\left(\frac{1}{2} + \frac{1}{2d} - |z|\right)}}\,.
    \end{align}
\end{proof}
\end{lemma}
Following similar steps to the proof of lemma \ref{lemma:truncation_theta_sum}, one can prove:
\begin{lemma}
\label{lemma:truncation_theta_prime_sum}
    Let $f$ denote a complex-valued function defined on $\mathbf{R}$ such that there exists $c, \beta \geq 0$ for which:
    \begin{align*}
        |f(x)| & \leq c|x|^{\beta}\,, \qquad x \in \mathbf{R}\,.
    \end{align*}
    Let $z \in \left(-\frac{1}{2}, \frac{1}{2}\right), \xi > 0$ and $d \in \mathbf{N}$ odd. Then the following estimate holds:
    \begin{align}
        \sum_{-\frac{d - 1}{2} \leq k \leq \frac{d - 1}{2}}f(k)\theta_3'\left(z + \frac{k}{d}, \frac{i\xi^2}{d}\right) & = -\frac{2\pi}{\sqrt{\xi^2d}}\sum_{k \in \mathbf{Z}}f(k)(dz + k)\exp\left(-\frac{\pi}{\xi^2d}(dz + k)^2\right) + \varepsilon\,,\\
        |\varepsilon| & \leq 2^{3 - \beta}\pi(|z| + 1)c\sqrt{\frac{1}{\xi^2d}}e^{\frac{\beta^2\xi^2}{\pi d}}(d + 1)^{\beta + 1}\frac{e^{-\frac{\pi d}{\xi^2}\left(\frac{1}{2} + \frac{1}{2d} - |z|\right)^2}}{1 - e^{-\frac{2\pi}{\xi^2}\left(\frac{1}{2} + \frac{1}{2d} - |z|\right)}}\,.
    \end{align}
\end{lemma}

\subsection{Approximation of the QND condition by the quasi-ideal clock}
In this section, we will be concerned with how well the time operator $\Hat{t}_d$ of the quasi-ideal clock approximates the QND condition with respect to the Hamiltonian $\Hat{H}_d$. More precisely, we will derive a bound for how close $[\Hat{t}_d(t_1), \Hat{t}_d(t_0)]\ket{\Psi_0}$ is to $0$ given two times $t_0, t_1$ and some quasi-ideal state $\ket{\Psi_0}$.

\subsubsection{With a ``linear'' time operator}
We now want to show that the commutator $[\hat{t}_d(t_1), \hat{t}_d(t_0)]$ (where $\hat{t}_d(t) := e^{i\hat{H}_dt}\hat{t}_de^{-i\hat{H}_dt}$), when applied to a Gaussian state, gives a vector whose magnitude decays exponentially with the dimension. It will be convenient to introduce parameters $\eta_0, \eta_1$ for $t_0, t_1$ which measure their distance to the points $\pm \frac{d - 1}{2}$, playing the same role as $\alpha_0$ with respect to $n_0$ (equation \ref{eq:definition_alpha_0}):
\begin{align}
    \eta_0 & := 1 - \frac{2}{d - 1}|t_0|\,,\\
    \eta_1 & := 1 - \frac{2}{d - 1}|t_1|\,.
\end{align}
From now on, we assume $|t_0|, |t_1| < \frac{d - 1}{2}$ and hence $0 \leq \eta_0, \eta_1 < 1$. Namely, we will show:
\begin{proposition}
    For all $-\frac{d - 1}{2} \leq k \leq \frac{d - 1}{2}$,
    \begin{align}
        \left|\braket{\theta_k|[\Hat{t}_d(t_1), \Hat{t}_d(t_0)]|\Psi_0}\right| & = \mathcal{O}\left(d^{\left(3 - \frac{\beta}{2}\right) \vee \left(\frac{1}{4} + \frac{\beta}{4}\right)}\frac{e^{-\frac{\pi d^{1 + \beta}}{4}(1 - \alpha_0)^2}}{1 - e^{-\pi d^{\beta}(1 - \alpha_0)}} + d^{\frac{3}{4} - \frac{\beta}{4}}\frac{e^{-\frac{\pi d^{1 - \beta}}{4}\left[(1 - \eta_0 \vee \eta_1)^2\right]}}{1 - e^{-\pi d^{-\beta}\left[(1 - \eta_0 \vee \eta_1)\right]}}\right)
    \end{align}
\end{proposition}
\begin{proof}
To start with, we first express this commutator in the time eigenbasis:
\begin{align}
    & [\hat{t}_d(t_1), \hat{t}_d(t_0)]\nonumber\\
    & = \sum_{-\frac{d - 1}{2} \leq k_0, k_1 \leq \frac{d - 1}{2}}k_0k_1\left(\ket{\theta_{k_1 - t_1}}\braket{\theta_{k_1 - t_1}|\theta_{k_0 - t_0}}\bra{\theta_{k_0 - t_0}} - \ket{\theta_{k_0 - t_0}}\braket{\theta_{k_0 - t_0}|\theta_{k_1 - t_1}}\bra{\theta_{k_1 - t_1}}\right)\\
    & = \sum_{\substack{-\frac{d - 1}{2} \leq k_0, k_1 \leq \frac{d - 1}{2}\\-\frac{d - 1}{2} \leq k, l \leq \frac{d - 1}{2}}}k_0k_1\left(\braket{\theta_k|\theta_{k_1 - t_1}}\braket{\theta_{k_1 - t_1}|\theta_{k_0 - t_0}}\braket{\theta_{k_0 - t_0}|\theta_l}\right.\nonumber\\
    & \left. \hspace{0.25\textwidth} - \braket{\theta_k|\theta_{k_0 - t_0}}\braket{\theta_{k_0 - t_0}|\theta_{k_1 - t_1}}\braket{\theta_{k_1 - t_1}|\theta_l}\right)\ket{\theta_k}\bra{\theta_l}\\
    & = \frac{1}{d^3}\sum_{\substack{-\frac{d - 1}{2} \leq k_0, k_1 \leq \frac{d - 1}{2}\\-\frac{d - 1}{2} \leq k, l \leq \frac{d - 1}{2}\\-\frac{d - 1}{2} \leq p, q, r \leq \frac{d - 1}{2}}}k_0k_1\left[\exp\left(\frac{2\pi i}{d}\left(p(k - k_1 + t_1) + q(k_1 - k_0 - t_1 + t_0) + r(k_0 - l - t_0)\right)\right)\right.\nonumber\\
    & \left. \hspace{0.15\textwidth} - \exp\left(\frac{2\pi i}{d}\left(p(k - k_0 + t_0) + q(k_0 - k_1 + t_1 - t_0) + r(k_1 - l - t_1)\right)\right)\right]\ket{\theta_k}\bra{\theta_l}\,.
\end{align}
It follows:
\begin{align}
    & \braket{\theta_k|[\hat{t}_d(t_1), \hat{t}_d(t_0)]|\Psi_0}\nonumber\\
    & = \frac{1}{d^{5/2}}\sum_{\substack{-\frac{d - 1}{2} \leq k_0, k_1 \leq \frac{d - 1}{2}\\-\frac{d - 1}{2} \leq p, q, r \leq \frac{d - 1}{2}}}k_0k_1\left[\exp\left(\frac{2\pi i}{d}\left(p(k - k_1 + t_1) + q(k_1 - k_0 - t_1 + t_0) + r(k_0 - t_0)\right)\right)\right.\nonumber\\
    & \left. \hspace{0.15\textwidth} - \exp\left(\frac{2\pi i}{d}\left(p(k - k_0 + t_0) + q(k_0 - k_1 + t_1 - t_0) + r(k_1 - t_1)\right)\right)\right]\widetilde{\psi}(r)\,.
\end{align}
Let us focus on estimating the first term above, 
\begin{align}
    S & := \frac{1}{d^{5/2}}\sum_{\substack{-\frac{d - 1}{2} \leq k_0, k_1 \leq \frac{d - 1}{2}\\-\frac{d - 1}{2} \leq p, q, r \leq \frac{d - 1}{2}}}k_0k_1\exp\left(\frac{2\pi i}{d}\left(p(k - k_1 + t_1) + q(k_1 - k_0 - t_1 + t_0) + r(k_0 - t_0)\right)\right)\,\widetilde{\psi}(r)\,.
\end{align}
By lemma \ref{lemma:truncation_theta_sum} (applied with $c = 1, \beta = 0, z = -\frac{n_0}{d}$), one may first perform the summation in $r$:
\begin{align}
    & \sum_{-\frac{d - 1}{2} \leq r \leq \frac{d - 1}{2}}\exp\left(\frac{2\pi ir(k_0 - t_0)}{d}\right)\widetilde{\psi}(r)\nonumber\\
    & = \mathcal{N}'\sum_{-\frac{d - 1}{2} \leq r \leq \frac{d - 1}{2}}\exp\left(\frac{2\pi ir(k_0 - t_0)}{d}\right)\theta_3\left(\frac{r - n_0}{d}, \frac{i}{\xi^2d}\right)\\
    & = \mathcal{N}'\sqrt{\frac{1}{\xi^2d}}\sum_{m \in \mathbf{Z}}\exp\left(\frac{2\pi im(k_0 - t_0)}{d}\right)\exp\left(-\frac{\pi\xi^2}{d}(m - n_0)^2\right) + \varepsilon_1\,,
\end{align}
where
\begin{align}
    |\varepsilon_1| & = \mathcal{O}\left(\sqrt{\frac{1}{\xi^2d}}\frac{e^{-\frac{\pi\xi^2 d}{4}(1 - \alpha_0)^2}}{1 - e^{-\pi\xi^2(1 - \alpha_0)}}\right)\,.
\end{align}
The leading above term can be cast to a $\theta$ function:
\begin{align}
    & \mathcal{N}'\sqrt{\frac{1}{\xi^2d}}\sum_{m \in \mathbf{Z}}\exp\left(\frac{2\pi im(k_0 - t_0)}{d}\right)\exp\left(-\frac{\pi\xi^2}{d}(m - n_0)^2\right)\nonumber\\
    & = \mathcal{N}'\sqrt{\frac{1}{\xi^2d}}\exp\left(\frac{2\pi in_0(k_0 - t_0)}{d}\right)\sum_{m \in \mathbf{Z}}\exp\left(\frac{2\pi im(k_0 - t_0)}{d}\right)\exp\left(-\frac{\pi\xi^2}{d}m^2\right)\\
    & = \mathcal{N}'\sqrt{\frac{1}{\xi^2d}}\exp\left(\frac{2\pi in_0(k_0 - t_0)}{d}\right)\theta_3\left(\frac{k_0}{d} - \frac{t_0}{d}, \frac{i\xi^2}{d}\right) \,.
\end{align}
Therefore
\begin{align}
    S & = \frac{1}{d^{5/2}}\sum_{\substack{-\frac{d - 1}{2} \leq p, q \leq \frac{d - 1}{2}\\-\frac{d - 1}{2} \leq k_0, k_1 \leq \frac{d - 1}{2}}}k_0k_1\exp\left(\frac{2\pi i}{d}\left(p(k - k_1 + t_1) + q(k_1 - k_0 - t_1 + t_0)\right)\right)\nonumber\\
    & \hspace{0.15\textwidth} \times \mathcal{N}'\sqrt{\frac{1}{\xi^2d}}\exp\left(\frac{2\pi in_0(k_0 - t_0)}{d}\right)\theta_3\left(\frac{k_0}{d} - \frac{t_0}{d}, \frac{i\xi^2}{d}\right) + \varepsilon_1'\,,\\
    \varepsilon_1' & = \mathcal{O}\left(d^{7/2}\sqrt{\frac{1}{\xi^2d}}\frac{e^{-\frac{\pi\xi^2d}{4}(1 - \alpha_0)^2}}{1 - e^{-\pi\xi^2(1 - \alpha_0)}}\right)\,.
\end{align}
We will now use lemma \ref{lemma:truncation_theta_sum} again to perform the summation over $k_0$:
\begin{align}
    & \sum_{-\frac{d - 1}{2} \leq k_0 \leq \frac{d - 1}{2}}k_0\exp\left(\frac{2\pi ik_0(n_0 - q)}{d}\right)\theta_3\left(\frac{k_0}{d} - \frac{t_0}{d}, \frac{i\xi^2}{d}\right)\nonumber\\
    & = \sqrt{\frac{\xi^2}{d}}\sum_{m \in \mathbf{Z}}m\exp\left(\frac{2\pi im(n_0 - q)}{d}\right)\exp\left(-\frac{\pi}{\xi^2d}(m - t_0)^2\right) + \varepsilon_2\,,
\end{align}
where
\begin{align}
    |\varepsilon_2| & = \mathcal{O}\left(\sqrt{\frac{\xi^2}{d}}d\frac{e^{-\frac{\pi d}{4\xi^2}(1 - \eta_0)^2}}{1 - e^{-\frac{\pi}{\xi^2}(1 - \eta_0)}}\right)\,.
\end{align}
The leading term above can be written as the derivative of a $\theta$ function:
\begin{align}
    & \sqrt{\frac{\xi^2}{d}}\sum_{m \in \mathbf{Z}}m\exp\left(\frac{2\pi im(n_0 - q)}{d}\right)\exp\left(-\frac{\pi}{\xi^2d}(m - t_0)^2\right)\nonumber\\
    & = \sqrt{\frac{\xi^2}{d}}\exp\left(\frac{2\pi it_0(n_0 - q)}{d}\right)\sum_{m \in \mathbf{Z}}m\exp\left(\frac{2\pi im(n_0 - q)}{d}\right)\exp\left(-\frac{\pi}{\xi^2d}m^2\right)\\
    & = -\frac{id}{2\pi}\sqrt{\frac{\xi^2}{d}}\exp\left(\frac{2\pi it_0(n_0 - q)}{d}\right)\theta_3'\left(\frac{q}{d} - \frac{n_0}{d}, \frac{i}{\xi^2d}\right)\,.
\end{align}
Therefore, we have established 
\begin{align}
    S & = \frac{1}{d^{5/2}}\sum_{\substack{-\frac{d - 1}{2} \leq k_1 \leq \frac{d - 1}{2}\\-\frac{d - 1}{2} \leq p, q \leq \frac{d - 1}{2}}}k_1\exp\left(\frac{2\pi i}{d}\left(p(k - k_1 + t_1) + q(k_1 - t_1)\right)\right)\frac{-i\mathcal{N}'}{2\pi}\theta_3'\left(\frac{q}{d} - \frac{n_0}{d}, \frac{i}{\xi^2d}\right)\nonumber\\
    & \hspace{0.05\textwidth} + \varepsilon_1' + \varepsilon_2'\,,
\end{align}
with
\begin{align}
    |\varepsilon_2'| & \leq \mathcal{O}\left(\mathcal{N}'d^{3/2}\frac{e^{-\frac{\pi d}{4\xi^2}(1 - \eta_0)^2}}{1 - e^{-\frac{\pi}{\xi^2}(1 - \eta_0)}}\right)\,.\\\nonumber
\end{align}
One may now perform the summation in $q$ using lemma \ref{lemma:truncation_theta_prime_sum}:
\begin{align}
    & \sum_{-\frac{d - 1}{2} \leq q \leq \frac{d - 1}{2}}\exp\left(\frac{2\pi iq(k_1 - t_1)}{d}\right)\theta_3'\left(\frac{q}{d} - \frac{n_0}{d}, \frac{i}{\xi^2d}\right)\nonumber\\
    & = -2\pi\sqrt{\frac{\xi^2}{d}}\sum_{m \in \mathbf{Z}}\exp\left(\frac{2\pi im(k_1 - t_1)}{d}\right)(m - n_0)\exp\left(-\frac{\pi\xi^2}{d}(m - n_0)^2\right) + \varepsilon_3\,,
\end{align}
where
\begin{align}
    |\varepsilon_3| & = \mathcal{O}\left(\sqrt{\frac{\xi^2}{d}}d\frac{e^{-\frac{\pi\xi^2d}{4}(1 - \alpha_0)^2}}{1 - e^{-\pi\xi^2(1 - \alpha_0)}}\right)\,.
\end{align}
The leading above term can again be written as the derivative of a $\theta$ function:
\begin{align}
    & -2\pi\sqrt{\frac{\xi^2}{d}}\sum_{m \in \mathbf{Z}}\exp\left(\frac{2\pi im(k_1 - t_1)}{d}\right)(m - n_0)\exp\left(-\frac{\pi\xi^2}{d}(m - n_0)^2\right)\nonumber\\
    & = -2\pi\sqrt{\frac{\xi^2}{d}}\exp\left(\frac{2\pi in_0(k_1 - t_1)}{d}\right)\sum_{m \in \mathbf{Z}}\exp\left(\frac{2\pi im(k_1 - t_1)}{d}\right)m\exp\left(-\frac{\pi\xi^2}{d}m^2\right)\\
    & = id\sqrt{\frac{\xi^2}{d}}\exp\left(\frac{2\pi in_0(k_1 - t_1)}{d}\right)\theta_3'\left(\frac{k_1}{d} - \frac{t_1}{d}, \frac{i\xi^2}{d}\right)\,.
\end{align}
Therefore
\begin{align}
    S & = \frac{1}{d^{5/2}}\sum_{\substack{-\frac{d - 1}{2} \leq k_1 \leq \frac{d - 1}{2}\\-\frac{d - 1}{2} \leq p \leq \frac{d - 1}{2}}}k_1\exp\left(\frac{2\pi i}{d}\left(p(k - k_1 + t_1) + n_0(k_1 - t_1)\right)\right)\frac{d}{2\pi}\mathcal{N}'\sqrt{\frac{\xi^2}{d}}\theta_3'\left(\frac{k_1}{d} - \frac{t_1}{d}, \frac{i\xi^2}{d}\right)\nonumber\\
    & \hspace{0.05\textwidth} + \varepsilon_1' + \varepsilon_2' + \varepsilon_3'\,,
\end{align}
with
\begin{align}
    |\varepsilon_3'| & = \mathcal{O}\left(\mathcal{N}'\sqrt{\frac{\xi^2}{d}}d^{3/2}\frac{e^{-\frac{\pi\xi^2d}{4}(1 - \alpha_0)^2}}{1 - e^{-\pi\xi^2(1 - \alpha_0)}}\right)\,.
\end{align}
One now carries out the summation over $k_1$:
\begin{align}
    & \sum_{-\frac{d - 1}{2} \leq k_1 \leq \frac{d - 1}{2}}k_1\exp\left(\frac{2\pi ik_1(n_0 - p)}{d}\right)\theta_3'\left(\frac{k_1}{d} - \frac{t_1}{d}, \frac{i\xi^2}{d}\right)\nonumber\\
    & = -\frac{2\pi}{\sqrt{\xi^2d}}\sum_{m \in \mathbf{Z}}m(m - t_1)\exp\left(\frac{2\pi im(n_0 - p)}{d}\right)\exp\left(-\frac{\pi}{\xi^2d}(m - t_1)^2\right) + \varepsilon_4\,,
\end{align}
where
\begin{align}
    |\varepsilon_4| & = \mathcal{O}\left(\frac{1}{\sqrt{\xi^2d}}d^2\frac{e^{-\frac{\pi d}{4\xi^2}(1 - \eta_1)^2}}{1 - e^{-\frac{\pi}{\xi^2}(1 - \eta_1)}}\right)\,.
\end{align}
The leading terms can be rewritten as
\begin{align}
    & -\frac{2\pi}{\sqrt{\xi^2d}}\sum_{m \in \mathbf{Z}}m(m - t_1)\exp\left(\frac{2\pi im(n_0 - p)}{d}\right)\exp\left(-\frac{\pi}{\xi^2d}(m - t_1)^2\right)\nonumber\\
    & = -\frac{2\pi}{\sqrt{\xi^2d}}\exp\left(\frac{2\pi it_1(n_0 - p)}{d}\right)\sum_{m \in \mathbf{Z}}m^2\exp\left(\frac{2\pi im(n_0 - p)}{d}\right)\exp\left(-\frac{\pi}{\xi^2d}m^2\right)\\
    & = \frac{d^2}{2\pi\sqrt{\xi^2d}}\exp\left(\frac{2\pi it_1(n_0 - p)}{d}\right)\theta_3''\left(\frac{p}{d} - \frac{n_0}{d}, \frac{i}{\xi^2d}\right)\,.
\end{align}
Therefore
\begin{align}
    S & = \frac{1}{d^{5/2}}\sum_{-\frac{d - 1}{2} \leq p \leq \frac{d - 1}{2}}\exp\left(\frac{2\pi ipk}{d}\right)\mathcal{N}'\frac{d^2}{4\pi^2}\theta_3''\left(\frac{p}{d} - \frac{n_0}{d}, \frac{i}{\xi^2d}\right)\\
    & \hspace{0.05\textwidth} + \varepsilon_1' + \varepsilon_2' + \varepsilon_3' + \varepsilon_4'\,,
\end{align}
with
\begin{align}
    |\varepsilon_4'| & = \mathcal{O}\left(\mathcal{N}'d^{1/2}\frac{e^{-\frac{\pi d}{4\xi^2}(1 - \eta_1)^2}}{1 - e^{-\frac{\pi}{\xi^2}(1 - \eta_1)}}\right)\,.
\end{align}
All in all, after recalling the scalings for $\xi^2$ and $\mathcal{N}'$ we enforced or established, one obtains:
\begin{align}
    S & = \frac{1}{d^{5/2}}\sum_{-\frac{d - 1}{2} \leq p \leq \frac{d - 1}{2}}\exp\left(\frac{2\pi ipk}{d}\right)\mathcal{N}'\frac{d^2}{4\pi^2}\theta_3''\left(\frac{p}{d} - \frac{n_0}{d}, \frac{i}{\xi^2d}\right)\nonumber\\
    & \hspace{0.05\textwidth} + \varepsilon_1' + \varepsilon_2' + \varepsilon_3' + \varepsilon_4'\,,\\
    |\varepsilon_1'| & = \mathcal{O}\left(d^{3 - \frac{\beta}{2}}\frac{e^{-\frac{\pi d^{1 + \beta}}{4}(1 - \alpha_0)^2}}{1 - e^{-\pi d^{\beta}(1 - \alpha_0)}}\right)\,,\\
    |\varepsilon_2'| & = \mathcal{O}\left(d^{\frac{3}{4} - \frac{\beta}{4}}\frac{e^{-\frac{\pi d^{1 - \beta}}{4}(1 - \eta_0)^2}}{1 - e^{-\pi d^{-\beta}(1 - \eta_0)}}\right)\,,\\
    |\varepsilon_3'| & = \mathcal{O}\left(d^{\frac{1}{4} + \frac{\beta}{4}}\frac{e^{-\frac{\pi d^{1 + \beta}}{4}(1 - \alpha_0)^2}}{1 - e^{-\pi d^{\beta}(1 - \alpha_0)}}\right)\,,\\
    |\varepsilon_4'| & = \mathcal{O}\left(d^{-\frac{1}{4} - \frac{\beta}{4}}\frac{e^{-\frac{\pi d^{1 - \beta}}{4}(1 - \eta_1)^2}}{1 - e^{-\pi d^{-\beta}(1 - \eta_1)}}\right)\,.
\end{align}
Note that the leading term does not depend at all on $t_0, t_1$, so in particular not on their ordering. (However, $t_0, t_1$ do contribute to the errors; this essentially says that for the latter to be under control, the wavefunction should not have moved too close to the ``boundary times'' $\pm\frac{d - 1}{2}$ at times $t_0, t_1$.) Therefore, we have indeed established that $\braket{\theta_k|[\hat{t}_d(t_1), \hat{t}_d(t_0)]|\Psi_0}$ vanishes up to exponential errors.
\end{proof}

\subsubsection{With a ``periodic'' time operator}
In this section, we will essentially repeat the calculation we have just performed in the previous subsection, except that we will replace the time operator by a $d$-periodized version. More precisely, given fix integers $m, n$ $\left(-\frac{d - 1}{2} \leq m, n \leq \frac{d - 1}{2}\right)$, we will estimate
\begin{equation}
    \left[\exp\left(\frac{2\pi in\hat{t}_d(t_1)}{d}\right), \exp\left(\frac{2\pi im\hat{t}_d(t_0)}{d}\right)\right]\,.
\end{equation}
Precisely, we will prove:
\begin{proposition}
For all $-\frac{d - 1}{2} \leq k \leq \frac{d - 1}{2}$,
\begin{align}
    \left|\braket{\theta_k|\left[\exp\left(\frac{2\pi in\Hat{t}_d(t_1)}{d}\right), \exp\left(\frac{2\pi im\Hat{t}_d(t_0)}{d}\right)\right]|\Psi_0}{}\right| & = \mathcal{O}\left(\xi\mathcal{N}'\frac{e^{-\frac{\pi\xi^2}{d}\left(\frac{d + 1}{2} - |m + n| \vee |m| - n_0\right)^2}}{1 - e^{-\frac{2\pi\xi^2}{d}\left(\frac{d + 1}{2} - |m + n| \vee |m| - n_0\right)}}\right)
\end{align}
\end{proposition}
\begin{proof}
One has:
\begin{align}
    & \exp\left(\frac{2\pi in\hat{t}_d(t_1)}{d}\right)\exp\left(\frac{2\pi im\hat{t}_d(t_0)}{d}\right)\nonumber\\
    & = \frac{1}{d^3}\sum_{\substack{-\frac{d - 1}{2} \leq k_0, k_1 \leq \frac{d - 1}{2}\\-\frac{d - 1}{2} \leq k, l \leq \frac{d - 1}{2}\\-\frac{d - 1}{2} \leq p, q, r \leq \frac{d - 1}{2}}}\exp\left(\frac{2\pi i}{d}\left(nk_1 + mk_0 + p(k - k_1 + t_1) + q(k_1 - k_0 - t_1 + t_0)\right.\right.\nonumber\\
    & \left.\left. \hspace{0.3\textwidth} + r(k_0 - l - t_0)\right)\right)\ket{\theta_k}\bra{\theta_l}\\
    & = \frac{1}{d}\sum_{\substack{-\frac{d - 1}{2} \leq k, l \leq \frac{d - 1}{2}\\-\frac{d - 1}{2} \leq r \leq \frac{d - 1}{2}\\p := n + q\,[d]\\q := m + r\,[d]}}\exp\left(\frac{2\pi i}{d}\left(p(k + t_1) + q(t_0 - t_1) - r(l + t_0)\right)\right)\ket{\theta_k}\bra{\theta_l}\\
    & = \frac{1}{d}\sum_{\substack{-\frac{d - 1}{2} \leq k, l \leq \frac{d - 1}{2}\\-\frac{d - 1}{2} \leq r \leq \frac{d - 1}{2}\\p := m + n + r\,[d]\\q := m + r\,[d]}}\exp\left(\frac{2\pi i}{d}\left(p(k + t_1) + q(t_0 - t_1) - r(l + t_0)\right)\right)\ket{\theta_k}\bra{\theta_l}\,.
\end{align}
Applying this operator to the initial state $\ket{\Psi_0}$ and projecting onto $\ket{\theta_k}$:
\begin{align}
    & \braket{\theta_k|\exp\left(\frac{2\pi in\hat{t}_d(t_1)}{d}\right)\exp\left(\frac{2\pi im\hat{t}_d(t_0)}{d}\right)|\Psi_0}\nonumber\\
    & = \frac{1}{\sqrt{d}}\sum_{\substack{-\frac{d - 1}{2} \leq k, l \leq \frac{d - 1}{2}\\-\frac{d - 1}{2} \leq r \leq \frac{d - 1}{2}\\p := m + n + r\,[d]\\q := m + r\,[d]}}\exp\left(\frac{2\pi i}{d}\left(p(k + t_1) + q(t_0 - t_1) - rt_0\right)\right)\widetilde{\psi}(r)\,.
\end{align}
Since the times $t_1, t_0$ are not necessarily integers, the summand is a priori not invariant in $p, q$ modulo $d$. Therefore, one must distinguish between the case where $|m + n + r|, |m + r| \leq \frac{d  - 1}{2}$, and the case where at least one of these conditions is violated. One can write:
\begin{align}
    & \frac{1}{\sqrt{d}}\sum_{\substack{-\frac{d - 1}{2} \leq r \leq \frac{d - 1}{2}\\p := m + n + r\,[d]\\q := m + r\,[d]}}\exp\left(\frac{2\pi i}{d}\left(p(k + t_1) + q(t_0 - t_1) - r(l + t_0)\right)\right)\widetilde{\psi}(r)\nonumber\\
    & = \frac{1}{\sqrt{d}}\sum_{\substack{-\frac{d - 1}{2} \leq r \leq \frac{d - 1}{2}\\p = m + n + r\\q = m + r}}\exp\left(\frac{2\pi i}{d}\left(p(k + t_1) + q(t_0 - t_1) - rt_0\right)\right)\widetilde{\psi}(r) + \varepsilon_1\\
    & = \frac{1}{\sqrt{d}}\sum_{\substack{-\frac{d - 1}{2} \leq r \leq \frac{d - 1}{2}}}\exp\left(\frac{2\pi i}{d}\left(k(m + n + r) + nt_1 + mt_0\right)\right)\widetilde{\psi}(r) + \varepsilon_1\\
    & = \exp\left(\frac{2\pi i}{d}\left(n(k + t_1) + m(k + t_0)\right)\right)\psi(k) + \varepsilon_1\,,
\end{align}
where
\begin{align}
    |\varepsilon_1| & \leq \frac{2}{\sqrt{d}}\sum_{\left(\frac{d + 1}{2} - |m + n| \vee |m|\right) \leq |r| \leq \frac{d - 1}{2}}|\widetilde{\psi}(r)|\,,
\end{align}
which is indeed small provided $m, n$ are of order unity. More precisely, assume for simplicity $|m| \vee |m + n| \leq n_0 \leq \frac{d + 1}{2} - |m + n| \vee |m|$. Then:
\begin{align}
    |\varepsilon_1| & \leq \frac{2}{\sqrt{d}}\sum_{\left(\frac{d + 1}{2} - |m + n| \vee |m|\right) \leq |r| \leq \frac{d - 1}{2}}\mathcal{N'}\theta_3\left(\frac{r - n_0}{d}, \frac{i}{\xi^2d}\right)\\
    & = \frac{2}{\sqrt{d}}\sum_{\frac{d + 1}{2} - |m + n| \vee |m| \leq r \leq \frac{d - 1}{2}}\mathcal{N'}\theta_3\left(\frac{r - n_0}{d}, \frac{i}{\xi^2d}\right)\\
    & \hspace{0.05\textwidth} + \frac{2}{\sqrt{d}}\sum_{-\frac{d - 1}{2} \leq r \leq -\frac{d + 1}{2} + |m + n| \vee |m|}\mathcal{N'}\theta_3\left(\frac{r - n_0}{d}, \frac{i}{\xi^2d}\right)\\
    & = \frac{2}{\sqrt{d}}\mathcal{N}'\left[\sum_{\frac{d + 1}{2} - |m + n| \vee |m| - n_0 \leq k \leq \frac{d - 1}{2} - n_0}\theta_3\left(\frac{k}{d}, \frac{i}{\xi^2d}\right)\right.\nonumber\\
    & \left. \hspace{0.15\textwidth} + \sum_{\frac{d + 1}{2} - n_0 \leq k \leq \frac{d - 1}{2} + |m + n| \vee |m| - n_0}\theta_3\left(\frac{k}{d}, \frac{i}{\xi^2d}\right)\right]\,.
\end{align}
One then rewrites each $\theta_3\left(\frac{k}{d}, \frac{i}{\xi^2d}\right)$ in the form:
\begin{align}
    \theta_3\left(\frac{k}{d}, \frac{i}{\xi^2d}\right) & = \sqrt{\xi^2d}\exp\left(-\frac{\pi\xi^2}{d}k^2\right)\theta_3\left(i\xi^2k, i\xi^2d\right)\\
    & = \sqrt{\xi^2d}\exp\left(-\frac{\pi\xi^2}{d}k^2\right)\left(1 + \mathcal{O}\left(\frac{e^{-\pi\xi^2d\left(1 - \frac{2}{d}\left(\frac{d - 1}{2} + |m + n| \vee |m| - n_0\right)\right)}}{1 - e^{-2\pi\xi^2d}}\right)\right)\\
    & = \sqrt{\xi^2d}\exp\left(-\frac{\pi\xi^2}{d}k^2\right)\left(1 + \mathcal{O}\left(\frac{e^{-2\pi\xi^2\left(\frac{1}{2} + n_0 - |m + n| \vee |m|\right)}}{1 - e^{-2\pi\xi^2d}}\right)\right)\\
    & = \sqrt{\xi^2d}\exp\left(-\frac{\pi\xi^2}{d}k^2\right)\left(1 + \mathcal{O}\left(e^{-2\pi\xi^2}\right)\right)\,.
\end{align}
Finally, applying \ref{prop:bound_tail_sum_gaussians} to sum over $k$ yields:
\begin{align}
    |\varepsilon_1| & \leq 4\xi\mathcal{N'}\frac{e^{-\frac{\pi\xi^2}{d}\left(\frac{d + 1}{2} - |m + n| \vee |m| - n_0\right)^2}}{1 - e^{-\frac{2\pi\xi^2}{d}\left(\frac{d + 1}{2} - |m + n| \vee |m| - n_0\right)}}\mathcal{O}(1)\,.
\end{align}
\end{proof}

\subsection{Measured quasi-ideal clock}
\label{sec:measured_quasi_ideal_clock}
In this section, we will incorporate measurement into the quasi-ideal clock. As the quasi-ideal clock is a finite-dimensional system --- for which in particular one cannot implement exact canonical commutation relations --- the analysis of measurement will not lend itself to the methods in \cite[Chapter 5]{braginsky_quantum_1992}. Before moving to the general setting, we will pause to describe in detail a specific subcase which will serve as a reference for what follows afterwards.

\subsubsection{Measurement in the time basis}
\label{sec:quasi_ideal_clock_time_basis}
In this section, we consider the degenerate case where the initial state is a time eigenstate and one repeatedly measures the clock in this same basis. Therefore, after each measurement, the state of the clock collapses to a time eigenstate and the state of the system at any given time is completely described by the measurement results one has obtained up to this time.

Recall from section \ref{General setting} the following definitions of the Hamiltonian $\hat{H}_d$ and time operator $\hat{t}_d$  --- which will allow us to exhibit in a convenient form the scaling of the measurement statistics as $d \to \infty$:
\begin{align}
    \hat{H}_d & := \frac{2\pi}{\sqrt{d}}\sum_{-\frac{d - 1}{2} \leq n \leq \frac{d - 1}{2}}n\ket{n}\bra{n}\,,\\
    \hat{t}_d & := \frac{1}{\sqrt{d}}\sum_{-\frac{d - 1}{2} \leq k \leq \frac{d - 1}{2}}k\ket{\theta_k}\bra{\theta_k}\,.
\end{align}
Here the ``time'' operator has eigenvalues ranging from $-\frac{\sqrt{d}}{2}\left(1 - \frac{1}{d}\right)$ to $\frac{\sqrt{d}}{2}\left(1 - \frac{1}{d}\right)$, spaced by $\frac{1}{\sqrt{d}}$; therefore, roughly speaking, time becomes continuous and unbounded as $d \to \infty$ which is what one would expect from taking the infinite-dimensional limit.

Now, suppose the clock is initially prepared in the state $\ket{\theta_k}$ and is let to evolve freely for a time $\frac{\delta}{\sqrt{d}}$ before being measured in the time eigenbasis. Then the probability to collapse to the state $\ket{\theta_l}$ is given by:
\begin{align}
    & \left|\braket{\theta_l|\exp\left(-\frac{i\delta\hat{H}_d}{\sqrt{d}}\right)|\theta_k}\right|^2\nonumber\\
    & = \left|\braket{\theta_l|\theta_{k + \delta}}\right|^2\\
    & = \left|\frac{1}{d}\sum_{-\frac{d - 1}{2} \leq p \leq \frac{d - 1}{2}}e^{\frac{2\pi ip(l - k - \delta)}{d}}\right|^2\\
    & = \frac{1}{d^2}\sum_{-\frac{d - 1}{2} \leq p, q \leq \frac{d - 1}{2}}e^{\frac{2\pi i(p - q)l}{d}}e^{-\frac{2\pi i(p - q)k}{d}}e^{-\frac{2\pi i\delta(p - q)}{d}}\\
    & = \frac{1}{d^2}\sum_{-(d - 1) \leq r \leq d - 1}e^{\frac{2\pi irl}{d}}e^{-\frac{2\pi irk}{d}}e^{-\frac{2\pi i\delta r}{d}}\sum_{\substack{-\frac{d - 1}{2} \leq p, q \leq \frac{d - 1}{2}\\p - q = r}}1\\
    & = \frac{1}{d^2}\sum_{-(d - 1) \leq r \leq d - 1}e^{\frac{2\pi irl}{d}}e^{-\frac{2\pi irk}{d}}e^{-\frac{2\pi i\delta r}{d}}(d - |r|)\\
    & = \frac{1}{d^2}d + \sum_{0 < r \leq \frac{d - 1}{2}}e^{\frac{2\pi irl}{d}}e^{-\frac{2\pi irk}{d}}\left(e^{-\frac{2\pi i\delta r}{d}}(d - |r|) + e^{-\frac{2\pi i\delta(r - d)}{d}}(d - |r - d|)\right)\nonumber\\
    & \hspace{0.05\textwidth} + \sum_{-\frac{d - 1}{2} \leq r < 0}e^{\frac{2\pi irl}{d}}e^{-\frac{2\pi irk}{d}}\left(e^{-\frac{2\pi i\delta r}{d}}(d - |r|) + e^{-\frac{2\pi i\delta(r + d)}{d}}(d - |r + d|)\right)\\
    & = \frac{1}{d^2}d + \frac{1}{d^2}\sum_{0 < r \leq \frac{d - 1}{2}}e^{\frac{2\pi irl}{d}}e^{-\frac{2\pi irk}{d}}\left(e^{-\frac{2\pi i\delta r}{d}}(d - |r|) + e^{-\frac{2\pi i\delta(r - d)}{d}}r\right)\nonumber\\
    & \hspace{0.05\textwidth} + \frac{1}{d^2}\sum_{-\frac{d - 1}{2} \leq r < 0}e^{\frac{2\pi irl}{d}}e^{-\frac{2\pi irk}{d}}\left(e^{-\frac{2\pi i\delta r}{d}}(d - |r|) + e^{-\frac{2\pi i\delta(r + d)}{d}}(-r)\right)\\
    & = \frac{1}{d}\sum_{-\frac{d - 1}{2} \leq r \leq \frac{d - 1}{2}}e^{\frac{2\pi irl}{d}}e^{-\frac{2\pi irk}{d}}e^{-\frac{2\pi i\delta r}{d}}\left(1 - \left(1 - e^{2\pi i\delta\sign(r)}\right)\frac{|r|}{d}\right)\,.
\end{align}
One may regard this expression as the coefficient of a Markov transition matrix $M$ ($M_{lk}$ giving the probability of transitioning from $k$ to $l$). It is clear that it can be diagonalized by the eigenvectors $v_n = \left(\frac{1}{\sqrt{d}}e^{\frac{2\pi ink}{d}}\right)_{-\frac{d - 1}{2} \leq k \leq \frac{d - 1}{2}}$ for $-\frac{d - 1}{2} \leq n \leq \frac{d - 1}{2}$, where $v_n$ is associated to the eigenvalue $e^{-\frac{2\pi i\delta n}{d}}\left(1 - \left(1 - e^{2\pi i\delta\sign(n)}\right)\frac{|n|}{d}\right)$. Also, for all $I \geq 0$:
\begin{align}
    \left(M^I\right)_{lk} & = \frac{1}{d}\sum_{-\frac{d - 1}{2} \leq r \leq \frac{d - 1}{2}}e^{\frac{2\pi irl}{d}}e^{-\frac{2\pi irk}{d}}e^{-\frac{2\pi iI\delta r}{d}}\left(1 - \left(1 - e^{2\pi i\delta\sign(r)}\frac{|r|}{d}\right)\right)^I\,.
\end{align}
As an example, if one starts the clock in the state $\ket{\theta_k}$ before performing $I$ measurements on it at time intervals $\frac{\delta}{\sqrt{d}}$, the expectation of $\exp\left(\frac{2\pi mil_I}{d}\right)$ --- where $m$ is any integer in $\left[-\frac{d - 1}{2}, \frac{d - 1}{2}\right]$ and $l_I$ is the integer in $\left[-\frac{d - 1}{2}, \frac{d - 1}{2}\right]$ corresponding to the result of the $I^{\textrm{th}}$ measurement --- is:
\begin{align}
    \left\langle\exp\left(\frac{2\pi iml_I}{d}\right)\right\rangle & = \sum_{-\frac{d - 1}{2} \leq l_I \leq \frac{d - 1}{2}}\left(M^I\right)_{l_Ik}\exp\left(\frac{2\pi iml_I}{d}\right)\\
    & = e^{\frac{2\pi im(k + I\delta)}{d}}\left(1 - \left(1 - e^{2\pi i\delta\sign(-m)}\right)\frac{|m|}{d}\right)^I\,.
\end{align}
The factor $e^{\frac{2\pi im(k + I\delta)}{d}}$ 
in the result indicates that the ``expected angle'' of the clock after $I$ measurements is essentially $k + I\delta$ --- as one might have anticipated. The factor $\left(1 - \left(1 - e^{2\pi\delta\sign(-m)}\right)\frac{|m|}{d}\right)^I$ 
conveys information about the ``dispersion'' of this angle: in the limit where the angle is certain, it is $1$; in the limit where it is completely uncertain, it is $0$.\footnote{It is not entirely true that the ``expected angle'' is $k + I\delta$. Indeed, unless $\delta$ is a half-integer, the factor decaying exponentially in $d$ is not a positive real and therefore also contributes to the argument of the expectation value.} Reassuringly, in case $\delta$ is an integer, this factor is manifestly always $1$ --- the measured time is certain. This can be generalized to:
\begin{align}
\label{eq:pseudo_correlations_time_eigenbasis}
    & \left\langle\prod_{1 \leq p \leq N}\exp\left(\frac{2\pi im_{I_p}l_{I_p}}{d}\right)\right\rangle\nonumber\\
    & = \prod_{1 \leq p \leq N}\exp\left(\frac{2\pi im_p}{d}\left(k + \delta\sum_{1 \leq q \leq k}I_q\right)\right)\left(1 - \left(1 - e^{2\pi i\delta\sign\left(-\sum_{1 \leq q \leq p}m_q\right)}\right)\frac{\left|\sum_{1 \leq q \leq p}m_q\right|}{d}\right)^{I_p - I_{p - 1}}\,,
\end{align}
where $I_0 := 0$.

Now, if one wants to consider the limit of ``continuous measurement'' and see how the expectation above scales as $d \to \infty$, one may set $I := \lceil \frac{\tau}{\delta}\sqrt{d} \rceil$ (where $\tau > 0$ is fixed and corresponds to the ``continuous time interval'' during which one measures) so that to leading order in $d$, $I\frac{\delta}{\sqrt{d}}$ is independent of both $\delta$ and $d$ as $d \to \infty$. Taking $k = 0$ for simplicity, the above expectation behaves as follows as $d \to \infty$:
\begin{align}
    \left\langle\exp\left(\frac{2\pi iml_I}{d}\right)\right\rangle & = \exp\left(\frac{2\pi im\tau}{\sqrt{d}}\right)\left(1 - \tau\frac{1 - e^{2\pi i\delta\sign(-m)}}{\delta}\frac{|m|}{\sqrt{d}}\right)\left(1 + \mathcal{O}\left(\frac{1}{d}\right)\right)\,.
\end{align}

\subsubsection{Measured quasi-ideal clock with pseudo-Gaussian Kraus operators and states: derivation of formulae}
\label{sec:measured_quasi_ideal_clock_derivation_formulae}
In this section, we will describe a more general treatment of measurement for the quasi-ideal clock. Namely, we will allow for more or less sharp measurements (instead of a sharp time measurement in the previous paragraph) and quasi-ideal states for the initial state (instead of a time eigenstate in the previous paragraph).

A common approach for the treatment of measurement in infinite-dimensional systems is to choose Kraus operators that are Gaussian in the measured observable. It is also common to use Gaussian states for the initial state of the system. A natural transposition of this setting to the quasi-ideal clock is to use Kraus operators that are Jacobi $\theta$ functions in the time eigenbasis and, similarly, initial states which are quasi-ideal states. More precisely, following the notations of the main text, we define the Kraus operators as follows:
\begin{align}
    \hat{\Omega}\left(\widetilde{\xi}\right) & := \frac{1}{\sqrt{\theta_3\left(0, \frac{2i\sigma_m^2}{d}\right)}}\frac{1}{d^{1/4}}\sum_{-\frac{d - 1}{2} \leq k \leq \frac{d - 1}{2}}\theta_3\left(\frac{k - \widetilde{\xi}\sqrt{d}}{d}, \frac{i\sigma_m^2}{d}\right)\ket{\theta_k}\bra{\theta_k}\,, \qquad -\frac{\sqrt{d}}{2} \leq \widetilde{\xi} \leq \frac{\sqrt{d}}{2}\label{eq:Krau ops def}\\
    & =: \sum_{-\frac{d - 1}{2} \leq k \leq \frac{d - 1}{2}}\Omega_k\left(\widetilde{\xi}\right)\ket{\theta_k}\bra{\theta_k}\,.
\end{align}
The rationale behind this definition is that we interpret two consecutive time eigenstates $\ket{\theta_k}, \ket{\theta_{k + 1}}$ as representing two times $\widetilde{\xi} = \frac{k}{\sqrt{d}}, \widetilde{\xi}' = \frac{k + 1}{\sqrt{d}}$ 
separated by $\frac{1}{\sqrt{d}}$. This explains the $k - \widetilde{\xi}\sqrt{d}$ in the $\theta$ function and the range $\left[-\frac{\sqrt{d}}{2}, \frac{\sqrt{d}}{2}\right]$ for $\widetilde{\xi}$. As for the $\sigma_m^2$ parameter, it controls the precision of the measurement; more precisely, $\sigma_m$ is exactly the precision with which one measures $\widetilde{\xi}$. Therefore, if one wants to keep measuring $\widetilde{\xi}$ with a fixed precision in the limiting process $d \to \infty$, $\sigma_m$ must scale as a constant in this process. Concerning the initial state, we keep using the quasi-ideal state $\ket{\Psi_0}$ defined in equations  \eqref{eq:quasi_ideal_states} and \eqref{eq:quasi_ideal_states_ft}.

We start by showing that the Kraus operators above indeed define a normalized POVM:
\begin{lemma}\label{lem:kraus ops are normalised}
    The Kraus operators defined in equation \eqref{eq:Krau ops def} are properly normalized, i.e
    \begin{align}
        \int_{-\frac{\sqrt{d}}{2}}^{\frac{\sqrt{d}}{2}}\!\mathrm{d}\widetilde{\xi}\,\hat{\Omega}\left(\widetilde{\xi}\right)^{\dagger}\Hat{\Omega}\left(\widetilde{\xi}\right) & = \mathbf{1}_d\,.
    \end{align}
\begin{proof}
    \begin{align}
        & \int_{-\frac{\sqrt{d}}{2}}^{\frac{\sqrt{d}}{2}}\!\mathrm{d}\widetilde{\xi}\,\hat{\Omega}\left(\widetilde{\xi}\right)^{\dagger}\Hat{\Omega}\left(\widetilde{\xi}\right)\nonumber\\
        & = \frac{1}{\theta_3\left(0, \frac{2i\sigma_m^2}{d}\right)}\frac{1}{d^{1/2}}\sum_{-\frac{d - 1}{2} \leq k \leq \frac{d - 1}{2}}\int_{-\frac{\sqrt{d}}{2}}^{\frac{\sqrt{d}}{2}}\!\mathrm{d}\widetilde{\xi}\,\theta_3\left(\frac{k - \widetilde{\xi}\sqrt{d}}{d}, \frac{i\sigma_m^2}{d}\right)^2\ket{\theta_k}\bra{\theta_k}\,.
    \end{align}
    One then transforms $\theta_3\left(\frac{k - \widetilde{\xi}\sqrt{d}}{d}, \frac{i\sigma_m^2}{d}\right)^2$ according to the first equation of proposition \ref{prop:theta3_multiplication}:
        \begin{align}
        & \theta_3\left(\frac{k - \widetilde{\xi}\sqrt{d}}{d}, \frac{i\sigma_m^2}{d}\right)^2\nonumber\\
        & = \theta_3\left(0, \frac{2i\sigma_m^2}{d}\right)\theta_3\left(2\frac{k - \widetilde{\xi}\sqrt{d}}{d}, \frac{2i\sigma_m^2}{d}\right)\\
        & \hspace{0.05\textwidth} + \exp\left(-\frac{\pi\sigma_m^2}{d} + 2\pi i\frac{k - \widetilde{\xi}\sqrt{d}}{d}\right)\theta_3\left(\frac{i\sigma_m^2}{d}, \frac{2i\sigma_m^2}{d}\right)\theta_3\left(2\frac{k - \widetilde{\xi}\sqrt{d}}{d} + \frac{i\sigma_m^2}{d}, \frac{2i\sigma_m^2}{d}\right)\,.
    \end{align}
    Now, notice that under $\widetilde{\xi} \to \widetilde{\xi} + \frac{\sqrt{d}}{2}$, $\exp\left(-\frac{2\pi\sigma_m^2}{d} + 2\pi i\frac{k - \widetilde{\xi}\sqrt{d}}{d}\right)$ is odd whereas $\theta_3\left(2\frac{k - \widetilde{\xi}\sqrt{d}}{d} + \frac{i\sigma_m^2}{d}, \frac{2i\sigma_m^2}{d}\right)$ is even. Therefore, the second term of the sum cancels when integrated over $\xi \in \left[-\frac{\sqrt{d}}{2}, \frac{\sqrt{d}}{2}\right]$. As for the first term,
    \begin{align}
        \int_{-\frac{\sqrt{d}}{2}}^{\frac{\sqrt{d}}{2}}\!\mathrm{d}\widetilde{\xi}\,\theta_3\left(2\frac{k - \widetilde{\xi}\sqrt{d}}{d}, \frac{2i\sigma_m^2}{d}\right) & = \int_{-\frac{\sqrt{d}}{2}}^{\frac{\sqrt{d}}{2}}\!\mathrm{d}\widetilde{\xi}\,\theta_3\left(-2\frac{\widetilde{\xi}}{\sqrt{d}}, \frac{2i\sigma_m^2}{d}\right)\\
        & = 2\int_0^{\frac{\sqrt{d}}{2}}\!\mathrm{d}\widetilde{\xi}\,\theta_3\left(2\frac{\widetilde{\xi}}{\sqrt{d}}, \frac{2i\sigma_m^2}{d}\right)\\
        & = \sqrt{d}\int_0^1\!\mathrm{d}x\,\theta_3\left(x, \frac{2i\sigma_m^2}{d}\right)\\
        & = \sqrt{d}\sqrt{\frac{d}{2\sigma_m^2}}\int_{\mathbf{R}}\!\mathrm{d}x\,\exp\left(-\frac{\pi d}{2\sigma_m^2}x^2\right)\\
        & = \sqrt{d}\,.
    \end{align}
Therefore,
\begin{align}
    \int_{-\frac{\sqrt{d}}{2}}^{\frac{\sqrt{d}}{2}}\!\mathrm{d}\widetilde{\xi}\,\theta_3\left(\frac{k - \widetilde{\xi}\sqrt{d}}{d}, \frac{i\sigma_m^2}{d}\right)^2 & = \sqrt{d}\theta_3\left(0, \frac{2i\sigma_m^2}{d}\right)
\end{align}
and the result follows.
\end{proof}
\end{lemma}

To perform the computations to come, one will systematically have to compute integrals of the following form:

\begin{lemma}
\label{lemma:quasi_ideal_clock_kraus_moment}
    Let $k, k', n$ denote integers. Then the following holds:
    \begin{align}
        \int_{-\frac{\sqrt{d}}{2}}^{\frac{\sqrt{d}}{2}}\!\mathrm{d}\widetilde{\xi}\,\Omega_{k'}\left(\widetilde{\xi}\right)\Omega_k\left(\widetilde{\xi}\right)\exp\left(\frac{2\pi in\widetilde{\xi}}{\sqrt{d}}\right) & = \frac{\theta_3\left(\frac{k - k'}{d} + \frac{i\sigma_m^2}{d}n, \frac{2i\sigma_m^2}{d}\right)}{\theta_3\left(0, \frac{2i\sigma_m^2}{d}\right)}\exp\left(\frac{2\pi ink}{d} - \frac{\pi\sigma_m^2n^2}{d}\right)\\
        & = \frac{\theta_3\left(\frac{k - k'}{d} - \frac{i\sigma_m^2}{d}n, \frac{2i\sigma_m^2}{d}\right)}{\theta_3\left(0, \frac{2i\sigma_m^2}{d}\right)}\exp\left(\frac{2\pi ink'}{d} - \frac{\pi\sigma_m^2n^2}{d}\right)
    \end{align}
\begin{proof}
The integral to be evaluated is
\begin{align}
    & \int_{-\frac{\sqrt{d}}{2}}^{\frac{\sqrt{d}}{2}}\!\mathrm{d}\widetilde{\xi}\,\Omega_{k'}\left(\widetilde{\xi}\right)\Omega_k\left(\widetilde{\xi}\right)\exp\left(\frac{2\pi in\widetilde{\xi}}{\sqrt{d}}\right)\nonumber\\
    & = \frac{1}{\theta_3\left(0, \frac{2i\sigma_m^2}{d}\right)}\frac{1}{d^{1/2}}\int_{-\frac{\sqrt{d}}{2}}^{\frac{\sqrt{d}}{2}}\!\mathrm{d}\widetilde{\xi}\,\theta_3\left(\frac{k' - \widetilde{\xi}\sqrt{d}}{d}, \frac{i\sigma_m^2}{d}\right)\theta_3\left(\frac{k - \widetilde{\xi}\sqrt{d}}{d}, \frac{i\sigma_m^2}{d}\right)\exp\left(\frac{2\pi in\widetilde{\xi}}{\sqrt{d}}\right)\,.
\end{align}
Similar to what was done to prove the normalization of the POVM in the proof of lemma \ref{lem:kraus ops are normalised}, one writes:
\begin{align}
    & \theta_3\left(\frac{k' - \widetilde{\xi}\sqrt{d}}{d}, \frac{i\sigma_m^2}{d}\right)\theta_3\left(\frac{k - \widetilde{\xi}\sqrt{d}}{d}, \frac{i\sigma_m^2}{d}\right)\exp\left(\frac{2\pi in\widetilde{\xi}}{\sqrt{d}}\right)\nonumber\\
    & = \exp\left(\frac{2\pi in\widetilde{\xi}}{\sqrt{d}}\right)\left[\theta_3\left(\frac{k - k'}{d}, \frac{2i\sigma_m^2}{d}\right)\theta_3\left(\frac{k' + k}{d} - 2\frac{\widetilde{\xi}}{\sqrt{d}}, \frac{2i\sigma_m^2}{d}\right)\right.\nonumber\\
    & \left. + \exp\left(-\frac{\pi\sigma_m^2}{d} + 2\pi i\frac{k - \widetilde{\xi}\sqrt{d}}{d}\right)\theta_3\left(\frac{k - k'}{d} + \frac{i\sigma_m^2}{d}, \frac{2i\sigma_m^2}{d}\right)\theta_3\left(\frac{k + k'}{d} - 2\frac{\widetilde{\xi}}{\sqrt{d}} + \frac{i\sigma_m^2}{d}, \frac{2i\sigma_m^2}{d}\right)\right]\,.
\end{align}
One can now use the same parity arguments as in the proof of the normalization of the POVM. For even $n$, we therefore need to evaluate:
\begin{align}
    & \int_{-\frac{\sqrt{d}}{2}}^{\frac{\sqrt{d}}{2}}\!\mathrm{d}\widetilde{\xi}\,\exp\left(\frac{2\pi in\widetilde{\xi}}{\sqrt{d}}\right)\theta_3\left(\frac{k + k'}{d} - 2\frac{\widetilde{\xi}}{\sqrt{d}}, \frac{2i\sigma_m^2}{d}\right)\nonumber\\
    & = \frac{\sqrt{d}}{2}\int_{-1}^1\!\mathrm{d}x\,\exp\left(i\pi nx\right)\theta_3\left(x - \frac{k + k'}{d}, \frac{2i\sigma_m^2}{d}\right)\\
    & = \sqrt{d}\exp\left(i\pi n\frac{k + k'}{d}\right)\int_0^1\!\mathrm{d}x\,\exp\left(i\pi n x\right)\theta_3\left(x, \frac{2i\sigma_m^2}{d}\right)\\
    & = \sqrt{d}\exp\left(i\pi n\frac{k + k'}{d} - \frac{\pi n^2\sigma_m^2}{2d}\right)\,.
\end{align}
Therefore:
\begin{align}
    & \int_{-\frac{\sqrt{d}}{2}}^{\frac{\sqrt{d}}{2}}\!\mathrm{d}\widetilde{\xi}\,\Omega_{k'}\left(\widetilde{\xi}\right)\Omega_k\left(\widetilde{\xi}\right)\exp\left(\frac{2\pi in\widetilde{\xi}}{\sqrt{d}}\right)\nonumber\\
    & = \frac{\theta_3\left(\frac{k - k'}{d}, \frac{2i\sigma_m^2}{d}\right)}{\theta_3\left(0, \frac{2i\sigma_m^2}{d}\right)}\exp\left(i\pi n\frac{k + k'}{d} - \frac{\pi\sigma_m^2n^2}{2d}\right)\\
    & = \frac{\theta_3\left(\frac{k - k'}{d}, \frac{2i\sigma_m^2}{d}\right)}{\theta_3\left(0, \frac{2i\sigma_m^2}{d}\right)}\exp\left(\frac{2\pi\sigma_m^2}{d}\left(\frac{n}{2}\right)^2 - 2\pi i\frac{n}{2}\frac{k - k'}{d} + \frac{2\pi ink}{d} - \frac{\pi\sigma_m^2n^2}{d}\right)\\
    & = \frac{\theta_3\left(\frac{k - k'}{d} + \frac{i\sigma_m^2}{d}n, \frac{2i\sigma_m^2}{d}\right)}{\theta_3\left(0, \frac{2i\sigma_m^2}{d}\right)}\exp\left(\frac{2\pi ink}{d} - \frac{\pi\sigma_m^2n^2}{d}\right)\\
    & = \frac{\theta_3\left(\frac{k - k'}{d} - \frac{i\sigma_m^2}{d}n, \frac{2i\sigma_m^2}{d}\right)}{\theta_3\left(0, \frac{2i\sigma_m^2}{d}\right)}\exp\left(\frac{2\pi\sigma_m^2}{d}n^2 -2\pi in\left(\frac{k - k'}{d} - \frac{i\sigma_m^2}{d}n\right) + \frac{2\pi ink}{d} - \frac{\pi\sigma_m^2n^2}{d}\right)\\
    & = \frac{\theta_3\left(\frac{k - k'}{d} - \frac{i\sigma_m^2}{d}n, \frac{2i\sigma_m^2}{d}\right)}{\theta_3\left(0, \frac{2i\sigma_m^2}{d}\right)}\exp\left(\frac{2\pi ink'}{d} - \frac{\pi\sigma_m^2n^2}{d}\right)\,.
\end{align}
For odd $n$, we need to evaluate:
\begin{align}
    & \int_{-\frac{\sqrt{d}}{2}}^{\frac{\sqrt{d}}{2}}\!\mathrm{d}\widetilde{\xi}\,\exp\left(\frac{2\pi i(n - 1)\widetilde{\xi}}{\sqrt{d}}\right)\theta_3\left(\frac{k + k'}{d} - 2\frac{\widetilde{\xi}}{\sqrt{d}} + \frac{i\sigma_m^2}{d}, \frac{2i\sigma_m^2}{d}\right)\nonumber\\
    & = \sqrt{d}\exp\left(i\pi(n - 1)\frac{k + k'}{d}\right)\int_0^1\!\mathrm{d}x\,\exp\left(i\pi(n - 1)x\right)\theta_3\left(x - \frac{i\sigma_m^2}{d}, \frac{2i\sigma_m^2}{d}\right)\\
    & = \sqrt{d}\exp\left(i\pi(n - 1)\frac{k + k'}{d} - \frac{\pi\sigma_m^2n^2}{2d} + \frac{\pi\sigma_m^2}{2d}\right)\,.
\end{align}
It follows:
\begin{align}
    & \int_{-\frac{\sqrt{d}}{2}}^{\frac{\sqrt{d}}{2}}\!\mathrm{d}\widetilde{\xi}\,\Omega_{k'}\left(\widetilde{\xi}\right)\Omega_k\left(\widetilde{\xi}\right)\exp\left(\frac{2\pi in\widetilde{\xi}}{\sqrt{d}}\right)\nonumber\\
    & = \exp\left(-\frac{\pi\sigma_m^2}{2d} - \frac{\pi\sigma_m^2n^2}{2d} + \frac{2\pi ik}{d} + i\pi(n - 1)\frac{k + k'}{d}\right)\frac{\theta_3\left(\frac{k - k'}{d} + \frac{i\sigma_m^2}{d}, \frac{2i\sigma_m^2}{d}\right)}{\theta_3\left(0, \frac{2i\sigma_m^2}{d}\right)}\\
    & = \exp\left(\frac{2\pi\sigma_m^2}{d}\left(\frac{n - 1}{2}\right)^2 - 2\pi i\left(\frac{n - 1}{2}\right)\left(\frac{k - k'}{d} + \frac{i\sigma_m^2}{d}\right)\right.\nonumber\\
    & \left.\hspace{0.05\textwidth} + \frac{2\pi ik}{d} + \frac{2\pi i(n - 1)}{d}k - \frac{\pi\sigma_m^2}{d}n^2\right)\frac{\theta_3\left(\frac{k - k'}{d} + \frac{i\sigma_m^2}{d}, \frac{2i\sigma_m^2}{d}\right)}{\theta_3\left(0, \frac{2i\sigma_m^2}{d}\right)}\\
    & = \frac{\theta_3\left(\frac{k - k'}{d} + \frac{i\sigma_m^2}{d}n, \frac{2i\sigma_m^2}{d}\right)}{\theta_3\left(0, \frac{2i\sigma_m^2}{d}\right)}\exp\left(\frac{2\pi ink}{d} - \frac{\pi\sigma_m^2n^2}{d}\right)\,. 
\end{align}
\end{proof}
\end{lemma}

Having established these lemmas, one can now derive an expression for the measurement statistics. We adopt the general description of measurement developed in the main text. We consider a sequence of $J \geq 2$ measurements such that the $j^{\textrm{th}}$ measurement ($j \geq 1$) is separated from the $(j - 1)^{\textrm{th}}$ by a time interval $\frac{\delta_j}{\sqrt{d}}$. The outcomes of the measurements are denoted by $\left(\widetilde{\xi}_j\right)_{1 \leq j \leq J}$ and their joint probability distribution $f$ is given by:
\begin{align}
    f\left(\widetilde{\xi}_1, \ldots, \widetilde{\xi}_J\right) & = \left\lVert\left(\prod_{1 \leq j \leq J}\hat{\Omega}\left(\widetilde{\xi}_j\right)e^{-\frac{2\pi i\hat{H}_d\delta_j}{\sqrt{d}}}\right)\ket{\Psi_0}\right\rVert_2^2\,.
\end{align}
One will be interested in finding an expression for the moments of order at most $2$ of this distribution. More precisely, given integers $m, n$, we will compute:
\begin{align}
    & \int_{\left[-\frac{\sqrt{d}}{2}, \frac{\sqrt{d}}{2}\right]^J}\!\left(\prod_{1 \leq j \leq J}\mathrm{d}\widetilde{\xi}_j\right)\,\exp\left(\frac{2\pi in\widetilde{\xi}_J}{\sqrt{d}}\right)\exp\left(\frac{2\pi im\widetilde{\xi}_I}{\sqrt{d}}\right)f\left(\widetilde{\xi}_1, \ldots, \widetilde{\xi}_J\right)\,.
\end{align}
These quantities constitute a natural transposition of two-times correlation functions to the setting of the quasi-ideal clock.

First note,
\begin{align}
    & \left(\prod_{1 \leq j \leq J}\hat{\Omega}\left(\widetilde{\xi}_j\right)e^{-\frac{2\pi i\hat{H}_d\delta_j}{\sqrt{d}}}\right)\ket{\Psi_0}\nonumber\\
    & = \sum_{\substack{-\frac{d - 1}{2} \leq k \leq \frac{d - 1}{2}\\-\frac{d - 1}{2} \leq k_1, \ldots, k_J \leq \frac{d - 1}{2}}}\left(\prod_{1 \leq j \leq J}\Omega_{k_j}\left(\widetilde{\xi}_j\right)\right)\left(\prod_{1 \leq j < J}\braket{\theta_{k_{j + 1} + \delta_{j + 1}}|\theta_{k_j}}\right)\braket{\theta_{k_1 + \delta_1}|\theta_k}\braket{\theta_k|\Psi_0}\ket{\theta_{k_J}}\\
    & = \frac{1}{d^J}\sum_{\substack{-\frac{d - 1}{2} \leq k_0, k_1, \ldots, k_J \leq \frac{d - 1}{2}\\-\frac{d - 1}{2} \leq p_1, \ldots, p_J \leq \frac{d - 1}{2}}}\left(\prod_{1 \leq j \leq J}\Omega_{k_j}\left(\widetilde{\xi}_j\right)e^{\frac{2\pi ip_j(k_j + \delta_j - k_{j - 1})}{d}}\right)\psi(k_0)\ket{\theta_{k_J}}\,.
\end{align}
Therefore:
\begin{align*}
    & \left\lVert\left(\prod_{1 \leq j \leq J}\hat{\Omega}\left(\widetilde{\xi}_j\right)e^{-\frac{2\pi i\hat{H}_d\delta_j}{\sqrt{d}}}\right)\ket{\Psi_0}\right\rVert_2^2\\
    & = \frac{1}{d^{2J}}\sum_{\substack{-\frac{d - 1}{2} \leq k_0, \ldots, k_J \leq \frac{d - 1}{2}\\-\frac{d - 1}{2} \leq k_0', \ldots, k_J' \leq \frac{d - 1}{2}\\-\frac{d - 1}{2} \leq p_1, \ldots, p_J \leq \frac{d - 1}{2}\\-\frac{d - 1}{2} \leq p_1', \ldots, p_J' \leq \frac{d - 1}{2}}}\delta_{k_J' - k_J}\left(\prod_{1 \leq j \leq J}\Omega_{k_j}\left(\widetilde{\xi}_j\right)\Omega_{k_j'}\left(\widetilde{\xi}_j\right)e^{\frac{2\pi i}{d}\left(p_j(k_j + \delta_j - k_{j - 1}) - p_j'(k_j' + \delta_j - k_{j - 1}')\right)}\right)\,\psi(k_0)\psi(k_0')^*\,.
\end{align*}
Integrating $f\left(\widetilde{\xi}_1, \ldots, \widetilde{\xi}_J\right)$ over $\widetilde{\xi}_1, \ldots, \widetilde{\xi}_J$ against $\exp\left(\frac{2\pi in\widetilde{\xi}_J}{\sqrt{d}}\right)\exp\left(\frac{2\pi im\widetilde{\xi}_I}{\sqrt{d}}\right)$ yields (using lemma \ref{lemma:quasi_ideal_clock_kraus_moment}):
\begin{align*}
    & \int_{\left[-\frac{\sqrt{d}}{2}, \frac{\sqrt{d}}{2}\right]^J}\!\left(\prod_{1 \leq j \leq J}\mathrm{d}\widetilde{\xi}_j\right)\,\exp\left(\frac{2\pi in\widetilde{\xi}_J}{\sqrt{d}}\right)\exp\left(\frac{2\pi im\widetilde{\xi}_I}{\sqrt{d}}\right)f\left(\widetilde{\xi}_1, \ldots, \widetilde{\xi}_J\right)\\
    & = \frac{1}{d^{2J}}\sum_{\substack{-\frac{d - 1}{2} \leq k_0, \ldots, k_J \leq \frac{d - 1}{2}\\-\frac{d - 1}{2} \leq k_0', \ldots, k_J' \leq \frac{d - 1}{2}\\-\frac{d - 1}{2} \leq p_1, \ldots, p_J \leq \frac{d - 1}{2}\\-\frac{d - 1}{2} \leq p_1', \ldots, p_J' \leq \frac{d - 1}{2}}}\delta_{k_J' - k_J}\left(\prod_{\substack{1 \leq j \leq J\\j \neq I, J}}\frac{\theta_3\left(\frac{k_j' - k_j}{d}, \frac{2i\sigma_m^2}{d}\right)}{\theta_3\left(0, \frac{2i\sigma_m^2}{d}\right)}e^{\frac{2\pi i}{d}\left(p_j(k_j + \delta_j - k_{j - 1}) - p_j'(k_j' + \delta_j - k_{j - 1}')\right)}\right)\\
    & \hspace{0.1\textwidth} \times \frac{\theta_3\left(\frac{k_I' - k_I}{d} - \frac{i\sigma_m^2}{d}m, \frac{2i\sigma_m^2}{d}\right)}{\theta_3\left(0, \frac{2i\sigma_m^2}{d}\right)}e^{-\frac{\pi\sigma_m^2m^2}{d} + \frac{2\pi i}{d}\left(mk_I + p_I(k_I + \delta_I - k_{I - 1}) - p_I'(k_I' + \delta_I - k_{I - 1}')\right)}\\
    & \hspace{0.1\textwidth} \times \frac{\theta_3\left(\frac{k_J' - k_J}{d} - \frac{i\sigma_m^2}{d}n, \frac{2i\sigma_m^2}{d}\right)}{\theta_3\left(0, \frac{2i\sigma_m^2}{d}\right)}e^{-\frac{\pi\sigma_m^2n^2}{d} + \frac{2\pi i}{d}\left(nk_J + p_J(k_J + \delta_J - k_{J - 1}) - p_J'(k_J' + \delta_J - k_{J - 1}')\right)}\\
    & \hspace{0.1\textwidth} \times \psi(k_0)\psi(k_0')^*\\
    & = \frac{1}{d^{2J}}\sum_{\substack{-\frac{d - 1}{2} \leq k_0, \ldots, k_J \leq \frac{d - 1}{2}\\-\frac{d - 1}{2} \leq k_0', \ldots, k_{J - 1}' \leq \frac{d - 1}{2}\\-\frac{d - 1}{2} \leq p_1, \ldots, p_J \leq \frac{d - 1}{2}\\-\frac{d - 1}{2} \leq p_1', \ldots, p_J' \leq \frac{d - 1}{2}}}\left(\prod_{\substack{1 \leq j \leq J\\j \neq I, J}}\frac{\theta_3\left(\frac{k_j' - k_j}{d}, \frac{2i\sigma_m^2}{d}\right)}{\theta_3\left(0, \frac{2i\sigma_m^2}{d}\right)}e^{\frac{2\pi i}{d}\left(p_j(k_j + \delta_j - k_{j - 1}) - p_j'(k_j' + \delta_j - k_{j - 1}')\right)}\right)\\
    & \hspace{0.1\textwidth} \times \frac{\theta_3\left(\frac{k_I' - k_I}{d} - \frac{i\sigma_m^2}{d}m, \frac{2i\sigma_m^2}{d}\right)}{\theta_3\left(0, \frac{2i\sigma_m^2}{d}\right)}e^{-\frac{\pi\sigma_m^2m^2}{d} + \frac{2\pi i}{d}\left(mk_I + p_I(k_I + \delta_I - k_{I - 1}) - p_I'(k_I' + \delta_I - k_{I - 1}')\right)}\\
    & \hspace{0.1\textwidth} \times \frac{\theta_3\left(-\frac{i\sigma_m^2}{d}n, \frac{2i\sigma_m^2}{d}\right)}{\theta_3\left(0, \frac{2i\sigma_m^2}{d}\right)}e^{-\frac{\pi\sigma_m^2n^2}{d} + \frac{2\pi i}{d}\left(nk_J + p_J(k_J + \delta_J - k_{J - 1}) - p_J'(k_J + \delta_J - k_{J - 1}')\right)}\\
    & \hspace{0.1\textwidth} \times \psi(k_0)\psi(k_0')^*\,.
\end{align*}
One can then perform the summations over $k_1', \ldots, k_{J - 1}'$ which amounts to discrete Fourier transforms of $\theta_3$ functions (we therefore use equations \eqref{eq:theta3_dft} and \eqref{eq:theta3_inv_dft}). For example:
\begin{align}
    & \sum_{-\frac{d - 1}{2} \leq k_I' \leq \frac{d - 1}{2}}\theta_3\left(\frac{k_I' - k_I}{d} - \frac{i\sigma_m^2}{d}m, \frac{2i\sigma_m^2}{d}\right)e^{-\frac{2\pi i\left(p_I - p_{I + 1}'\right)k_I'}{d}}\nonumber\\
    & = \sqrt{\frac{d}{2\sigma_m^2}}\exp\left(-\frac{\pi d}{2\sigma_m^2}\left(-\frac{i\sigma_m^2}{d}m - \frac{k_I}{d}\right)^2\right)\theta_3\left(\frac{i}{2\sigma_m^2}\left(-\frac{k_I}{d} - \frac{i\sigma_m^2}{d}m\right) - \frac{p_I' - p_{I + 1}'}{d}, \frac{i}{2\sigma_m^2d}\right)\\
    & = \sqrt{\frac{d}{2\sigma_m^2}}\exp\left(\frac{\pi\sigma_m^2m^2}{2d} - \frac{\pi k_I^2}{2\sigma_m^2d} - \frac{i\pi mk_I}{d}\right)\theta_3\left(\frac{m}{2d} - \frac{ik_I}{2\sigma_m^2d} - \frac{p_I' - p_{I + 1}'}{d}, \frac{i}{2\sigma_m^2d}\right)\\
    & = \sqrt{\frac{d}{2\sigma_m^2}}\exp\left(\frac{\pi\sigma_m^2m^2}{2d} - \frac{\pi k_I^2}{2\sigma_m^2d} - \frac{i\pi mk_I}{d} + \frac{\pi}{2\sigma_m^2d}(-k_I)^2 - 2\pi i(-k_I)\left(\frac{m}{2d} - \frac{p_I'}{d}\right)\right)\nonumber\\
    & \hspace{0.05\textwidth} \times \theta_3\left(\frac{m}{2d} - \frac{p_I' - p_{I + 1}'}{d}, \frac{i}{2\sigma_m^2d}\right)\\
    & = \sqrt{\frac{d}{2\sigma_m^2}}\exp\left(\frac{\pi\sigma_m^2m^2}{2d} - \frac{2\pi ik_I\left(p_I' - p_{I + 1}'\right)}{d}\right)\theta_3\left(\frac{m}{2d} - \frac{p_I' - p_{I + 1}'}{d}, \frac{i}{2\sigma_m^2d}\right)\,.
\end{align}
One can now write:
\begin{align}
    & \int_{\left[-\frac{\sqrt{d}}{2}, \frac{\sqrt{d}}{2}\right]^\frac{J - 1}{2}}\!\left(\prod_{1 \leq j \leq J}\mathrm{d}\widetilde{\xi}_j\right)\,\exp\left(\frac{2\pi in\widetilde{\xi}_J}{\sqrt{d}}\right)\exp\left(\frac{2\pi im\widetilde{\xi}_I}{\sqrt{d}}\right)f\left(\widetilde{\xi}_1, \ldots, \widetilde{\xi}_J\right)\nonumber\\
    & = \frac{\left(\frac{d}{2\sigma_m^2}\right)^{J/2}}{d^{2J}\theta_3\left(0, \frac{2i\sigma_m^2}{d}\right)^J}\nonumber\\
    & \hspace{0.02\textwidth} \times \sum_{\substack{-\frac{d - 1}{2} \leq k_0, \ldots, k_J \leq \frac{d - 1}{2}\\-\frac{d - 1}{2} \leq k_0' \leq \frac{d - 1}{2}\\-\frac{d - 1}{2} \leq p_1, \ldots, p_J \leq \frac{d - 1}{2}\\-\frac{d - 1}{2} \leq p_1', \ldots, p_J' \leq \frac{d - 1}{2}}}\left(\prod_{\substack{1 \leq j \leq J\\j \neq I, J}}\theta_3\left(\frac{p_j' - p_{j + 1}'}{d}, \frac{i}{2\sigma_m^2d}\right)e^{\frac{2\pi i}{d}\left(-k_j(p_j' - p_{j + 1}') + p_j(k_j + \delta_j - k_{j - 1}) - p_j'\delta_j\right)}\right)\nonumber\\
    & \hspace{0.05\textwidth} \times \theta_3\left(\frac{p_I' - p_{I + 1}'}{d} - \frac{m}{2d}, \frac{i}{2\sigma_m^2d}\right)e^{-\frac{\pi\sigma_m^2m^2}{2d} + \frac{2\pi i}{d}\left(k_I(m - p_I' + p_{I + 1}') + p_I(k_I + \delta_I - k_{I - 1}) - p_I'\delta_I\right)}\nonumber\\
    & \hspace{0.05\textwidth} \times \theta_3\left(-\frac{i\sigma_m^2}{d}n, \frac{2i\sigma_m^2}{d}\right)e^{-\frac{\pi\sigma_m^2n^2}{d} + \frac{2\pi i}{d}\left(nk_J + p_J(k_J + \delta_J - k_{J - 1}) - p_J'(k_J + \delta_J)\right)}\nonumber\\
    & \hspace{0.05\textwidth} \times e^{\frac{2\pi ik_0'p_1'}{d}}\psi(k_0)\psi(k_0')^*\,.
\end{align}
As one can see from the expression above, it is now easy to perform the summation over $k_1, \ldots, k_J$. For convenience, we introduce the notation (pseudo-Kronecker delta):
\begin{align}
    \delta^{[d]}_{a} & := \left\{\begin{array}{ll}
        1 & \textrm{ if } a := 0\,[d]\\
        0 & \textrm{ otherwise} \,.
    \end{array}\right.
\end{align}
One obtains:
\begin{align}
    & \int_{\left[-\frac{\sqrt{d}}{2}, \frac{\sqrt{d}}{2}\right]^J}\!\left(\prod_{1 \leq j \leq J}\mathrm{d}\widetilde{\xi}_j\right)\,\exp\left(\frac{2\pi in\widetilde{\xi}_J}{\sqrt{d}}\right)\exp\left(\frac{2\pi im\widetilde{\xi}_I}{\sqrt{d}}\right)f\left(\widetilde{\xi}_1, \ldots, \widetilde{\xi}_J\right)\\
    & = \frac{\left(\frac{d}{2\sigma_m^2}\right)^{\frac{J - 1}{2}}}{d^J\theta_3\left(0, \frac{2i\sigma_m^2}{d}\right)^J}\sum_{\substack{-\frac{d - 1}{2} \leq k_0, k_0' \leq \frac{d - 1}{2}\\-\frac{d - 1}{2} \leq p_1, \ldots, p_J \leq \frac{d - 1}{2}\\-\frac{d - 1}{2} \leq p_1', \ldots, p_J' \leq \frac{d - 1}{2}}}\left(\prod_{\substack{1 \leq j \leq J\\j \neq I, J}}\theta_3\left(\frac{p_j' - p_{j + 1}'}{d}, \frac{i}{2\sigma_m^2d}\right)\right)\\
    & \hspace{0.05\textwidth} \times \theta_3\left(\frac{p_I' - p_{I + 1}'}{d} - \frac{m}{2d}, \frac{i}{2\sigma_m^2d}\right)\theta_3\left(-\frac{i\sigma_m^2}{d}n, \frac{2i\sigma_m^2}{d}\right)e^{-\frac{\pi\sigma_m^2m^2}{2d} - \frac{\pi\sigma_m^2n^2}{d} + \frac{2\pi i}{d}\sum_{1 \leq j \leq J}\delta_j(p_j - p_j')}\\
    & \hspace{0.05\textwidth} \times \left(\prod_{\substack{1 \leq j \leq J\\j \neq I, J}}\delta^{[d]}_{p_{j + 1}' - p_j' - p_{j + 1} + p_j}\right)\delta^{[d]}_{m + p_{I + 1}' - p_I' - p_{I + 1} + p_I}\delta^{[d]}_{n + p_J - p_J'}\\
    & \hspace{0.05\textwidth} \times e^{-\frac{2\pi ip_1k_0}{d}}e^{\frac{2\pi ip_1'k_0'}{d}}\psi(k_0)\psi(k_0')^*\,.
\end{align}
Now, one may simplify the product of pseudo-Kronecker deltas as follows:
\begin{align}
    & \left(\prod_{\substack{1 \leq j \leq J\\j \neq I, J}}\delta^{[d]}_{p_{j + 1}' - p_j' - p_{j + 1} + p_j}\right)\delta^{[d]}_{m + p_{I + 1}' - p_{I + 1} - p_I' + p_I}\delta^{[d]}_{n + p_J - p_J'}\\
    & = \left(\prod_{1 \leq j \leq I}\delta^{[d]}_{m + n - p_j' + p_j}\right)\left(\prod_{I + 1 \leq j \leq J}\delta^{[d]}_{n - p_j' + p_j}\right)\,.
\end{align}
One should be careful though before substituting (for instance) $p_j \to p_j' - m - n\,(1 \leq j \leq I)$ in the summand. Indeed, unless all the $\delta_j$ are integers, $\exp\left(\frac{2\pi i}{d}\sum_{1 \leq j \leq J}\delta_j(p_j - p_j')\right)$ is not invariant in the variables $p_j, p_j'$ modulo $d$. However, if for the sake of simplicity one assumes $|m + n| \leq d$, one knows that for $m + n \geq 0$ (respectively $m + n \leq 0$), $p_j' - m - n$ lies either in $\left[-\frac{3d - 1}{2}, -\frac{d + 1}{2}\right]$ (resp. $\left[\frac{d + 1}{2}, \frac{3d - 1}{2}\right]$) or $\left[-\frac{d - 1}{2}, \frac{d - 1}{2}\right]$. In the former case, the solution $p_j$ to $p_j := p_j' - m - n\,[d]$ is $p_j = p_j' - m - n + d$ (resp. $p_j = p_j' - m - n - d$) whereas it is simply $p_j = p_j' - m - n$ in the latter situation. All in all, it is legitimate to write:
\begin{align}
    & \delta_{m + n - p_j' + p_j}e^{\frac{2\pi i\delta_j(p_j - p_j')}{d}}\\
    & = e^{-\frac{2\pi i\delta_j(m + n)}{d}}\left(1\mathbf{1}_{-\frac{d - 1}{2} \leq p_j' - m - n \leq \frac{d - 1}{2}} + e^{2\pi i\delta_j}\mathbf{1}_{p_j' - m - n < -\frac{d - 1}{2}} + e^{-2\pi i\delta_j}\mathbf{1}_{p_j' - m - n > \frac{d - 1}{2}}\right)\delta^{[d]}_{m + n - p_j' + p_j}\\
    & = e^{-\frac{2\pi i\delta_j(m + n)}{d}}\left(1 - \left(1 - e^{2\pi i\delta_j}\right)\mathbf{1}_{p_j' - m - n < -\frac{d - 1}{2}} - \left(1 - e^{-2\pi i\delta_j}\right)\mathbf{1}_{p_j' - m - n > \frac{d - 1}{2}}\right)\delta^{[d]}_{m + n - p_j' + p_j}\\
    & = e^{-\frac{2\pi i\delta_j(m + n)}{d}}\left(1 - \left(1 - e^{2\pi i\delta_j\sign(m + n)}\right)\mathbf{1}_{|p_j' - m - n| > \frac{d - 1}{2}}\right)\delta^{[d]}_{m + n - p_j' + p_j}\,.
\end{align}
Therefore, the original integral becomes:
\begin{align}
\label{eq:pseudo_correlation_derivation_step1}
    & \int_{\left[-\frac{\sqrt{d}}{2}, \frac{\sqrt{d}}{2}\right]^J}\!\left(\prod_{1 \leq j \leq J}\mathrm{d}\widetilde{\xi}_j\right)\,\exp\left(\frac{2\pi in\widetilde{\xi}_J}{\sqrt{d}}\right)\exp\left(\frac{2\pi im\widetilde{\xi}_I}{\sqrt{d}}\right)f\left(\widetilde{\xi}_1, \ldots, \widetilde{\xi}_J\right)\nonumber\\
    & = \frac{\left(\frac{d}{2\sigma_m^2}\right)^{\frac{J - 1}{2}}}{d^{J - 1}\theta_3\left(0, \frac{2i\sigma_m^2}{d}\right)^J}\sum_{\substack{-\frac{d - 1}{2} \leq p_1, \ldots, p_J \leq \frac{d - 1}{2}\\-\frac{d - 1}{2} \leq p_1', \ldots, p_J' \leq \frac{d - 1}{2}}}\left(\prod_{\substack{1 \leq j \leq J\\j \neq I, J}}\theta_3\left(\frac{p_j' - p_{j + 1}'}{d}, \frac{i}{2\sigma_m^2d}\right)\right)\nonumber\\
    & \hspace{0.05\textwidth} \times \theta_3\left(\frac{p_I' - p_{I + 1}'}{d} - \frac{m}{2d}, \frac{i}{2\sigma_m^2d}\right)\theta_3\left(-\frac{i\sigma_m^2}{d}n, \frac{2i\sigma_m^2}{d}\right)e^{-\frac{\pi\sigma_m^2m^2}{2d} - \frac{\pi\sigma_m^2n^2}{d} + \frac{2\pi i}{d}\sum_{1 \leq j \leq J}\delta_j(p_j - p_j')}\nonumber\\
    & \hspace{0.05\textwidth} \times \left(\prod_{1 \leq j \leq I}\delta^{[d]}_{m + n - p_j' + p_j}\right)\left(\prod_{I + 1 \leq j \leq J}\delta^{[d]}_{n - p_j' + p_j}\right)\, \widetilde{\psi}(p_1)\widetilde{\psi}(p_1')^*\nonumber\\
    & = \frac{\left(\frac{d}{2\sigma_m^2}\right)^{\frac{J - 1}{2}}}{d^{J - 1}\theta_3\left(0, \frac{2i\sigma_m^2}{d}\right)^J}e^{-\frac{\pi\sigma_m^2m^2}{2d} - \frac{\pi\sigma_m^2n^2}{d} -\frac{2\pi i}{d}\left(m\sum_{1 \leq j \leq I}\delta_j + n\sum_{1 \leq j \leq J}\delta_j\right)}\nonumber\\
    & \hspace{0.05\textwidth} \times \sum_{\substack{-\frac{d - 1}{2} \leq p_1, \ldots, p_J \leq \frac{d - 1}{2}\\-\frac{d - 1}{2} \leq p_1', \ldots, p_J' \leq \frac{d - 1}{2}}}\left(\prod_{\substack{1 \leq j \leq J\\j \neq I, J}}\theta_3\left(\frac{p_j' - p_{j + 1}'}{d}, \frac{i}{2\sigma_m^2d}\right)\right)\nonumber\\
    & \hspace{0.2\textwidth} \times \theta_3\left(\frac{p_I' - p_{I + 1}'}{d} - \frac{m}{2d}, \frac{i}{2\sigma_m^2d}\right)\theta_3\left(-\frac{i\sigma_m^2}{d}n, \frac{2i\sigma_m^2}{d}\right)\nonumber\\
    & \hspace{0.2\textwidth} \times \left(\prod_{1 \leq j \leq I}\delta^{[d]}_{m + n - p_j' + p_j}\right)\left(\prod_{I + 1 \leq j \leq J}\delta^{[d]}_{n - p_j' + p_j}\right)\nonumber\\
    & \hspace{0.2\textwidth} \times \prod_{1 \leq j \leq I}\left(1 - \left(1 - e^{2\pi i\delta_j\sign(m + n)}\right)\mathbf{1}_{|p_j' - m - n| > \frac{d - 1}{2}}\right)\nonumber\\
    & \hspace{0.2\textwidth} \times \prod_{I + 1 \leq j \leq J}\left(1 - \left(1 - e^{2\pi i\delta_j\sign(n)}\right)\mathbf{1}_{|p_j' - n| > \frac{d - 1}{2}}\right)\, \widetilde{\psi}(p_1)\widetilde{\psi}(p_1')^*\nonumber\\
    & = \frac{\left(\frac{d}{2\sigma_m^2}\right)^{\frac{J - 1}{2}}}{d^{J - 1}\theta_3\left(0, \frac{2i\sigma_m^2}{d}\right)^J}e^{-\frac{\pi\sigma_m^2m^2}{2d} - \frac{\pi\sigma_m^2n^2}{d} -\frac{2\pi i}{d}\left(m\sum_{1 \leq j \leq I}\delta_j + n\sum_{1 \leq j \leq J}\delta_j\right)}\nonumber\\
    & \hspace{0.05\textwidth} \times \sum_{\substack{-\frac{d - 1}{2} \leq p_1', \ldots, p_J' \leq \frac{d - 1}{2}}}\left(\prod_{\substack{1 \leq j \leq J\\j \neq I, J}}\theta_3\left(\frac{p_j' - p_{j + 1}'}{d}, \frac{i}{2\sigma_m^2d}\right)\right)\nonumber\\
    & \hspace{0.2\textwidth} \times \theta_3\left(\frac{p_I' - p_{I + 1}'}{d} - \frac{m}{2d}, \frac{i}{2\sigma_m^2d}\right)\theta_3\left(-\frac{i\sigma_m^2}{d}n, \frac{2i\sigma_m^2}{d}\right)\nonumber\\
    & \hspace{0.2\textwidth} \times \prod_{1 \leq j \leq I}\left(1 - \left(1 - e^{2\pi i\delta_j\sign(m + n)}\right)\mathbf{1}_{|p_j' - m - n| > \frac{d - 1}{2}}\right)\nonumber\\
    & \hspace{0.2\textwidth} \times \prod_{I + 1 \leq j \leq J}\left(1 - \left(1 - e^{2\pi i\delta_j\sign(n)}\right)\mathbf{1}_{|p_j' - n| > \frac{d - 1}{2}}\right)\nonumber\\
    & \hspace{0.2\textwidth} \times \widetilde{\psi}(p_1' - m - n)\widetilde{\psi}(p_1')^*\,.
\end{align}
Next, one wants to ``normalize'' the $\theta_3$ functions appearing in the sum so that they add up to $1$ when summed over one of the $p_j'$ appearing in their first argument. To achieve this, we use the formula (derived from equations \eqref{eq:theta3_dft} and \eqref{eq:theta3_inv_dft}):
\begin{align}
\label{eq:theta_3_sum_over_p}
    \sum_{-\frac{d - 1}{2} \leq p \leq \frac{d - 1}{2}}\theta_3\left(\frac{p}{d} + \frac{m}{2d}, \frac{i}{2\sigma_m^2d}\right) & = \sqrt{2\sigma_m^2d}e^{-\frac{\pi\sigma_m^2m^2}{2d}}\theta_3\left(\frac{i\sigma_m^2}{d}m, \frac{2i\sigma_m^2}{d}\right)\,.
\end{align}
Therefore by plugging equation \eqref{eq:theta_3_sum_over_p} into \eqref{eq:pseudo_correlation_derivation_step1}, we ``normalize'' the $\theta_3$ functions. Performing this step followed by renaming the indices $p' \to p$ for simplicity, one ends up with:
\begin{align}
    & \int_{\left[-\frac{\sqrt{d}}{2}, \frac{\sqrt{d}}{2}\right]^J}\!\left(\prod_{1 \leq j \leq J}\mathrm{d}\widetilde{\xi}_j\right)\,\exp\left(\frac{2\pi in\widetilde{\xi}_J}{\sqrt{d}}\right)\exp\left(\frac{2\pi im\widetilde{\xi}_I}{\sqrt{d}}\right)f\left(\widetilde{\xi}_1, \ldots, \widetilde{\xi}_J\right)\nonumber\\
    & = \frac{\theta_3\left(\frac{i\sigma_m^2}{d}m, \frac{2i\sigma_m^2}{d}\right)\theta_3\left(\frac{i\sigma_m^2}{d}n, \frac{2i\sigma_m^2}{d}\right)}{\theta_3\left(0, \frac{2i\sigma_m^2}{d}\right)^2}e^{-\frac{\pi\sigma_m^2}{d}(m^2 + n^2) -\frac{2\pi i}{d}\left(m\sum_{1 \leq j \leq I}\delta_j + n\sum_{1 \leq j \leq J}\delta_j\right)}\nonumber\\
    & \hspace{0.05\textwidth} \times \sum_{-\frac{d - 1}{2} \leq p_1, \ldots, p_J \leq \frac{d - 1}{2}}\left(\prod_{\substack{1 \leq j < J\\j \neq I}}\frac{\theta_3\left(\frac{p_{j + 1} - p_j}{d}, \frac{i}{2\sigma_m^2d}\right)}{\sqrt{2\sigma_m^2d}\theta_3\left(0, \frac{2i\sigma_m^2}{d}\right)}\right)\frac{\theta_3\left(\frac{p_{I + 1} - p_I}{d} + \frac{m}{2d}, \frac{i}{2\sigma_m^2d}\right)}{\sqrt{2\sigma_m^2d}e^{-\frac{\pi\sigma_m^2m^2}{2d}}\theta_3\left(\frac{i\sigma_m^2}{d}m, \frac{2i\sigma_m^2}{d}\right)}\nonumber\\
    & \hspace{0.1\textwidth} \times \prod_{1 \leq j \leq I}\left(1 - \left(1 - e^{2\pi i\delta_j\sign(m + n)}\right)\mathbf{1}_{|p_j - m - n| > \frac{d - 1}{2}}\right)\nonumber\\
    & \hspace{0.1\textwidth} \times \prod_{I + 1 \leq j \leq J}\left(1 - \left(1 - e^{2\pi i\delta_j\sign(n)}\right)\mathbf{1}_{|p_j - n| > \frac{d - 1}{2}}\right)\nonumber\\
    & \hspace{0.1\textwidth} \times \widetilde{\psi}(p_1 - m - n)\widetilde{\psi}(p_1)^*\,.
\end{align}
Finally, one will normalize the wavefunction part in the sum so that summing $\textrm{constant} \times \widetilde{\psi}(p_1 - m - n)\widetilde{\psi}(p_1)^*$ over $p_1$ yields $1$. Given that $\widetilde{\psi}(p) = \mathcal{N}'\theta_3\left(\frac{p - n_0}{d}, \frac{i}{\xi^2d}\right)$, one needs to compute for all integer $r$:
\begin{align}
    & \sum_{-\frac{d - 1}{2} \leq k \leq \frac{d - 1}{2}}\theta_3\left(\frac{k - r}{d}, \frac{i}{\xi^2d}\right)\theta_3\left(\frac{k}{d}, \frac{i}{\xi^2d}\right)\nonumber\\
    & = \sum_{-\frac{d - 1}{2} \leq k \leq \frac{k - 1}{2}}\frac{1}{2}\left[\theta_3\left(-\frac{r}{2d}, \frac{i}{2\xi^2d}\right)\theta_3\left(\frac{k}{d} - \frac{r}{2d}, \frac{i}{2\xi^2d}\right)\right.\nonumber\\
    & \left. \hspace{0.1\textwidth} + \theta_3\left(-\frac{r}{2d} + \frac{1}{2}, \frac{i}{2\xi^2d}\right)\theta_3\left(\frac{k}{d} - \frac{r}{2d} + \frac{1}{2}, \frac{i}{2\xi^2d}\right)\right]\\
    & = \frac{1}{2}\sqrt{2\xi^2d}\left[\theta_3\left(-\frac{r}{2d}, \frac{i}{2\xi^2d}\right)\exp\left(-2\pi\xi^2d\left(-\frac{r}{2d}\right)^2\right)\theta_3\left(2i\xi^2\left(-\frac{r}{2d}\right), \frac{2i\xi^2}{d}\right)\right.\nonumber\\
    & \left.\hspace{0.1\textwidth} + \theta_3\left(-\frac{r}{2d} + \frac{1}{2}, \frac{i}{2\xi^2d}\right)\exp\left(-2\pi\xi^2d\left(-\frac{r}{2d} + \frac{1}{2}\right)^2\right)\theta_3\left(2i\xi^2\left(-\frac{r}{2d} + \frac{1}{2}\right), \frac{2i\xi^2}{d}\right)\right]\\
    & = \frac{d}{2}\left(\theta_3\left(\frac{r}{2d}, \frac{i}{2\xi^2d}\right)\theta_3\left(\frac{r}{2}, \frac{id}{2\xi^2}\right) + \theta_3\left(-\frac{r}{2d} + \frac{1}{2}, \frac{i}{2\xi^2d}\right)\theta_3\left(\frac{d}{2}- \frac{r}{2}, \frac{id}{2\xi^2}\right)\right)\\
    & = \frac{d}{2}\left(\theta_3\left(\frac{r}{2d}, \frac{i}{2\xi^2d}\right)\theta_3\left(\frac{r}{2}, \frac{id}{2\xi^2}\right) + \theta_3\left(-\frac{r}{2d} + \frac{1}{2}, \frac{i}{2\xi^2d}\right)\theta_3\left(\frac{r}{2} + \frac{1}{2}, \frac{id}{2\xi^2}\right)\right)\,,
\end{align}
where we used in the last line that $d$ is odd. For $r = 0$, this is simply $\frac{1}{\mathcal{N}'^2}$. Therefore, the sought normalized expression of the integral under consideration is:
\begin{align}
\label{eq:quasi_ideal_clock_correlation}
    & \int_{\left[-\frac{\sqrt{d}}{2}, \frac{\sqrt{d}}{2}\right]^J}\!\left(\prod_{1 \leq j \leq J}\mathrm{d}\widetilde{\xi}_j\right)\,\exp\left(\frac{2\pi in\widetilde{\xi}_J}{\sqrt{d}}\right)\exp\left(\frac{2\pi im\widetilde{\xi}_I}{\sqrt{d}}\right)f\left(\widetilde{\xi}_1, \ldots, \widetilde{\xi}_J\right)\nonumber\\
    & = \frac{\theta_3\left(\frac{i\sigma_m^2}{d}m, \frac{2i\sigma_m^2}{d}\right)\theta_3\left(\frac{i\sigma_m^2}{d}n, \frac{2i\sigma_m^2}{d}\right)}{\theta_3\left(0, \frac{2i\sigma_m^2}{d}\right)^2}e^{-\frac{\pi\sigma_m^2}{d}(m^2 + n^2) +\frac{2\pi i}{d}\left(m\sum_{1 \leq j \leq I}\delta_j + n\sum_{1 \leq j \leq J}\delta_j\right)}\nonumber\\
    & \hspace{0.05\textwidth} \times \frac{\theta_3\left(\frac{m + n}{2d}, \frac{i}{2\xi^2d}\right)\theta_3\left(\frac{m + n}{2}, \frac{id}{2\xi^2}\right) + \theta_3\left(\frac{1}{2} - \frac{m + n}{2d}, \frac{i}{2\xi^2d}\right)\theta_3\left(\frac{m + n}{2} + \frac{1}{2}, \frac{id}{2\xi^2}\right)}{\theta_3\left(0, \frac{i}{2\xi^2d}\right)\theta_3\left(0, \frac{id}{2\xi^2}\right) + \theta_3\left(\frac{1}{2}, \frac{i}{2\xi^2d}\right)\theta_3\left(\frac{1}{2}, \frac{id}{2\xi^2}\right)}\nonumber\\
    & \hspace{0.05\textwidth} \times \sum_{-\frac{d - 1}{2} \leq p_1, \ldots, p_J \leq \frac{d - 1}{2}}\left(\prod_{\substack{1 \leq j < J\\j \neq I}}\frac{\theta_3\left(\frac{p_{j + 1} - p_j}{d}, \frac{i}{2\sigma_m^2d}\right)}{\sqrt{2\sigma_m^2d}\theta_3\left(0, \frac{2i\sigma_m^2}{d}\right)}\right)\frac{\theta_3\left(\frac{p_{I + 1} - p_I}{d} + \frac{m}{2d}, \frac{i}{2\sigma_m^2d}\right)}{\sqrt{2\sigma_m^2d}e^{-\frac{\pi\sigma_m^2m^2}{2d}}\theta_3\left(\frac{i\sigma_m^2}{d}m, \frac{2i\sigma_m^2}{d}\right)}\nonumber\\
    & \hspace{0.1\textwidth} \times \prod_{1 \leq j \leq I}\left(1 - \left(1 - e^{-2\pi i\delta_j\sign(m + n)}\right)\mathbf{1}_{|p_j - m - n| > \frac{d - 1}{2}}\right)\nonumber\\
    & \hspace{0.1\textwidth} \times \prod_{I + 1 \leq j \leq J}\left(1 - \left(1 - e^{-2\pi i\delta_j\sign(n)}\right)\mathbf{1}_{|p_j - n| > \frac{d - 1}{2}}\right)\nonumber\\
    & \hspace{0.1\textwidth} \times \frac{\theta_3\left(\frac{p_1 - n_0 - m - n}{d}, \frac{i}{\xi^2d}\right)\theta_3\left(\frac{p_1 - n _0}{d}, \frac{i}{\xi^2d}\right)}{\frac{d}{2}\left(\theta_3\left(\frac{m + n}{2d}, \frac{i}{2\xi^2d}\right)\theta_3\left(\frac{m + n}{2}, \frac{id}{2\xi^2}\right) + \theta_3\left(\frac{1}{2} - \frac{m + n}{2d}, \frac{i}{2\xi^2d}\right)\theta_3\left(\frac{m + n}{2} + \frac{1}{2}, \frac{id}{2\xi^2}\right)\right)}\,.
\end{align}
The sum has now a clear probabilitic interpretation as an expectation computed over a random walk on a ring. $p_1, \ldots, p_J$ are to be interpreted as the steps of this random walk and the probability distribution for the amplitude of a jump is
\begin{align}
    q & \longmapsto \frac{\theta_3\left(\frac{q}{d}, \frac{i}{2\sigma_m^2d}\right)}{\sqrt{2\sigma_m^2d}\theta_3\left(0, \frac{2i\sigma_m^2}{d}\right)}
\end{align}
for all jumps except for jump $I$ (leading from position $p_I$ to position $p_{I + 1}$) where it is given by:
\begin{align}
    q_I & \longmapsto \frac{\theta_3\left(\frac{q_I}{d} + \frac{m}{2d}, \frac{i}{2\sigma_m^2d}\right)}{\sqrt{2\sigma_m^2d}e^{-\frac{\pi\sigma_m^2m^2}{d}}\theta_3\left(\frac{i\sigma_m^2}{d}m, \frac{2i\sigma_m^2}{d}\right)}\,.
\end{align}
Roughly speaking, this means that the jumps $j \neq I$ have expectation $0$, with a typical standard deviation $\frac{\sqrt{d}}{\sigma_m}$. As for the jump $I$, it has an expectation $-\frac{m}{2}$ and the same standard deviation. Finally, the initial distribution of the random walk is prescribed by the initial state as follows:
\begin{align}
    p_1 & \longmapsto \frac{\theta_3\left(\frac{p_1 - n_0 - m - n}{d}, \frac{i}{\xi^2d}\right)\theta_3\left(\frac{p_1 - n _0}{d}, \frac{i}{\xi^2d}\right)}{\frac{d}{2}\left(\theta_3\left(\frac{m + n}{2d}, \frac{i}{2\xi^2d}\right)\theta_3\left(\frac{m + n}{2}, \frac{id}{2\xi^2}\right) + \theta_3\left(\frac{1}{2} - \frac{m + n}{2d}, \frac{i}{2\xi^2d}\right)\theta_3\left(\frac{m + n}{2} + \frac{1}{2}, \frac{id}{2\xi^2}\right)\right)}\,.
\end{align}
With this interpretation in mind, the sum could be rewritten as a probabilistic expectation:
\begin{align}
    & \mathbf{E}^{\ket{\Psi_0}}\left[\prod_{1 \leq j \leq I}\left(1 - \left(1 - e^{2\pi i\delta_j\sign(m + n)}\right)\mathbf{1}_{|p_j - m - n| > \frac{d - 1}{2}}\right)\right.\nonumber\\
    & \left.\hspace{0.1\textwidth} \times \prod_{I + 1 \leq j \leq J}\left(1 - \left(1 - e^{2\pi i\delta_j\sign(n)}\right)\mathbf{1}_{|p_j - n| > \frac{d - 1}{2}}\right)\right]\,,
\end{align}
where we used the standard notation $\mathbf{E}^{\mu}[\cdot]$ for the expectation given an initial distribution $\mu$. Note that the expectation is trivially bounded by $1$ in norm as the random variable is. This implies that the parametrization of the Kraus operators (through $\sigma_m$) and that of the initial state (through $\widetilde{\xi}^2$) alone enforce bounds on the modulus of $\left\langle\exp\left(\frac{2\pi in\widetilde{\xi}_J}{\sqrt{d}}\right)\exp\left(\frac{2\pi im\widetilde{\xi}_I}{\sqrt{d}}\right)\right\rangle$, as specified in the following two propositions.

The proposition below provides a bound on $\left\langle\exp\left(\frac{2\pi in\widetilde{\xi}_J}{\sqrt{d}}\right)\exp\left(\frac{2\pi im\widetilde{\xi}_I}{\sqrt{d}}\right)\right\rangle$ based on the scaling of $\xi$ with respect to $d$.
\begin{proposition}
\label{prop:variance_suppression_wavefunction}
    Let $r$ be an integer, $c > 0$ and $-1 < \alpha < 1$. Suppose $\xi^2 = cd^{\alpha}$. Then as $d \to \infty$, the following estimate holds:
    \begin{align}
        & \frac{\theta_3\left(\frac{r}{2d}, \frac{i}{2\xi^2d}\right)\theta_3\left(\frac{r}{2}, \frac{id}{2\xi^2}\right) + \theta_3\left(\frac{1}{2} - \frac{r}{2d}, \frac{i}{2\xi^2d}\right)\theta_3\left(\frac{1}{2} - \frac{r}{2}, \frac{id}{2\xi^2}\right)}{\theta_3\left(0, \frac{i}{2\xi^2d}\right)\theta_3\left(0, \frac{id}{2\xi^2}\right) + \theta_3\left(\frac{1}{2}, \frac{i}{2\xi^2d}\right)\theta_3\left(\frac{1}{2}, \frac{id}{2\xi^2}\right)}\nonumber\\
        & = e^{-\frac{\pi cr^2d^{\alpha - 1}}{2}}\left(1 + \mathcal{O}\left(\frac{e^{-\frac{\pi cd^{1 + \alpha}}{2}\left(1 - \frac{r}{d}\right)^2}}{1 - e^{-2\pi cd^{1 + \alpha}\left(1 - \frac{r}{d}\right)}}\right)\right)\\
        & = 1 - \frac{\pi c r^2}{2}\frac{1}{d^{1 - \alpha}} + \mathcal{O}\left(\frac{1}{d^{2(1 - \alpha)}}\right)\,.
    \end{align}
\begin{proof}
    Let us inject the stated scaling for $\xi^2$ in the above expression.
    \begin{align}
        & \frac{\theta_3\left(\frac{r}{2d}, \frac{i}{2\xi^2d}\right)\theta_3\left(\frac{r}{2}, \frac{id}{2\xi^2}\right) + \theta_3\left(\frac{1}{2} - \frac{r}{2d}, \frac{i}{2\xi^2d}\right)\theta_3\left(\frac{1}{2} - \frac{r}{2}, \frac{id}{2\xi^2}\right)}{\theta_3\left(0, \frac{i}{2\xi^2d}\right)\theta_3\left(0, \frac{id}{2\xi^2}\right) + \theta_3\left(\frac{1}{2}, \frac{i}{2\xi^2d}\right)\theta_3\left(\frac{1}{2}, \frac{id}{2\xi^2}\right)}\nonumber\\
        & = \frac{\theta_3\left(\frac{r}{2d}, \frac{i}{2cd^{1 + \alpha}}\right)\theta_3\left(\frac{r}{2}, \frac{id^{1 - \alpha}}{2c}\right) + \theta_3\left(\frac{1}{2} - \frac{r}{2d}, \frac{i}{2cd^{1 + \alpha}}\right)\theta_3\left(\frac{1}{2} - \frac{r}{2}, \frac{id^{1 - \alpha}}{2c}\right)}{\theta_3\left(0, \frac{i}{2cd^{1 + \alpha}}\right)\theta_3\left(0, \frac{id^{1 - \alpha}}{2c}\right) + \theta_3\left(\frac{1}{2}, \frac{i}{2cd^{1 + \alpha}}\right)\theta_3\left(\frac{1}{2}, \frac{id^{1 - \alpha}}{2c}\right)}\,.
    \end{align}
    Let us now restrict ourselves to the case $r$ even. This means that one can replace $\frac{r}{2} \to 0$ in the first arguments of the $\theta_3$ functions (but not $\frac{r}{2d} \to 0$\,!). The above then reduces to:
    \begin{align}
        & \frac{\theta_3\left(\frac{r}{2d}, \frac{i}{2cd^{1 + \alpha}}\right)\theta_3\left(0, \frac{id^{1 - \alpha}}{2c}\right) + \theta_3\left(\frac{1}{2} - \frac{r}{2d}, \frac{i}{2cd^{1 + \alpha}}\right)\theta_3\left(\frac{1}{2}, \frac{id^{1 - \alpha}}{2c}\right)}{\theta_3\left(0, \frac{i}{2cd^{1 + \alpha}}\right)\theta_3\left(0, \frac{id^{1 - \alpha}}{2c}\right) + \theta_3\left(\frac{1}{2}, \frac{i}{2cd^{1 + \alpha}}\right)\theta_3\left(\frac{1}{2}, \frac{id^{1 - \alpha}}{2c}\right)}\nonumber\\
        & = \frac{\sqrt{2cd^{1 + \alpha}}e^{-\frac{\pi cr^2d^{\alpha - 1}}{2}}\theta_3\left(2icd^{1 + \alpha}\frac{r}{2d}, 2icd^{1 + \alpha}\right)\theta_3\left(0, \frac{id^{1 - \alpha}}{2c}\right) + \theta_3\left(\frac{1}{2} - \frac{r}{2d}, \frac{i}{2cd^{1 + \alpha}}\right)\theta_3\left(\frac{1}{2}, \frac{id^{1 - \alpha}}{2c}\right)}{\sqrt{2cd^{1 + \alpha}}\theta_3\left(0, 2icd^{1 + \alpha}\right)\theta_3\left(0, \frac{id^{1 - \alpha}}{2c}\right) + \theta_3\left(\frac{1}{2}, \frac{i}{2cd^{1 + \alpha}}\right)\theta_3\left(\frac{1}{2}, \frac{id^{1 - \alpha}}{2c}\right)}\,.
    \end{align}
    Combining the estimates in equation \eqref{eq:bound_tail_theta} and proposition \ref{prop:theta3_approx_1}, the above is found to behave as follows as $d \to \infty$:
    \begin{align}
        & \frac{\theta_3\left(\frac{r}{2d}, \frac{i}{2cd^{1 + \alpha}}\right)\theta_3\left(0, \frac{id^{1 - \alpha}}{2c}\right) + \theta_3\left(\frac{1}{2} - \frac{r}{2d}, \frac{i}{2cd^{1 + \alpha}}\right)\theta_3\left(\frac{1}{2}, \frac{id^{1 - \alpha}}{2c}\right)}{\theta_3\left(0, \frac{i}{2cd^{1 + \alpha}}\right)\theta_3\left(0, \frac{id^{1 - \alpha}}{2c}\right) + \theta_3\left(\frac{1}{2}, \frac{i}{2cd^{1 + \alpha}}\right)\theta_3\left(\frac{1}{2}, \frac{id^{1 - \alpha}}{2c}\right)}\nonumber\\
        & \hspace{0.05\textwidth} = e^{-\frac{\pi cr^2d^{\alpha - 1}}{2}}\left(1 + \mathcal{O}\left(\frac{e^{-\frac{\pi cd^{1 + \alpha}}{2}\left(1 - \frac{r}{d}\right)^2}}{1 - e^{-2\pi cd^{1 + \alpha}\left(1 - \frac{r}{d}\right)}}\right)\right)\\
        & = 1 - \frac{\pi c r^2}{2}\frac{1}{d^{1 - \alpha}} + \mathcal{O}\left(\frac{1}{d^{2(1 - \alpha)}}\right)\,.
    \end{align}
    The case $r$ odd is very similar and we omit it.
\end{proof}
\end{proposition}
The following proposition provides a bound on $\left\langle\exp\left(\frac{2\pi in\widetilde{\xi}_J}{\sqrt{d}}\right)\exp\left(\frac{2\pi im\widetilde{\xi}_I}{\sqrt{d}}\right)\right\rangle$ based on the scaling of $\sigma_m$ with respect to $d$.
\begin{proposition}
\label{prop:variance_suppression_kraus_operators}
    Let $m, n$ be two integers, $c > 0$ and $\beta \in \mathbf{R}$. Assume $\sigma_m^2 := cd^{\beta}$. Then for $\beta > 1$,
    \begin{align}
        \frac{\theta_3\left(\frac{i\sigma_m^2}{d}m, \frac{2i\sigma_m^2}{d}\right)\theta_3\left(\frac{i\sigma_m^2}{d}n, \frac{2i\sigma_m^2}{d}\right)}{\theta_3\left(0, \frac{2i\sigma_m^2}{d}\right)^2}e^{-\frac{\pi\sigma_m^2}{d}(m^2 + n^2)} & = \mathcal{O}\left(e^{-\frac{\pi c(m^2 + n^2) d^{\beta - 1}}{2}}\right) \qquad \text{as }\,\, d \to \infty
    \end{align}
    while for $\beta < 1$,
    \begin{align}
        & \frac{\theta_3\left(\frac{i\sigma_m^2}{d}m, \frac{2i\sigma_m^2}{d}\right)\theta_3\left(\frac{i\sigma_m^2}{d}n, \frac{2i\sigma_m^2}{d}\right)}{\theta_3\left(0, \frac{2i\sigma_m^2}{d}\right)^2}e^{-\frac{\pi\sigma_m^2}{d}(m^2 + n^2)}\nonumber\\
        & = e^{-\frac{\pi cd^{\beta - 1}}{2}(m^2 + n^2)}\left(1 + \mathcal{O}\left(e^{-\frac{\pi d^{1 - \beta}}{2c}}\right)\right)\\
        & = 1 - \frac{\pi c(m^2 + n^2)}{2}\frac{1}{d^{1 - \beta}} + \mathcal{O}\left(\frac{1}{d^{2(1 - \beta)}}\right) \qquad\text{as }\,\, d \to \infty\,.
    \end{align}
\begin{proof}
    This results from straightforward analysis, after rewriting
    \begin{align}
        & \frac{\theta_3\left(\frac{i\sigma_m^2}{d}m, \frac{2i\sigma_m^2}{d}\right)\theta_3\left(\frac{i\sigma_m^2}{d}n, \frac{2i\sigma_m^2}{d}\right)}{\theta_3\left(0, \frac{2i\sigma_m^2}{d}\right)^2}e^{-\frac{\pi\sigma_m^2}{d}(m^2 + n^2)}\nonumber\\
        & = \frac{\theta_3\left(\frac{m}{2}, \frac{id}{2\sigma_m^2}\right)\theta_3\left(\frac{n}{2}, \frac{id}{2\sigma_m^2}\right)}{\theta_3\left(0, \frac{id}{2\sigma_m^2}\right)^2}e^{-\frac{\pi\sigma_m^2}{2d}(m^2 + n^2)}\,.
    \end{align}
\end{proof}
\end{proposition}
The last two propositions say that if one wants to make the prefactors of the probabilistic expectation as close to $1$ as possible, one has to take both the width of the Kraus operator and that of the initial state to decrease quickly with $d$. Therefore, if one were to ignore finite-dimensional effects, hence the contribution of the probabilistic expectation (assuming it should be $1$), one may think there is a way to tune the parameters $\xi^2, \sigma_m^2$ so as to make $\left|\left\langle\exp\left(\frac{2\pi in\widetilde{\xi}_J}{\sqrt{d}}\right)\exp\left(\frac{2\pi im\widetilde{\xi}_I}{\sqrt{d}}\right)\right\rangle\right|$ arbitrarily close to $1$. Unfortunately, we show in the next paragraph that this is not possible as for an exceedingly sharp measurement, the probabilistic expectation systematically exhibits a poor scaling that will ruin the improvement of the $\frac{\theta_3\left(\frac{i\sigma_m^2}{d}m, \frac{2i\sigma_m^2}{d}\right)\theta_3\left(\frac{i\sigma_m^2}{d}n, \frac{2i\sigma_m^2}{d}\right)}{\theta_3\left(0, \frac{2i\sigma_m^2}{d}\right)^2}e^{-\frac{\pi\sigma_m^2}{d}(m^2 + n^2)}$ prefactor brought on by decreasing $\sigma_m$.

\subsection{Measured quasi-ideal clock with pseudo-Gaussian Kraus operators and states: scalings}
\label{sec:measured_quasi_ideal_clock_scalings}
After deriving a general formula for the pseudo-correlation function $\left\langle\exp\left(\frac{2\pi in\widetilde{\xi}_J}{\sqrt{d}}\right)\exp\left(\frac{2\pi im\widetilde{\xi}_I}{\sqrt{d}}\right)\right\rangle$, we will restrict ourselves in this section to the case $m = -1, n = 0$. In other words, we will consider the pseudo-variance
\begin{align*}
    & \left\langle\exp\left(-\frac{2\pi i\widetilde{\xi}_I}{\sqrt{d}}\right)\right\rangle\,.
\end{align*}
Starting from the general formula \ref{eq:quasi_ideal_clock_correlation}, this can be written in the form
\begin{align}
    & \left\langle\exp\left(-\frac{2\pi i\widetilde{\xi}_I}{\sqrt{d}}\right)\right\rangle\nonumber\\
    & = \frac{\theta_3\left(\frac{i\sigma_m^2}{d}, \frac{2i\sigma_m^2}{d}\right)}{\theta_3\left(0, \frac{2i\sigma_m^2}{d}\right)}e^{-\frac{\pi\sigma_m^2}{d} - \frac{2\pi i}{d}\sum_{1 \leq j \leq I}\delta_j}\nonumber\\
    & \hspace{0.05\textwidth} \times \frac{\theta_3\left(\frac{1}{2d}, \frac{i}{2\xi^2d}\right)\theta_3\left(\frac{1}{2}, \frac{id}{2\xi^2}\right) + \theta_3\left(\frac{1}{2} - \frac{1}{2d}, \frac{i}{2\xi^2d}\right)\theta_3\left(0, \frac{id}{2\xi^2}\right)}{\theta_3\left(0, \frac{i}{2\xi^2d}\right)\theta_3\left(0, \frac{id}{2\xi^2}\right) + \theta_3\left(\frac{1}{2}, \frac{i}{2\xi^2d}\right)\theta_3\left(\frac{1}{2}, \frac{id}{2\xi^2}\right)}\nonumber\\
    & \hspace{0.05\textwidth} \times \sum_{-\frac{d - 1}{2} \leq p_1, \ldots, p_J \leq \frac{d - 1}{2}}\left(\prod_{\substack{1 \leq j < J\\j \neq I}}\frac{\theta_3\left(\frac{p_{j + 1} - p_j}{d}, \frac{i}{2\sigma_m^2d}\right)}{\sqrt{2\sigma_m^2d}\theta_3\left(0, \frac{2i\sigma_m^2}{d}\right)}\right)\frac{\theta_3\left(\frac{p_{I + 1} - p_I}{d} - \frac{1}{2d}, \frac{i}{2\sigma_m^2d}\right)}{\sqrt{2\sigma_m^2d}e^{-\frac{\pi\sigma_m^2}{2d}}\theta_3\left(-\frac{i\sigma_m^2}{d}, \frac{2i\sigma_m^2}{d}\right)}\nonumber\\
    & \hspace{0.1\textwidth} \times \prod_{1 \leq j \leq I}\left(1 - \left(1 - e^{2\pi i\delta_j\sign(-1))}\right)\mathbf{1}_{|p_j + 1| > \frac{d - 1}{2}}\right)\nonumber\\
    & \hspace{0.1\textwidth} \times \prod_{I + 1 \leq j \leq J}\left(1 - \left(1 - e^{2\pi i\delta_j\sign(0)}\right)\mathbf{1}_{|p_j| > \frac{d - 1}{2}}\right)\nonumber\\
    & \hspace{0.1\textwidth} \times \frac{\theta_3\left(\frac{p_1 - n_0 + 1}{d}, \frac{i}{\xi^2d}\right)\theta_3\left(\frac{p_1 - n _0}{d}, \frac{i}{\xi^2d}\right)}{\frac{d}{2}\left(\theta_3\left(\frac{1}{2d}, \frac{i}{2\xi^2d}\right)\theta_3\left(\frac{1}{2}, \frac{id}{2\xi^2}\right) + \theta_3\left(\frac{1}{2} - \frac{1}{2d}, \frac{i}{2\xi^2d}\right)\theta_3\left(0, \frac{id}{2\xi^2}\right)\right)}\nonumber\\
    & = \frac{\theta_3\left(\frac{i\sigma_m^2}{d}, \frac{2i\sigma_m^2}{d}\right)}{\theta_3\left(0, \frac{2i\sigma_m^2}{d}\right)}e^{-\frac{\pi\sigma_m^2}{d} +\frac{2\pi i}{d}\sum_{1 \leq j \leq I}\delta_j}\nonumber\\
    & \hspace{0.05\textwidth} \times \frac{\theta_3\left(\frac{1}{2d}, \frac{i}{2\xi^2d}\right)\theta_3\left(\frac{1}{2}, \frac{id}{2\xi^2}\right) + \theta_3\left(\frac{1}{2} - \frac{1}{2d}, \frac{i}{2\xi^2d}\right)\theta_3\left(0, \frac{id}{2\xi^2}\right)}{\theta_3\left(0, \frac{i}{2\xi^2d}\right)\theta_3\left(0, \frac{id}{2\xi^2}\right) + \theta_3\left(\frac{1}{2}, \frac{i}{2\xi^2d}\right)\theta_3\left(\frac{1}{2}, \frac{id}{2\xi^2}\right)}\nonumber\\
    & \hspace{0.05\textwidth} \times \sum_{-\frac{d - 1}{2} \leq p_1, \ldots, p_I \leq \frac{d - 1}{2}}\left(\prod_{1 \leq j < I}\frac{\theta_3\left(\frac{p_{j + 1} - p_j}{d}, \frac{i}{2\sigma_m^2d}\right)}{\sqrt{2\sigma_m^2d}\theta_3\left(0, \frac{2i\sigma_m^2}{d}\right)}\right)\prod_{1 \leq j \leq I}\left(1 - \left(1 - e^{-2\pi i\delta_j}\right)\mathbf{1}_{p_j = \frac{d - 1}{2}}\right)\nonumber\\
    & \hspace{0.1\textwidth} \times \frac{\theta_3\left(\frac{p_1 - n_0 + 1}{d}, \frac{i}{\xi^2d}\right)\theta_3\left(\frac{p_1 - n _0}{d}, \frac{i}{\xi^2d}\right)}{\frac{d}{2}\left(\theta_3\left(\frac{1}{2d}, \frac{i}{2\xi^2d}\right)\theta_3\left(\frac{1}{2}, \frac{id}{2\xi^2}\right) + \theta_3\left(\frac{1}{2} - \frac{1}{2d}, \frac{i}{2\xi^2d}\right)\theta_3\left(0, \frac{id}{2\xi^2}\right)\right)}\,.\label{eq:quasi_ideal_clock_pseudo_variance}
\end{align}
One remarks that in this particular case, all jumps of the random walk have identical distribution.

\subsubsection{Case $\delta = \frac{1}{2}$}
\label{sec:IId1}
One will now specialize the analysis further by assuming $\delta_j = \frac{1}{2}$ for all $j$. (This is in some sense the simplest non-trivial case because if all $\delta_j$ are integers, $1 - e^{-2\pi i\delta_j} = 0$ and the probabilistic expectation is $1$.) In the formula for the pseudo-variance above, one will therefore substitute $1 - e^{-2\pi i\delta_j} \to 2$. Actually, for reasons that will become clearer later, it is convenient to make the substitution $1 - e^{-2\pi i\delta_j} \to c_3$ instead, where $c_3$ may be any constant in $[0, 2]$ (though in the present case it will simply be $2$). The general idea will be to first estimate
\begin{align}
\label{eq:pseudo_variance_probabilistic_expectation}
    & \mathbf{E}^{p_1}\left[\prod_{1 \leq j \leq I}\left(1 - c_3\mathbf{1}_{p_j = \frac{d - 1}{2}}\right)\right]\nonumber\\
    & = \sum_{-\frac{d - 1}{2} \leq p_1, \ldots, p_I \leq \frac{d - 1}{2}}\left(\prod_{1 \leq j < I}\frac{\theta_3\left(\frac{p_{j + 1} - p_j}{d}, \frac{i}{2\sigma_m^2d}\right)}{\sqrt{2\sigma_m^2d}\theta_3\left(0, \frac{2i\sigma_m^2}{d}\right)}\right)\prod_{1 \leq j \leq I}\left(1 - c_3\mathbf{1}_{p_j = \frac{d - 1}{2}}\right)
\end{align}
for all $p_1$ (more or less precisely depending on the range of $p_1$ and the scaling of $\sigma_m^2$) and to deduce an approximation for $\left\langle\exp\left(\frac{2\pi im\widetilde{\xi}_I}{\sqrt{d}}\right)\right\rangle$ from these estimates.

The analysis breaks into two cases according to whether $\sigma_m^2$ is ``big'' or ``small'' with respect to $d$. The final result dealing with the first situation is corollary \ref{cor:quasi_ideal_clock_variance_good_scaling}, the final result dealing with the second case is corollary \ref{cor:quasi_ideal_clock_variance_bad_scaling}.

Let us then start by controlling $\mathbf{E}^{p_1}[\cdot]$. In the case where $\sigma_m^2$ is ``big'', this simply involves controlling the marginal distribution of each of the steps of the random walk and applying a union bound to show that the expectation is $1$ up to an exponentially small error in $d$. We start by the lemma which allows to bound the probability of reaching $\frac{d - 1}{2}$ after $j$ steps, provided the starting point $p_1$ of the walk is bounded away from $d$:
\begin{lemma}
\label{lemma:probability_antipode_small}
    Let $t > 0, \gamma > 0, c_2 > 0$ and $-\frac{1}{2} < \beta < 0$. Suppose $\sigma_m^2$ scales in $d$ according to $\sigma_m^2 := c_2d^{\beta}$. Assume $\gamma d > 1$ and $|p_1| \leq \frac{1}{2}(\gamma d - 1)$. Then for all $\eta \in (0, 1)$, $\zeta > 0$, for all integer $j \in \left[2, t\sqrt{d}\right]$:
    \begin{align}
        \mathbf{P}^{p_1}\left[p_j = \frac{d - 1}{2}\right] & \leq 2(1 + \eta)\left(\left(2 + \frac{\pi}{6c_2} + \frac{2c_2}{\pi}\right)td^{1/2 - \beta} + 2d^{-1}\right)e^{-\frac{\pi}{16c_2}d^{1 - \beta}}\nonumber\\
        & \hspace{0.05\textwidth} + \frac{7}{3}\sqrt{c_2}d^{\beta/2 - 1/2}e^{-\frac{\pi c_2}{2t}(1 - \gamma)^2d^{1/2 + \beta}}\,,
    \end{align}
    provided
    \begin{align}
        c_2 & \leq \frac{\pi}{2}\frac{1}{\log\left(\frac{1 + \eta}{\eta}\right)}\,,\\
        d & \geq d_0 \vee d_3\,.
    \end{align}
    A weaker form (i.e. bounding the probability in a power-decreasing function of $d$ instead of an exponentially decreasing one) is:
    \begin{align}
        \mathbf{P}^{p_1}\left[p_j = \frac{d - 1}{2}\right] & \leq \frac{4}{d^{\zeta}}\,,
    \end{align}
    which holds if one has furthermore
    \begin{align}
        d & \geq d_1 \vee d_2 \vee d_4\,,
    \end{align}
    where:
    \begin{align}
        d_0 & := \left(\frac{32c_2(1 + \eta)t}{\pi}\right)^{\frac{1}{1/2 - \beta}}\,,\\
        d_1 & := \max\left\{ 1, \left[\frac{\zeta + 3/2 - \beta}{1 - \beta}\frac{16c_2}{\pi}\left(1 + \frac{1}{a_0}\log\left(\frac{\zeta + 3/2 - \beta}{1 - \beta}\frac{16c_2}{\pi}\right)\right)\right.\right.\nonumber\\
        & \left.\left. \hspace{0.2\textwidth} + \frac{16c_2}{a_0\pi}\log\left(2(1 + \eta)\left(2 + \frac{\pi}{6c_2} + \frac{2c_2}{\pi}\right)t\right)\right] \right\}^{\frac{1}{1 - \beta}}\,,\\
        d_2 & := \max\left\{1, \left[\frac{\zeta}{1 - \beta}\frac{16c_2}{\pi}\left(1 + \frac{1}{a_0}\log\left(\frac{\zeta}{1 - \beta}\frac{16c_2}{\pi}\right)\right)\right.\right.\nonumber\\
        & \left.\left.\hspace{0.2\textwidth} + \frac{16c_2}{a_0\pi}\log\left(2(1 + \eta)\right)\right] \right\}^{\frac{1}{1 - \beta}}\,,\\
        d_3 & := \left(\frac{\log(2)}{2\pi c_2}\frac{t}{\gamma}\right)^{\frac{1}{1/2 + \beta}}\,,\\
        d_4 & := \max\left\{ 1, \left[\frac{\zeta + \beta/2 - 1/2}{1/2 + \beta}\frac{2t}{\pi c_2(1 - \gamma)^2}\left(1 + \frac{1}{a_0}\log\left(\frac{\zeta + \beta/2 - 1/2}{1/2 + \beta}\frac{2t}{\pi c_2(1 - \gamma)^2}\right)\right)\right.\right.\nonumber\\
        & \left.\left.\hspace{0.2\textwidth} + \frac{1}{a_0}\frac{2t}{\pi c_2(1 - \gamma)^2}\log\left(\frac{3}{7}\sqrt{c_2}\right)\right]\right\}^{\frac{1}{1/2 + \beta}}\,.
    \end{align}
\begin{proof}
From proposition \ref{prop:estimate_convolution_theta}, one may write:
\begin{align}
    \mathbf{P}^{p_1}\left[p_j = \frac{d - 1}{2}\right] & = \frac{d^{j - 1}}{(2\sigma_m^2j)^{d/2}\theta_3\left(0, \frac{2i\sigma_m^2}{d}\right)^j}\left(\theta_3\left(\frac{1}{d}\left(\frac{d - 1}{2} - p_1\right), \frac{ij}{2\sigma_m^2d}\right) + \varepsilon\right)\nonumber\\
    & = \frac{1}{d\theta_3\left(0, \frac{id}{2\sigma_m^2}\right)^j}\left(\theta_3\left(\frac{1}{d}\left(\frac{d - 1}{2} - p_1\right), \frac{ij}{2\sigma_m^2d}\right) + \varepsilon\right)\,,\\
    |\varepsilon| & \leq (j - 1)\frac{2e^{-\frac{\pi d}{8\sigma_m^2}}}{1 - e^{-\frac{\pi d}{2\sigma_m^2}}}\left(1 + \frac{2e^{-\frac{\pi(j - 1)}{2\sigma_m^2d}}}{1 - e^{-\frac{\pi(j - 1)}{\sigma_m^2d}}}\right)\left(1 + \frac{2e^{-\frac{\pi}{2\sigma_m^2}}}{1 - e^{-\frac{\pi d}{2\sigma_m^2}}}\right)^{j - 2}\nonumber\\
    & \hspace{0.05\textwidth} + \frac{2e^{-\frac{\pi d}{8\sigma_m^2}}}{1 - e^{-\frac{\pi}{2\sigma_m^2}}}\left(1 + \frac{2e^{-\frac{\pi d}{8\sigma_m^2}}}{1 - e^{-\frac{\pi d}{2\sigma_m^2}}}\right)^{j - 1}  + \frac{2e^{-\frac{\pi jd}{4\sigma_m^2}}}{1 - e^{-\frac{\pi j}{2\sigma_m^2}}}\,.
\end{align}
Let us substitute the scaling $\sigma_m^2 := c_2d^{\beta}$ as well as the inequality $2 \leq j \leq t\sqrt{d}$ into the bound for $\varepsilon$:
\begin{align}
    |\varepsilon| & \leq t\sqrt{d}\frac{2e^{-\frac{\pi}{8c_2}d^{1 - \beta}}}{1 - e^{-\frac{\pi}{2c_2}d^{1 - \beta}}}\left(1 + \frac{2e^{-\frac{\pi}{2c_2}d^{-1 - \beta}}}{1 - e^{-\frac{\pi}{c_2}d^{-1 - \beta}}}\right)\left(1 + \frac{2e^{-\frac{\pi}{2c_2}d^{-\beta}}}{1 - e^{-\frac{\pi}{2c_2}d^{1 - \beta}}}\right)^{t\sqrt{d}}\nonumber\\
    & \hspace{0.05\textwidth} + \frac{2e^{-\frac{\pi}{8c_2}d^{1 - \beta}}}{1 - e^{-\frac{\pi}{2c_2}d^{-\beta}}}\left(1 + \frac{2e^{-\frac{\pi}{8c_2}d^{1 - \beta}}}{1 - e^{-\frac{\pi}{2c_2}d^{1 - \beta}}}\right)^{t\sqrt{d}} + \frac{2e^{-\frac{\pi}{2c_2}d^{1 - \beta}}}{1 - e^{-\frac{\pi}{c_2}d^{-\beta}}}\,.
\end{align}
Now, assume
\begin{align}
    c_2 & \leq \frac{\pi}{2}\frac{1}{\log\left(\frac{1 + \eta}{\eta}\right)}
\end{align}
for some $\eta > 0$, which implies $\frac{1}{1 - e^{-\frac{\pi}{c_2}}} \leq \frac{1}{1 - e^{-\frac{\pi}{2c_2}}} \leq 1 + \eta$. Furthermore, one uses equation \eqref{eq:1_over_1_minus_exp_bounds} to bound:
\begin{align}
    1 + \frac{2e^{-\frac{\pi}{2c_2}d^{-1 - \beta}}}{1 - e^{-\frac{\pi}{c_2}d^{-1 - \beta}}} & \leq 1 + \frac{2}{1 - e^{-\frac{\pi}{c_2}d^{-1 - \beta}}}\\
    & \leq 1 + 2\left(\frac{c_2}{\pi}d^{1 + \beta} + \frac{1}{2} + \frac{\pi}{12c_2}d^{-1 - \beta}\right)\\
    & \leq 2 + \frac{\pi}{6c_2} + \frac{2c_2}{\pi}d^{1 + \beta}\\
    & \leq \left(2 + \frac{\pi}{6c_2} + \frac{2c_2}{\pi}\right)d^{1 + \beta}\,.
\end{align}
This yields:
\begin{align}
    |\varepsilon| & \leq 2(1 + \eta)t\sqrt{d}e^{-\frac{\pi}{8c_2}d^{1 - \beta}}\left(2 + \frac{\pi}{6c_2} + \frac{2c_2}{\pi}\right)d^{1 + \beta}\left(1 + 2(1 + \eta)e^{-\frac{\pi}{2c_2}d^{-\beta}}\right)^{t\sqrt{d}}\nonumber\\
    & \hspace{0.05\textwidth} + 2(1 + \eta)e^{-\frac{\pi}{8c_2}d^{1 - \beta}}\left(1 + 2(1 + \eta)e^{-\frac{\pi}{8c_2}d^{1 - \beta}}\right)^{t\sqrt{d}} + 2(1 + \eta)e^{-\frac{\pi}{2c_2}d^{1 - \beta}}\\
    & \leq 2(1 + \eta)\left(2 + \frac{\pi}{6c_2} + \frac{2c_2}{\pi}\right)td^{3/2 - \beta}e^{-\frac{\pi}{8c_2}d^{1 - \beta}}e^{2(1 + \eta)t\sqrt{d}}\nonumber\\
    & \hspace{0.05\textwidth} + 2(1 + \eta)e^{-\frac{\pi}{8c_2}d^{1 - \beta}}\exp\left(2(1 + \eta)t\sqrt{d}e^{-\frac{\pi}{8c_2}d^{1 - \beta}}\right) + 2(1 + \eta)e^{-\frac{\pi}{2c_2}d^{1 - \beta}}\,.
\end{align}
Now, provided one chooses
\begin{align}
    d & \geq d_0 := \left(\frac{32c_2(1 + \eta)t}{\pi}\right)^{\frac{1}{1/2 - \beta}}\,,
\end{align}
one has $2(1 + \eta)t\sqrt{d} \leq \frac{\pi}{16c_2}d^{1 - \beta}$ and the inequality becomes
\begin{align}
    |\varepsilon| & \leq 2(1 + \eta)\left(2 + \frac{\pi}{6c_2} + \frac{2c_2}{\pi}\right)td^{3/2 - \beta}e^{-\frac{\pi}{16c_2}d^{1 - \beta}} + 2(1 + \eta)e^{-\frac{\pi}{16c_2}d^{1 - \beta}} + 2(1 + \eta)e^{-\frac{\pi}{2c_2}d^{1 - \beta}}\,.
\end{align}
Then, according to equation \eqref{eq:power_vs_exp_inequality}, if one chooses
\begin{align}
    d & \geq d_1 := \max\left\{ 1, \left[\frac{\zeta + 3/2 - \beta}{1 - \beta}\frac{16c_2}{\pi}\left(1 + \frac{1}{a_0}\log\left(\frac{\zeta + 3/2 - \beta}{1 - \beta}\frac{16c_2}{\pi}\right)\right)\right.\right.\nonumber\\
    & \left.\left. \hspace{0.2\textwidth} + \frac{16c_2}{a_0\pi}\log\left(2(1 + \eta)\left(2 + \frac{\pi}{6c_2} + \frac{2c_2}{\pi}\right)t\right)\right]^{\frac{1}{1 - \beta}} \right\}\,,\\
    d & \geq d_2 := \max\left\{1, \left[\frac{\zeta}{1 - \beta}\frac{16c_2}{\pi}\left(1 + \frac{1}{a_0}\log\left(\frac{\zeta}{1 - \beta}\frac{16c_2}{\pi}\right)\right)\right.\right.\nonumber\\
    & \left.\left.\hspace{0.2\textwidth} + \frac{16c_2}{a_0\pi}\log\left(2(1 + \eta)\right)\right]^{\frac{1}{1 - \beta}}\right\}\,
\end{align}
where $\zeta > 0$, one obtains
\begin{align}
    |\varepsilon| & \leq \frac{3}{d^{\zeta}}\,.
\end{align}
It now remains to control the term $\frac{\theta_3\left(\frac{1}{d}\left(\frac{d - 1}{2} - p_1\right), \frac{ij}{2\sigma_m^2d}\right)}{d\theta_3\left(0, \frac{id}{2\sigma_m^2}\right)^j}$. We find
\begin{align}
\label{eq:theta3_ij}
    & \frac{\theta_3\left(\frac{1}{d}\left(\frac{d - 1}{2} - p_1\right), \frac{ij}{2\sigma_m^2d}\right)}{d\theta_3\left(0, \frac{id}{2\sigma_m^2}\right)^j}\nonumber\\
    & = \frac{\theta_3\left(\frac{1}{d}\left(\frac{d - 1}{2} - p_1\right), \frac{ij}{2c_2d^{1 + \beta}}\right)}{d\theta_3\left(0, \frac{id^{1 - \beta}}{2c_2}\right)^j}\nonumber\\
    & \leq \frac{1}{d}\theta_3\left(\frac{1}{2}(1 - \gamma), \frac{ij}{2c_2d^{1 + \beta}}\right)\nonumber\\
    & = \frac{1}{d}\sqrt{\frac{2c_2d^{1 + \beta}}{j}}\exp\left(-\frac{\pi c_2}{2j}(1 - \gamma)^2d^{1 + \beta}\right)\theta_3\left(ic_2(1 - \gamma), \frac{2ic_2d^{1 + \beta}}{j}\right)\nonumber\\
    & \leq \frac{1}{d}\sqrt{\frac{2c_2d^{1 + \beta}}{j}}\exp\left(-\frac{\pi c_2}{2j}(1 - \gamma)^2d^{1 + \beta}\right)\left(1 + \frac{2e^{-\frac{2\pi c_2d^{1 + \beta}\gamma}{j}}}{1 - e^{-\frac{4\pi c_2d^{1 + \beta}\gamma}{j}}}\right)\nonumber\\
    & \leq \sqrt{c_2}d^{\beta/2 - 1/2}\exp\left(-\frac{\pi c_2}{2t}(1 - \gamma)^2d^{1/2 + \beta}\right)\left(1 + \frac{2e^{-\frac{2\pi c_2d^{1/2 + \beta}\gamma}{t}}}{1 - e^{-\frac{4\pi c_2d^{1/2 + \beta}\gamma}{t}}}\right)\,,
\end{align}
where we used proposition \eqref{prop:theta3_approx_1} to obtain the penultimate line.
Now, provided
\begin{align}
    d & \geq d_3 := \left(\frac{\log(2)}{2\pi c_2}\frac{t}{\gamma}\right)^{\frac{1}{1/2 + \beta}}\,,
\end{align}
\eqref{eq:theta3_ij}
can be bounded by:
\begin{align}
    & \frac{7}{3}\sqrt{c_2}d^{\beta/2 - 1/2}\exp\left(-\frac{\pi c_2}{2t}(1 - \gamma)^2d^{1/2 + \beta}\right)\,.
\end{align}
Finally, provided
\begin{align}
\label{eq:bound_sqrt_c2_1_minus_gamma_sq}
    d & \geq d_4 := \max\left\{ 1, \left[\frac{\zeta + \beta/2 - 1/2}{1/2 + \beta}\frac{2t}{\pi c_2(1 - \gamma)^2}\left(1 + \frac{1}{a_0}\log\left(\frac{\zeta + \beta/2 - 1/2}{1/2 + \beta}\frac{2t}{\pi c_2(1 - \gamma)^2}\right)\right)\right.\right.\nonumber\\
    & \left.\left.\hspace{0.2\textwidth} + \frac{1}{a_0}\frac{2t}{\pi c_2(1 - \gamma)^2}\log\left(\frac{7}{3}\sqrt{c_2}\right)\right]^{\frac{1}{1/2 + \beta}}\right\}
\end{align}
\eqref{eq:bound_sqrt_c2_1_minus_gamma_sq} 
can be bounded by:
\begin{align*}
    & \frac{1}{d^{\zeta}}\,.
\end{align*}
\end{proof}
\end{lemma}
Next, one will need an extra technical lemma to bound the tail of the initial distribution
\begin{align*}
    p_1 & \longmapsto \frac{\theta_3\left(\frac{p_1 - n_0 + 1}{d}, \frac{i}{\xi^2d}\right)\theta_3\left(\frac{p_1 - n _0}{d}, \frac{i}{\xi^2d}\right)}{\frac{d}{2}\left(\theta_3\left(\frac{1}{2d}, \frac{i}{2\xi^2d}\right)\theta_3\left(\frac{1}{2}, \frac{id}{2\xi^2}\right) + \theta_3\left(\frac{1}{2} - \frac{1}{2d}, \frac{i}{2\xi^2d}\right)\theta_3\left(0, \frac{id}{2\xi^2}\right)\right)}
\end{align*}

\begin{lemma}
\label{lemma:bound_tail_initial_distribution}
    Let $\gamma \in (0, 1)$ such that $\frac{2}{\gamma d} \leq \frac{1}{2}$ (also suppose $d \geq 2$). Let $\xi^2 := c_1d^{\alpha}$ as usual. Then the following estimate holds:
    \begin{align}
        \sum_{\substack{-\frac{d - 1}{2} \leq p \leq \frac{d - 1}{2}\\|p| > \frac{1}{2}(\gamma d - 1)}}\theta_3\left(\frac{p + 1}{d}, \frac{i}{\xi^2d}\right)\theta_3\left(\frac{p}{d}, \frac{i}{\xi^2d}\right) & \leq \sqrt{c_1}d^{1/2 + \alpha/2}\left(52e^{-\frac{\pi c_1\gamma^2}{16}d^{1 + \alpha}} +  18de^{-\frac{\pi c_1}{8}d^{1 + \alpha}}\right)
    \end{align}
    provided
    \begin{align}
        d & \geq d_0 := \left(\frac{\log(2)}{2\pi c_1}\right)^{\frac{1}{1 + \alpha}}\,,\\
        d & \geq d_1 := \left(\frac{c_1}{\pi}\log(2)\right)^{\frac{1}{1 + \alpha}}\,.
    \end{align}
    Furthermore, if $\alpha < 0$ and
    \begin{align}
        d & \geq d_2 := \left(\frac{\log(2)}{\pi c_1\gamma}\right)^{\frac{1}{\alpha}}\,,\\
        d & \geq d_3 := \max\left\{1, \left[\frac{\zeta + 1/2 + \alpha/2}{1 + \alpha}\frac{16}{\pi c_1\gamma^2}\left(1 + \frac{1}{a_0}\log\left(\frac{\zeta + 1/2 + \alpha/2}{1 + \alpha}\frac{16}{\pi c_1\gamma^2}\right)\right)\right.\right.\nonumber\\
        & \left.\left. \hspace{0.3\textwidth} + \frac{1}{a_0}\frac{16}{\pi c_1\gamma^2}\log\left(52\sqrt{c_1}\right)\right]\right\}^{\frac{1}{1 + \alpha}}\,,\\
        d & \geq d_4 := \max\left\{ 1, \left[\frac{\zeta + 3/2 + \alpha/2}{1 + \alpha}\frac{8}{\pi c_1}\left(1 + \frac{1}{a_0}\log\left(\frac{\zeta + 3/2 + \alpha/2}{1 + \alpha}\frac{8}{\pi c_1}\right)\right)\right.\right.\nonumber\\
        & \left.\left.\hspace{0.3\textwidth} + \frac{1}{a_0}\frac{8}{\pi c_1}\log\left(18\sqrt{c_1}\right)\right]\right\}^{\frac{1}{1 + \alpha}}\,,
    \end{align}
    for some $\zeta > 0$, then
    \begin{align}
        \sum_{\substack{-\frac{d - 1}{2} \leq p \leq \frac{d - 1}{2}\\|p| > \frac{1}{2}(\gamma d - 1)}}\theta_3\left(\frac{p + 1}{d}, \frac{i}{\xi^2d}\right)\theta_3\left(\frac{p}{d}, \frac{i}{\xi^2d}\right) & \leq \frac{2}{d^{\zeta}}\,.
    \end{align}
\begin{proof}
    One starts by using the multiplication formula from proposition \ref{prop:theta3_multiplication} to rewrite the product of $\theta$ functions as:
    \begin{align}
        & \theta_3\left(\frac{p + 1}{d}, \frac{i}{\xi^2d}\right)\theta_3\left(\frac{p}{d}, \frac{i}{\xi^2d}\right)\nonumber\\
        & = \frac{1}{2}\left[\theta_3\left(\frac{1}{2d}, \frac{i}{2\xi^2d}\right)\theta_3\left(\frac{p}{d} + \frac{1}{2d}, \frac{i}{2\xi^2d}\right) + \theta_3\left(\frac{1}{2d} + \frac{1}{2}, \frac{i}{2\xi^2d}\right)\theta_3\left(\frac{p}{d} + \frac{1}{2d} + \frac{1}{2}, \frac{i}{2\xi^2d}\right)\right]\,.
    \end{align}
    One will use the estimate \ref{eq:bound_tail_theta} when $z \in \left(-\frac{1}{2}, \frac{1}{2}\right)$, which implies in this case:
    \begin{align}
        |\theta_3\left(z, i\sigma\right)| & \leq \frac{1}{\sigma^{1/2}}e^{-\frac{\pi}{\sigma}z^2}\left(1 + \frac{2}{1 - e^{-\frac{\pi}{\sigma}}}\right)\,.
    \end{align}
    One now applies it to achieve the bound:
    \begin{align}
        & \sum_{\substack{-\frac{d - 1}{2} \leq p \leq \frac{d - 1}{2}\\|p| > \frac{1}{2}(\gamma d - 1)}}\theta_3\left(\frac{p}{d} + \frac{1}{2d}, \frac{i}{2\xi^2d}\right)\nonumber\\
        & \leq \sqrt{2\xi^2d}\sum_{\substack{-\frac{d - 1}{2} \leq p \leq \frac{d - 1}{2}\\|p| > \frac{1}{2}(\gamma d - 1)}}e^{-2\pi\xi^2d\left(\frac{p}{d} + \frac{1}{2d}\right)^2}\left(1 + \frac{2}{1 - e^{-2\pi\xi^2d}}\right)\\
        & \leq \sqrt{2\xi^2d}\left(\frac{e^{-\frac{\pi\xi^2\gamma^2 d}{4}}}{1 - e^{-2\pi\xi^2\gamma}} + \frac{e^{-\frac{\pi\xi^2\gamma^2 d}{4}\left(1 - \frac{2}{\gamma d}\right)^2}}{1 - e^{-2\pi\xi^2\gamma\left(1 - \frac{2}{\gamma d}\right)}}\right)\left(1 + \frac{2}{1 - e^{-2\pi\xi^2d}}\right)\\
        & \leq \sqrt{2\xi^2d}\frac{2e^{-\frac{\pi\xi^2\gamma^2 d}{4}\left(1 - \frac{2}{\gamma d}\right)^2}}{1 - e^{-2\pi\xi^2\gamma\left(1 - \frac{2}{\gamma d}\right)}}\left(1 + \frac{2}{1 - e^{-2\pi\xi^2d}}\right)\,.
    \end{align}
    As for $\theta_3\left(\frac{p}{d} + \frac{1}{2d} + \frac{1}{2}, \frac{i}{2\xi^2d}\right)$, one may bound it in a cruder way:
    \begin{align}
        & \sum_{\substack{-\frac{d - 1}{2} \leq p \leq \frac{d - 1}{2}\\|p| > \frac{1}{2}(\gamma d - 1)}}\theta_3\left(\frac{p}{d} + \frac{1}{2d} + \frac{1}{2}, \frac{i}{2\xi^2d}\right)\nonumber\\
        & \leq \sum_{-\frac{d - 1}{2} \leq p \leq \frac{d - 1}{2}}\theta_3\left(\frac{p}{d} + \frac{1}{2d} + \frac{1}{2}, \frac{i}{2\xi^2d}\right)\\
        & = \sum_{-\frac{d - 1}{2} \leq p \leq \frac{d - 1}{2}}\theta_3\left(\frac{1}{d}\left(p + \frac{d + 1}{2}\right), \frac{i}{2\xi^2d}\right)\\
        & = \sum_{-\frac{d - 1}{2} \leq p \leq \frac{d - 1}{2}}\theta_3\left(\frac{p}{d}, \frac{i}{2\xi^2d}\right)\\
        & = \sqrt{2\xi^2d}\theta_3\left(0, \frac{2i\xi^2}{d}\right)\\
        & = d\theta_3\left(0, \frac{id}{2\xi^2}\right)\\
        & \leq d\left(1 + \frac{2e^{-\frac{\pi d}{2\xi^2}}}{1 - e^{-\frac{\pi d}{\xi^2}}}\right)\,.
    \end{align}
    One now uses the modular transformation properties of $\theta$ as well as equation \eqref{eq:approx_theta_im_arg} to treat the factors $\theta_3\left(\frac{1}{2d}, \frac{i}{2\xi^2d}\right)$ and $\theta_3\left(\frac{1}{2d} + \frac{1}{2}, \frac{i}{2\xi^2d}\right)$:
    \begin{align}
        \theta_3\left(\frac{1}{2d}, \frac{i}{2\xi^2d}\right) & = \sqrt{2\xi^2d}\exp\left(-\frac{\pi\xi^2}{2d}\right)\theta_3\left(2i\xi^2d\frac{1}{2d}, 2i\xi^2d\right)\\
        & \leq \sqrt{2\xi^2d}\left(1 + \frac{2e^{-2\pi\xi^2d\left(1 - \frac{1}{d}\right)^2}}{1 - e^{-4\pi\xi^2d}}\right)\,,\\
        \theta_3\left(\frac{1}{2d} + \frac{1}{2}, \frac{i}{2\xi^2d}\right) & = \theta_3\left(\frac{1}{2} - \frac{1}{2d}, \frac{i}{2\xi^2d}\right)\\
        & = \sqrt{2\xi^2d}\exp\left(-\frac{\pi\xi^2d}{2}\left(1 - \frac{1}{d}\right)^2\right)\theta_3\left(2i\xi^2d\frac{1}{2}\left(1 - \frac{1}{d}\right), 2i\xi^2d\right)\\
        & \leq \sqrt{2\xi^2d}e^{-\frac{\pi\xi^2d}{2}\left(1 - \frac{1}{d}\right)^2}\left(1 + \frac{2e^{-2\pi\xi^2}}{1 - e^{-4\pi\xi^2d}}\right)\,.
    \end{align}
    One may now use $\frac{2}{\gamma d} \leq \frac{1}{2}$ and $d \geq 2$ to obtain a slightly simplified bound:
    \begin{align}
        & \sum_{\substack{-\frac{d - 1}{2} \leq p \leq \frac{d - 1}{2}\\|p| > \frac{1}{2}(\gamma d - 1)}}\theta_3\left(\frac{p + 1}{d}, \frac{i}{\xi^2d}\right)\theta_3\left(\frac{p}{d}, \frac{i}{\xi^2d}\right)\nonumber\\
        & \leq \sqrt{\frac{\xi^2d}{2}}\left(1 + \frac{2}{1 - e^{-2\pi\xi^2d}}\right)\left[\left(1 + \frac{2e^{-\frac{\pi\xi^2d}{2}}}{1 - e^{-4\pi\xi^2d}}\right)\frac{2e^{-\frac{\pi\xi^2\gamma^2d}{16}}}{1 - e^{-\pi\xi^2\gamma}} + \left(1 + \frac{2e^{-\frac{\pi d}{2\xi^2}}}{1 - e^{-\frac{\pi d}{\xi^2}}}\right)de^{-\frac{\pi\xi^2d}{8}}\right]\\
        & = \sqrt{\frac{c_1}{2}}d^{1/2 + \alpha/2}\left(1 + \frac{2}{1 - e^{-2\pi c_1d^{1 + \alpha}}}\right)\left[\left(1 + \frac{2e^{-\frac{\pi c_1}{2}d^{1 + \alpha}}}{1 - e^{-4\pi c_1d^{1 + \alpha}}}\right)\frac{2e^{-\frac{\pi c_1\gamma^2}{16}d^{1 + \alpha}}}{1 - e^{-\pi c_1\gamma d^{\alpha}}}\right.\nonumber\\
        & \left. \hspace{0.3\textwidth} + \left(1 + \frac{2e^{-\frac{\pi}{2c_1}d^{1 - \alpha}}}{1 - e^{-\frac{\pi}{c_1}d^{1 - \alpha}}}\right)de^{-\frac{\pi c_1}{8}d^{1 + \alpha}}\right]\,.
    \end{align}
    Now, suppose:
    \begin{align}
        d & \geq d_0 := \left(\frac{\log(2)}{2\pi c_1}\right)^{\frac{1}{1 + \alpha}}\,,\\
        d & \geq d_1 := \left(\frac{c_1}{\pi}\log(2)\right)^{\frac{1}{1 + \alpha}}\,.
    \end{align}
    One may then further simplify the bound:
    \begin{align}
        & \sum_{\substack{-\frac{d - 1}{2} \leq p \leq \frac{d - 1}{2}\\|p| > \frac{1}{2}(\gamma d - 1)}}\theta_3\left(\frac{p + 1}{d}, \frac{i}{\xi^2d}\right)\theta_3\left(\frac{p}{d}, \frac{i}{\xi^2d}\right)\nonumber\\
        & \leq \sqrt{c_1}d^{1/2 + \alpha/2}\left(\frac{26}{1 - e^{-\pi c_1\gamma d^{\alpha}}}e^{-\frac{\pi c_1\gamma^2}{16}d^{1 + \alpha}} +  18de^{-\frac{\pi c_1}{8}d^{1 + \alpha}}\right)\,.
    \end{align}
    In the case $\alpha < 0$, choosing
    \begin{align}
        d & \geq d_2 := \left(\frac{\log(2)}{\pi c_1\gamma}\right)^{\frac{1}{\alpha}} 
    \end{align}
    implies
     \begin{align}
        & \sum_{\substack{-\frac{d - 1}{2} \leq p \leq \frac{d - 1}{2}\\|p| > \frac{1}{2}(\gamma d - 1)}}\theta_3\left(\frac{p + 1}{d}, \frac{i}{\xi^2d}\right)\theta_3\left(\frac{p}{d}, \frac{i}{\xi^2d}\right)\nonumber\\
        & \leq \sqrt{c_1}d^{1/2 + \alpha/2}\left(52e^{-\frac{\pi c_1\gamma^2}{16}d^{1 + \alpha}} +  18de^{-\frac{\pi c_1}{8}d^{1 + \alpha}}\right)\,;
    \end{align}
    and provided
    \begin{align}
        d & \geq d_3 := \max\left\{1, \left[\frac{\zeta + 1/2 + \alpha/2}{1 + \alpha}\frac{16}{\pi c_1\gamma^2}\left(1 + \frac{1}{a_0}\log\left(\frac{\zeta + 1/2 + \alpha/2}{1 + \alpha}\frac{16}{\pi c_1\gamma^2}\right)\right)\right.\right.\nonumber\\
        & \left.\left. \hspace{0.3\textwidth} + \frac{1}{a_0}\frac{16}{\pi c_1\gamma^2}\log\left(52\sqrt{c_1}\right)\right]\right\}^{\frac{1}{1 + \alpha}}\,,\\
        d & \geq d_4 := \max\left\{ 1, \left[\frac{\zeta + 3/2 + \alpha/2}{1 + \alpha}\frac{8}{\pi c_1}\left(1 + \frac{1}{a_0}\log\left(\frac{\zeta + 3/2 + \alpha/2}{1 + \alpha}\frac{8}{\pi c_1}\right)\right)\right.\right.\nonumber\\
        & \left.\left.\hspace{0.3\textwidth} + \frac{1}{a_0}\frac{8}{\pi c_1}\log\left(18\sqrt{c_1}\right)\right]\right\}^{\frac{1}{1 + \alpha}}\,,
    \end{align}
    One has indeed
    \begin{align}
        \sum_{\substack{-\frac{d - 1}{2} \leq p \leq \frac{d - 1}{2}\\|p| > \frac{1}{2}(\gamma d - 1)}}\theta_3\left(\frac{p + 1}{d}, \frac{i}{\xi^2d}\right)\theta_3\left(\frac{p}{d}, \frac{i}{\xi^2d}\right) & \leq \frac{2}{d^{\zeta}}\,.
    \end{align}
\end{proof}
\end{lemma}
The following lemma allows one to precisely lower-bound the denominator appearing in the expression for the initial distribution.
\begin{lemma}
\label{lemma:lower_bound_denominator_initial_distribution}
    Let $c_1 > 0$, $-1 < \alpha < 1$ and $\xi^2 := c_1d^{\alpha}$. Then the following lower bound holds:
    \begin{align}
        & \frac{d}{2}\left(\theta_3\left(\frac{1}{2d}, \frac{i}{2\xi^2d}\right)\theta_3\left(\frac{1}{2}, \frac{id}{2\xi^2}\right) + \theta_3\left(\frac{1}{2} - \frac{1}{2d}, \frac{i}{2\xi^2d}\right)\theta_3\left(0, \frac{id}{2\xi^2}\right)\right)\nonumber\\
        & \geq \sqrt{\frac{c_1}{2}}d^{3/2 + \alpha/2}\left(1 - \frac{\pi}{c_1}d^{\alpha - 1}\right)\left(1 - \frac{2e^{-\pi c_1d^{1 + \alpha}}}{1 - e^{-4\pi c_1d^{1 + \alpha}}}\right)\left(1 - \frac{2e^{-\frac{\pi}{2c_1}d^{1 - \alpha}}}{1 - e^{-\frac{\pi}{c_1}d^{1 - \alpha}}}\right)\,.
    \end{align}
\begin{proof}
    Since all the $\theta$ functions are positive, it suffices to lower bound each of them.
    First,
    \begin{align}
        \theta_3\left(\frac{1}{2d}, \frac{i}{2\xi^2d}\right) & = \sqrt{2\xi^2d}\exp\left(-\frac{\pi\xi^2}{2d}\right)\theta_3\left(2i\xi^2d\frac{1}{2d}, 2i\xi^2d\right)\\
        & \geq \sqrt{2\xi^2d}\left(1 - \frac{\pi\xi^2}{2d}\right)\left(1 - \frac{2e^{-2\pi\xi^2d\left(1 - \frac{1}{d}\right)}}{1 - e^{-4\pi\xi^2d}}\right)\\
        & = \sqrt{2c_1}d^{1/2 + \alpha/2}\left(1 - \frac{\pi c_1}{2}d^{\alpha - 1}\right)\left(1 - \frac{2e^{-\pi c_1d^{1 + \alpha}}}{1 - e^{-4\pi c_1d^{1 + \alpha}}}\right)\,.
    \end{align}
    Secondly,
    \begin{align}
        \theta_3\left(\frac{1}{2}, \frac{id}{2\xi^2}\right) & \geq 1 - \frac{2e^{-\frac{\pi d}{2\xi^2}}}{1 - e^{-\frac{\pi d}{\xi^2}}}\\
        & = 1 - \frac{2e^{-\frac{\pi}{2c_1}d^{1 - \alpha}}}{1 - e^{-\frac{\pi}{c_1}d^{1 - \alpha}}}\,.
    \end{align}
    Thirdly,
    \begin{align}
        \theta_3\left(\frac{1}{2} - \frac{1}{2d}, \frac{i}{2\xi^2d}\right) & \geq \sqrt{2\xi^2d}\exp\left(-\frac{\pi\xi^2d}{2}\left(1 - \frac{1}{d}\right)^2\right)\theta_3\left(2i\xi^2d\left(\frac{1}{2} - \frac{1}{2d}\right), \frac{i}{2\xi^2d}\right)\\
        & \geq \sqrt{2\xi^2d}\exp\left(-\frac{\pi\xi^2d}{2}\left(1 - \frac{1}{d}\right)^2\right)\left(1 - \frac{2e^{-2\pi\xi^2}}{1 - e^{-4\pi\xi^2}}\right)\\
        & = \sqrt{2c_1}d^{1/2 + \alpha/2}\exp\left(-\frac{\pi c_1}{2}d^{1 + \alpha}\left(1 - \frac{1}{d}\right)^2\right)\left(1 - \frac{2e^{-2\pi c_1d^{\alpha}}}{1 - e^{-4\pi c_1d^{\alpha}}}\right)\,.
    \end{align}
    Fourthly,
    \begin{align}
        \theta_3\left(0, \frac{id}{2\xi^2}\right) & \geq1 - \frac{2e^{-\frac{\pi d}{2\xi^2}}}{1 - e^{-\frac{\pi d}{\xi^2}}}\\
        & \geq 1 - \frac{2e^{-\frac{\pi}{2c_1}d^{1 - \alpha}}}{1 - e^{-\frac{\pi}{c_1}d^{1 - \alpha}}}\,.
    \end{align}
    Recalling $-1 < \alpha < 1$, only the first two terms are relevant and one can then write:
    \begin{align}
        & \frac{d}{2}\left(\theta_3\left(\frac{1}{2d}, \frac{i}{2\xi^2d}\right)\theta_3\left(\frac{1}{2}, \frac{id}{2\xi^2}\right) + \theta_3\left(\frac{1}{2} - \frac{1}{2d}, \frac{i}{2\xi^2d}\right)\theta_3\left(0, \frac{id}{2\xi^2}\right)\right)\nonumber\\
        & \geq \sqrt{\frac{c_1}{2}}d^{3/2 + \alpha/2}\left(1 - \frac{\pi}{c_1}d^{\alpha - 1}\right)\left(1 - \frac{2e^{-\pi c_1d^{1 + \alpha}}}{1 - e^{-4\pi c_1d^{1 + \alpha}}}\right)\left(1 - \frac{2e^{-\frac{\pi}{2c_1}d^{1 - \alpha}}}{1 - e^{-\frac{\pi}{c_1}d^{1 - \alpha}}}\right)\,.
    \end{align}
\end{proof}
\end{lemma}
Putting all the last lemmas together, one can now precisely bound the probabilistic expectation in case $\sigma_m$ is ``big enough'' with respect to $d$. Doing so results in the following proposition.
\begin{proposition}
\label{prop:probabilistic_expectation_exponentially_suppressed}
    Let $c_1, c_2 > 0$, $-1 < \alpha < 1, -\frac{1}{2} < \beta < 0$. Suppose as usual that $\xi^2$ and $\sigma_m^2$ scale with $d$ according to $\xi^2 := c_1d^{\alpha}, \sigma_m^2 := c_2d^{\beta}$. Let $\gamma \in \left(\frac{1}{2}, 1\right)$ such that $\gamma > \frac{1}{2} \vee \frac{2|n_0| + 1}{d}$. Then as $d \to \infty$,
    \begin{align}
        & \sum_{-\frac{d - 1}{2} \leq p_1 \leq \frac{d - 1}{2}}\mathbf{E}^{p_1}\left[\prod_{1 \leq j \leq I}\left(1 - c_3\mathbf{1}_{p_j = \frac{d - 1}{2}}\right)\right]\nonumber\\
        & \hspace{0.1\textwidth} \times \frac{\theta_3\left(\frac{p_1 - n_0 + 1}{d}, \frac{i}{\xi^2d}\right)\theta_3\left(\frac{p_1 - n _0}{d}, \frac{i}{\xi^2d}\right)}{\frac{d}{2}\left(\theta_3\left(\frac{1}{2d}, \frac{i}{2\xi^2d}\right)\theta_3\left(\frac{1}{2}, \frac{id}{2\xi^2}\right) + \theta_3\left(\frac{1}{2} - \frac{1}{2d}, \frac{i}{2\xi^2d}\right)\theta_3\left(0, \frac{id}{2\xi^2}\right)\right)}\nonumber\\
        & = 1 - \mathcal{O}\left(d^{\beta/2 + 1/2}e^{-\frac{\pi c_2(1 - \gamma)^2}{2t}d^{1/2 + \beta}} + d^{-1}e^{-\frac{\pi c_1\gamma^2}{16}d^{1 + \alpha}}\right)\,.
    \end{align}
\begin{proof}
    Assume for simplicity $n_0 > 0$. Choose $\gamma \in \left(\frac{1}{2}, 1\right)$ such that $n_0 < \frac{1}{2}(\gamma d - 1)$. One will estimate
    \begin{align*}
        & \sum_{-\frac{d - 1}{2} \leq p_1 \leq \frac{d - 1}{2}}\mathbf{E}^{p_1}\left[\prod_{1 \leq j \leq I}\left(1 - c_3\mathbf{1}_{p_j = \frac{d - 1}{2}}\right)\right]\nonumber\\
        & \hspace{0.1\textwidth} \times \frac{\theta_3\left(\frac{p_1 - n_0 + 1}{d}, \frac{i}{\xi^2d}\right)\theta_3\left(\frac{p_1 - n _0}{d}, \frac{i}{\xi^2d}\right)}{\frac{d}{2}\left(\theta_3\left(\frac{1}{2d}, \frac{i}{2\xi^2d}\right)\theta_3\left(\frac{1}{2}, \frac{id}{2\xi^2}\right) + \theta_3\left(\frac{1}{2} - \frac{1}{2d}, \frac{i}{2\xi^2d}\right)\theta_3\left(0, \frac{id}{2\xi^2}\right)\right)}
    \end{align*}
    by distinguishing the $p_1$ according to whether $|p_1| \leq \frac{1}{2}(\gamma d - 1)$ or $|p_1| > \frac{1}{2}(\gamma d - 1)$. In the former case, lemma \ref{lemma:probability_antipode_small} tells us that $\mathbf{P}^{p_1}\left[p_j = \frac{1}{2}\right]$ is small for every $j$ so that $\mathbf{E}^{p_1}\left[\ldots\right]$ is also small by a union bound. In the latter case, we will only use that the expectation is bounded by $1$ but according to lemma \ref{lemma:bound_tail_initial_distribution}, one will now be considering the tail of the initial distribution and therefore get that this contribution is also small.
    
    More precisely, for $|p_1| \leq \frac{1}{2}(\gamma d - 1)$, lemma \ref{lemma:probability_antipode_small} allows to conclude:
    \begin{align}
        & \sum_{\substack{-\frac{d - 1}{2} \leq p_1 \leq \frac{d - 1}{2}\\|p_1| \leq \frac{1}{2}(\gamma d - 1)}}\mathbf{E}^{p_1}\left[\prod_{1 \leq j \leq I}\left(1 - c_3\mathbf{1}_{p_j = \frac{d - 1}{2}}\right)\right]\nonumber\\
        & \hspace{0.1\textwidth} \times \frac{\theta_3\left(\frac{p_1 - n_0 + 1}{d}, \frac{i}{\xi^2d}\right)\theta_3\left(\frac{p_1 - n _0}{d}, \frac{i}{\xi^2d}\right)}{\frac{d}{2}\left(\theta_3\left(\frac{1}{2d}, \frac{i}{2\xi^2d}\right)\theta_3\left(\frac{1}{2}, \frac{id}{2\xi^2}\right) + \theta_3\left(\frac{1}{2} - \frac{1}{2d}, \frac{i}{2\xi^2d}\right)\theta_3\left(0, \frac{id}{2\xi^2}\right)\right)}\nonumber\\
        & \geq \sum_{\substack{-\frac{d - 1}{2} \leq p_1 \leq \frac{d - 1}{2}\\|p_1| \leq \frac{1}{2}(\gamma d - 1)}}\left(1 - c_3\sum_{1 \leq j \leq I}\mathbf{P}^{p_1}\left[p_j = \frac{d - 1}{2}\right]\right)\nonumber\\
        & \hspace{0.1\textwidth} \times \frac{\theta_3\left(\frac{p_1 - n_0 + 1}{d}, \frac{i}{\xi^2d}\right)\theta_3\left(\frac{p_1 - n _0}{d}, \frac{i}{\xi^2d}\right)}{\frac{d}{2}\left(\theta_3\left(\frac{1}{2d}, \frac{i}{2\xi^2d}\right)\theta_3\left(\frac{1}{2}, \frac{id}{2\xi^2}\right) + \theta_3\left(\frac{1}{2} - \frac{1}{2d}, \frac{i}{2\xi^2d}\right)\theta_3\left(0, \frac{id}{2\xi^2}\right)\right)}\\
        & \geq \sum_{\substack{-\frac{d - 1}{2} \leq p_1 \leq \frac{d - 1}{2}\\|p_1| \leq \frac{1}{2}(\gamma d - 1)}}\left(1 - d\mathcal{O}\left((td^{1/2 - \beta} + d^{-1})e^{-\frac{\pi}{16c_2}d^{1 - \beta}} + d^{\beta/2 - 1/2}e^{-\frac{\pi c_2(1 - \gamma)^2}{2t}d^{1/2 + \beta}}\right)\right)\nonumber\\
        & \hspace{0.05\textwidth} \times \frac{\theta_3\left(\frac{p_1 - n_0 + 1}{d}, \frac{i}{\xi^2d}\right)\theta_3\left(\frac{p_1 - n _0}{d}, \frac{i}{\xi^2d}\right)}{\frac{d}{2}\left(\theta_3\left(\frac{1}{2d}, \frac{i}{2\xi^2d}\right)\theta_3\left(\frac{1}{2}, \frac{id}{2\xi^2}\right) + \theta_3\left(\frac{1}{2} - \frac{1}{2d}, \frac{i}{2\xi^2d}\right)\theta_3\left(0, \frac{id}{2\xi^2}\right)\right)}\\
        & \geq \left(1 - d\mathcal{O}\left(d^{\beta/2 - 1/2}e^{-\frac{\pi c_2(1 - \gamma)^2}{2t}d^{1/2 + \beta}}\right)\right)\nonumber\\
        & \hspace{0.05\textwidth} \times \sum_{\substack{-\frac{d - 1}{2} \leq p_1 \leq \frac{d - 1}{2}\\|p_1| \leq \frac{1}{2}(\gamma d - 1)}}\frac{\theta_3\left(\frac{p_1 - n_0 + 1}{d}, \frac{i}{\xi^2d}\right)\theta_3\left(\frac{p_1 - n _0}{d}, \frac{i}{\xi^2d}\right)}{\frac{d}{2}\left(\theta_3\left(\frac{1}{2d}, \frac{i}{2\xi^2d}\right)\theta_3\left(\frac{1}{2}, \frac{id}{2\xi^2}\right) + \theta_3\left(\frac{1}{2} - \frac{1}{2d}, \frac{i}{2\xi^2d}\right)\theta_3\left(0, \frac{id}{2\xi^2}\right)\right)}\,.
    \end{align}
    Now, lemmas \ref{lemma:bound_tail_initial_distribution} and \ref{lemma:lower_bound_denominator_initial_distribution} allow one to conclude
    \begin{align}
        & \sum_{\substack{-\frac{d - 1}{2} \leq p_1 \leq \frac{d - 1}{2}\\|p_1| \leq \frac{1}{2}(\gamma d - 1)}}\frac{\theta_3\left(\frac{p_1 - n_0 + 1}{d}, \frac{i}{\xi^2d}\right)\theta_3\left(\frac{p_1 - n _0}{d}, \frac{i}{\xi^2d}\right)}{\frac{d}{2}\left(\theta_3\left(\frac{1}{2d}, \frac{i}{2\xi^2d}\right)\theta_3\left(\frac{1}{2}, \frac{id}{2\xi^2}\right) + \theta_3\left(\frac{1}{2} - \frac{1}{2d}, \frac{i}{2\xi^2d}\right)\theta_3\left(0, \frac{id}{2\xi^2}\right)\right)}\nonumber\\
        & \geq \sum_{\substack{-\frac{d - 1}{2} \leq p_1 \leq \frac{d - 1}{2}}}\frac{\theta_3\left(\frac{p_1 - n_0 + 1}{d}, \frac{i}{\xi^2d}\right)\theta_3\left(\frac{p_1 - n _0}{d}, \frac{i}{\xi^2d}\right)}{\frac{d}{2}\left(\theta_3\left(\frac{1}{2d}, \frac{i}{2\xi^2d}\right)\theta_3\left(\frac{1}{2}, \frac{id}{2\xi^2}\right) + \theta_3\left(\frac{1}{2} - \frac{1}{2d}, \frac{i}{2\xi^2d}\right)\theta_3\left(0, \frac{id}{2\xi^2}\right)\right)}\nonumber\\
        & \hspace{0.05\textwidth} - \left|\mathcal{O}\left(d^{-1}e^{-\frac{\pi c_1\gamma^2}{16}d^{1 + \alpha}} + e^{-\frac{\pi c_1}{8}d^{1 + \alpha}}\right)\right|\\
        & = 1 - \left|\mathcal{O}\left(d^{-1}e^{-\frac{\pi c_1\gamma^2}{16}d^{1 + \alpha}} + e^{-\frac{\pi c_1}{8}d^{1 + \alpha}}\right)\right|\,.
    \end{align}
    This leads to
    \begin{align}
        & \sum_{\substack{-\frac{d - 1}{2} \leq p_1 \leq \frac{d - 1}{2}\\|p_1| \leq \frac{1}{2}(\gamma d - 1)}}\mathbf{E}^{p_1}\left[\prod_{1 \leq j \leq I}\left(1 - c_3\mathbf{1}_{p_j = \frac{d - 1}{2}}\right)\right]\nonumber\\
        & \hspace{0.1\textwidth} \times \frac{\theta_3\left(\frac{p_1 - n_0 + 1}{d}, \frac{i}{\xi^2d}\right)\theta_3\left(\frac{p_1 - n _0}{d}, \frac{i}{\xi^2d}\right)}{\frac{d}{2}\left(\theta_3\left(\frac{1}{2d}, \frac{i}{2\xi^2d}\right)\theta_3\left(\frac{1}{2}, \frac{id}{2\xi^2}\right) + \theta_3\left(\frac{1}{2} - \frac{1}{2d}, \frac{i}{2\xi^2d}\right)\theta_3\left(0, \frac{id}{2\xi^2}\right)\right)}\\
        & \geq 1 - \mathcal{O}\left(d^{\beta/2 + 1/2}e^{-\frac{\pi c_2(1 - \gamma)^2}{2t}d^{1/2 + \beta} + d^{-1}e^{-\frac{\pi c_1\gamma^2}{16}d^{1 + \alpha}}} + e^{-\frac{\pi c_1}{8}d^{1 + \alpha}}\right)\,.
    \end{align}
    Now, for the large $|p_1|$ contribution,
    \begin{align}
        & \sum_{\substack{-\frac{d - 1}{2} \leq p_1 \leq \frac{d - 1}{2}\\|p_1| > \frac{1}{2}(\gamma d - 1)}}\mathbf{E}^{p_1}\left[\prod_{1 \leq j \leq I}\left(1 - c_3\mathbf{1}_{p_j = \frac{d - 1}{2}}\right)\right]\nonumber\\
        & \hspace{0.1\textwidth} \times \frac{\theta_3\left(\frac{p_1 - n_0 + 1}{d}, \frac{i}{\xi^2d}\right)\theta_3\left(\frac{p_1 - n _0}{d}, \frac{i}{\xi^2d}\right)}{\frac{d}{2}\left(\theta_3\left(\frac{1}{2d}, \frac{i}{2\xi^2d}\right)\theta_3\left(\frac{1}{2}, \frac{id}{2\xi^2}\right) + \theta_3\left(\frac{1}{2} - \frac{1}{2d}, \frac{i}{2\xi^2d}\right)\theta_3\left(0, \frac{id}{2\xi^2}\right)\right)}\\
        & = \mathcal{O}\left(d^{-1}e^{-\frac{\pi c_1\gamma^2}{16}d^{1 + \alpha}} + e^{-\frac{\pi c_1}{8}d^{1 + \alpha}}\right)\,.
    \end{align}
    All in all,
    \begin{align}
        & \sum_{-\frac{d - 1}{2} \leq p_1 \leq \frac{d - 1}{2}}\mathbf{E}^{p_1}\left[\prod_{1 \leq j \leq I}\left(1 - c_3\mathbf{1}_{p_j = \frac{d - 1}{2}}\right)\right]\nonumber\\
        & \hspace{0.1\textwidth} \times \frac{\theta_3\left(\frac{p_1 - n_0 + 1}{d}, \frac{i}{\xi^2d}\right)\theta_3\left(\frac{p_1 - n _0}{d}, \frac{i}{\xi^2d}\right)}{\frac{d}{2}\left(\theta_3\left(\frac{1}{2d}, \frac{i}{2\xi^2d}\right)\theta_3\left(\frac{1}{2}, \frac{id}{2\xi^2}\right) + \theta_3\left(\frac{1}{2} - \frac{1}{2d}, \frac{i}{2\xi^2d}\right)\theta_3\left(0, \frac{id}{2\xi^2}\right)\right)}\\
        & \geq 1 - \mathcal{O}\left(d^{\beta/2 + 1/2}e^{-\frac{\pi c_2(1 - \gamma)^2}{2t}d^{1/2 + \beta}} + d^{-1}e^{-\frac{\pi c_1\gamma^2}{16}d^{1 + \alpha}}\right)\,.
    \end{align}
\end{proof}
\end{proposition}

To summarize and come back to the initial problem of estimating the pseudo-variance, we have proved:
\begin{corollary}
\label{cor:quasi_ideal_clock_variance_good_scaling}
    Let $c_1, c_2 > 0$, $-1 < \alpha < 1$ and $-\frac{1}{2} < \beta < 0$. Assume the scalings $\xi^2 := c_1d^{\alpha}$ and $\sigma_m^2 := c_2d^{\beta}$ and choose $\gamma \in \left(\frac{1}{2}, 1\right)$ such that $\gamma\frac{d}{2} - \frac{1}{2} > |n_0|$.\footnote{$n_0$ being allowed to scale with $d$. However, the constraint on $\gamma$ essentially means that $\frac{|n_0|}{d}$ should remain bounded by some constant $< 1$ in the process.} Finally, assume the integer $I$ satisfies $2 \leq I \leq t\sqrt{d}$ for some $t > 0$. One then has as $d \to \infty$,
    \begin{align}\label{eq:seudo var coll}
        & \left\langle\exp\left(-\frac{2\pi i\widetilde{\xi}_I}{\sqrt{d}}\right)\right\rangle\nonumber\\
        & = e^{-\frac{\pi c_1}{2}d^{\alpha - 1} - \frac{\pi c_2}{2}d^{\beta - 1}}\left(1 + \mathcal{O}\left(d^{\beta/2 + 1/2}e^{-\frac{\pi c_2(1 - \gamma)^2}{2t}d^{1/2 + \beta}} + d^{-1}e^{-\frac{\pi c_1\gamma^2}{16}d^{1 + \alpha}}\right)\right)\,.
    \end{align}
\begin{proof}
    This follows straightaway from propositions \ref{prop:variance_suppression_wavefunction}, \ref{prop:variance_suppression_kraus_operators} and \ref{prop:probabilistic_expectation_exponentially_suppressed}.
\end{proof}
\end{corollary}

We have therefore just proved that if $\sigma_m^2$ scales with $d$ such that it vanishes as $d \to \infty$, but not faster than $\frac{1}{\sqrt{d}}$, then the pseudo-variance in equation \eqref{eq:seudo var coll} behaves essentially as
\begin{align}
    & 1 - \frac{\pi c_1}{2}\frac{1}{d^{1 - \alpha}} - \frac{\pi c_2}{d}\frac{1}{d^{1 - \beta}}\,.  
\end{align}
Therefore, it deviates from its ``infinite-dimensional value'' 1 (cf. discussion of quantum measurements in infinite dimension in the main text)
by an error which scales as a power of $d$. This error contains a contribution both from the dispersion of the initial state (exponent $\alpha$) and from the imprecision of the measurement (exponent $\beta$). To connect it to the general theory of measurement exposed in the main text (and in greater detail in \cite[Chapter 5]{braginsky_quantum_1992}), one may say that \textit{the backaction contribution is exponentially suppressed}. By the conditions of application of the last proposition, one is allowed to make the error contributed by $\alpha$ as small as $\frac{1}{d^2}$ but one cannot make the term depending on $\beta$ scale better than $\frac{1}{d^{3/2}}$ since one assumed $\beta > -\frac{1}{2}$. One may then wonder whether one could not obtain a better scaling by choosing $\beta < -\frac{1}{2}$. The purpose of the following is to show that doing this will yield an error essentially as bad (i.e. of order $\frac{1}{\sqrt{d}}$) as the one obtained from measuring the clock in the time basis (a limiting case which was studied in section \ref{sec:quasi_ideal_clock_time_basis} and formally corresponds to letting $\sigma_m^2 \downarrow 0$ in the more general framework considered here).

Roughly speaking, the key idea is that the variance for a jump of the random walk under study is approximately $d^2\frac{1}{2\sigma_m^2d} = \frac{1}{2c_2}d^{1 - \beta}$. Therefore, one expects the variance of the marginal distribution of the $j^\textrm{th}$ step to scale as $jd^{1 - \beta}$. If now $j$ scales as $\sqrt{d}$, this becomes $d^{3/2 - \beta}$. One therefore sees that whenever $\beta < -\frac{1}{2}$, this grows faster than $d^2$ and one therefore expects the marginal distribution of the step to be close to uniform. This is the first ingredient of the proof, treated in the first lemma. However, this is not sufficient since although the steps may all be close to uniformly distributed beyond a certain number of iterations, they are not independent, precluding a priori an evaluation of the expectation. This point will be addressed in the second lemma.

For the following two proofs, it will be particularly helpful to write the distribution for a jump of the random walk under consideration,
\begin{align}
    q & \longmapsto \frac{\theta_3\left(\frac{q}{d}, \frac{i}{2\sigma_m^2d}\right)}{\sqrt{2\sigma_m^2d}\theta_3\left(0, \frac{2i\sigma_m^2}{d}\right)}\,,
\end{align}
as a discrete Fourier transform:
\begin{align}
    \frac{\theta_3\left(\frac{q}{d}, \frac{i}{2\sigma_m^2d}\right)}{\sqrt{2\sigma_m^2d}\theta_3\left(0, \frac{2i\sigma_m^2}{d}\right)} & = \frac{1}{d}\sum_{-\frac{d - 1}{2} \leq s \leq \frac{d - 1}{2}}\frac{\theta_3\left(\frac{s}{d}, \frac{2i\sigma_m^2}{d}\right)}{\theta_3\left(0, \frac{2i\sigma_m^2}{d}\right)}\exp\left(-\frac{2\pi isq}{d}\right)\\
    \label{eq:random_walk_dft_tilde}
    & = \frac{1}{d}\sum_{-\frac{d - 1}{2} \leq s \leq \frac{d - 1}{2}}\widetilde{\theta}_3\left(\frac{s}{d}, \frac{2i\sigma_m^2}{d}\right)\exp\left(-\frac{2\pi isq}{d}\right)\,,
\end{align}
where we defined $\widetilde{\theta}_3(z, i\tau) := \frac{\theta_3(z, i\tau)}{\theta_3(0, i\tau)}$ (therefore, $\widetilde{\theta}_3(0, i\tau) = 1$ and $\widetilde{\theta}_3(z, i\tau) < 1$ for $z \neq 0\,[1]$.
We therefore start to show that above a certain number of iterations scaling as $\sqrt{d}$, the random walk becomes (exponentially) close to completely mixed.
\begin{lemma}
\label{lemma:random_walk_small_sigma_m_marginal_distributions}
    Let $\beta < -\frac{1}{2}$ ($\sigma_m^2 := c_2d^{\beta}$) and $j \geq 1$. Then the variation distance between the marginal distribution for the $j^{\textrm{th}}$ step of the random walk under consideration is upper bounded by:
    \begin{align}
        & \frac{\sqrt{d}}{2}e^{-\frac{\pi}{2c_2}d^{-1 - \beta}j + \frac{\pi^3c_2}{12}d^{\beta - 3}j}\,.
    \end{align}
    In particular, if $j \geq t\sqrt{d}$ with $t > 0$, this yields:
    \begin{align}
        \textrm{variation distance} & \leq \frac{\sqrt{d}}{2}e^{-\frac{\pi}{2c_2}d^{-\frac{1}{2} - \beta}t + \frac{\pi^3c_2}{12}d^{\beta - \frac{5}{2}}t}\,.
    \end{align}
\begin{proof}
    Let $j \geq 1$. To show that the distribution of $p_j$ is close to uniform, we will bound its variation distance (denoted here by $\lVert .\rVert$) to the uniform distribution. To achieve this, we use the following bound for the variation distance between two probability distributions $P, Q$ on $\mathbf{Z}_d$ (for a very general exposition, including more general finite groups than $\mathbf{Z}_d$, see \cite{diaconis_group_1988}):
    \begin{align}
        \lVert P - Q \rVert^2 & \leq \frac{1}{4}d\sum_{-\frac{d - 1}{2} \leq s \leq \frac{d - 1}{2}}|\Hat{P}(s) - \Hat{Q}(s)|^2\,,
    \end{align}
    where $\hat{P}, \hat{Q}$ denote the discrete Fourier transforms of $P, Q$. This yields:
    \begin{align}
        \textrm{variation distance}^2 & \leq         \frac{1}{4}\sum_{\substack{-\frac{d - 1}{2} \leq s \leq \frac{d - 1}{2}\\s \neq 0}}\widetilde{\theta}_3\left(\frac{s}{d}, \frac{2i\sigma_m^2}{d}\right)^{2j}\\
        & = \frac{1}{2}\sum_{1 \leq s \leq \frac{d - 1}{2}}\widetilde{\theta}_3\left(\frac{s}{d}, \frac{2i\sigma_m^2}{d}\right)^{2j}\\
        & = \frac{1}{2}\sum_{1 \leq s \leq \frac{d - 1}{2}}\widetilde{\theta}_3\left(\frac{s}{d}, 2ic_2d^{\beta - 1}\right)^{2j}\,.
    \end{align}
    A convenient way to bound the summand is by way of the Jacobi triple product formula \ref{eq:jacobi_triple_product}. The latter implies:
    \begin{align}
        & \log\left(\widetilde{\theta}_3\left(\frac{s}{d}, 2ic_2d^{\beta - 1}\right)^{2j}\right)\nonumber\\
        & = 2j\log\widetilde{\theta}_3\left(\frac{s}{d}, 2ic_2d^{\beta - 1}\right)\\
        & = 2j\sum_{p \geq 1}\log\left(\frac{1 + 2\cos\left(2\pi\frac{s}{d}\right)e^{-(2p - 1)2\pi c_2d^{\beta - 1}} + e^{-(4p - 2)2\pi c_2d^{\beta - 1}}}{1 + 2e^{-(2p - 1)2\pi c_2d^{\beta - 1}} + e^{-(4p - 2)2\pi c_2d^{\beta - 1}}}\right)\\
        & = 2j\sum_{p \geq 1}\log\left(1 - \frac{2\left(1 - \cos\left(2\pi\frac{s}{d}\right)\right)e^{-(2p - 1)2\pi c_2d^{\beta - 1}}}{\left(1 + e^{-(2p - 1)2\pi c_2d^{\beta - 1}}\right)^2}\right)\\
        & \leq -4j\sum_{p \geq 1}\frac{\left(1 - \cos\left(2\pi\frac{s}{d}\right)\right)e^{-(2p - 1)2\pi c_2d^{\beta - 1}}}{\left(1 + e^{-(2p - 1)2\pi c_2d^{\beta - 1}}\right)^2}\,.
    \end{align}
    One may now use the estimate $\cos(2\pi x) \leq \exp\left(-\frac{(2\pi x)^2}{2}\right)$  valid for all $x \in \left[-\frac{1}{2}, \frac{1}{2}\right]$ to bound the above as:
    \begin{align}
        & \log\left(\widetilde{\theta}_3\left(\frac{s}{d}, 2ic_2d^{\beta - 1}\right)^{2j}\right)\nonumber\\
        & \leq -4j\left(1 - e^{-\frac{2\pi^2}{d^2}s^2}\right)\sum_{p \geq 1}\frac{e^{-(2p - 1)2\pi c_2d^{\beta - 1}}}{\left(1 + e^{-(2p - 1)2\pi c_2d^{\beta - 1}}\right)^2}\\
        & \leq -4j\left(1 - e^{-\frac{2\pi^2}{d^2}s^2}\right)\frac{e^{-2\pi c_2d^{\beta - 1}}}{1 - e^{-4\pi c_2d^{\beta - 1}}}\,.
    \end{align}
    One can now use the inequality \ref{eq:exp_over_1_minus_expr_bound} to write
    \begin{align}
        \frac{e^{-2\pi c_2d^{\beta - 1}}}{1 - e^{-4\pi c_2d^{\beta - 1}}} & \geq \frac{1}{4\pi c_2}d^{1 - \beta} - \frac{\pi c_2}{6}d^{\beta - 1}\,.
    \end{align}
    As for the prefactor containing $s$, one may use the crude bound
    \begin{align}
        1 - e^{-\frac{2\pi^2}{d^2}s^2} & \geq 1 - e^{-\frac{2\pi^2}{d^2}}\\
        & \geq \frac{\pi^2}{d^2}\,,
    \end{align}
    since $1 - e^{-x} \geq \frac{x}{2}$ for $x \in \left(0, \frac{1}{2}\right)$ for instance (which implies that the above holds for $d \geq 7$). Therefore,
    \begin{align}
        & \log\left(\widetilde{\theta}_3\left(\frac{s}{d}, 2ic_2d^{\beta - 1}\right)^{2j}\right)\nonumber\\
        & \leq - 4j\left(\frac{\pi}{4c_2}d^{-1 - \beta} - \frac{\pi^3 c_2}{6}d^{\beta - 3}\right)\,.
    \end{align}
    This entails:
    \begin{align}
        \textrm{variation distance}^2 & \leq \frac{d}{4}e^{-\frac{\pi}{c_2}d^{-1 - \beta}j + \frac{\pi^3c_2}{6}d^{\beta - 3}j}\,.
    \end{align}
    In particular, if $j = t\sqrt{d}$ where $t > 0$, the upper bound becomes:
    \begin{align}
        \textrm{variation distance}^2 & \leq \frac{d}{4}e^{-\frac{\pi}{c_2}d^{-\frac{1}{2} - \beta}t + \frac{\pi^3c_2}{6}d^{\beta - \frac{5}{2}}t}\,,
    \end{align}
    which indeed vanishes exponentially as $d \to \infty$ since $\beta < -\frac{1}{2}$.
\end{proof}
\end{lemma}

Having established this result on the marginal distributions of the steps of the random walk, one will now prove an important lemma that will allow us to control a family of expectations involving many $p_j$. 

\begin{lemma}
\label{lemma:bound_recurrence_random_walk}
    Let $I \geq 1$ denote an integer. Let $c > 0$. Then the following bounds hold:
    \begin{align}
        \mathbf{E}^{p_1 = 0}\left[\prod_{2 \leq k \leq I}\left(1 - c\mathbf{1}_{p_k = 0}\right)\right] & \geq \left(1 - \frac{cc_2}{d}\frac{1 + \frac{2e^{-\frac{\pi}{2c_2}d^{-\beta - 1}}}{1 - e^{-\frac{\pi}{c_2}d^{-\beta - 1}}}}{1 - \frac{2e^{-\frac{\pi}{2c_2}d^{1 - \beta}}}{1 - e^{-\frac{\pi}{c_2}d^{1 - \beta}}}}\right)^I\,, \quad (\beta < -1)\,,\\
        \mathbf{E}^{p_1 = 0}\left[\prod_{2 \leq k \leq I}\left(1 - c\mathbf{1}_{p_k = 0}\right)\right] & \geq \left(1 - \sqrt{2}cc_2^{3/2}d^{-1/2 + \beta/2}\frac{1 + \frac{2e^{-2\pi c_2d^{\beta + 1}}}{1 - e^{-4\pi c_2d^{\beta + 1}}}}{1 - \frac{2e^{-\frac{\pi}{2c_2}d^{-\beta + 1}}}{1 - e^{-\frac{\pi}{c_2}d^{-\beta + 1}}}}\right)^I\,,  \quad \left(-1 < \beta < -\frac{1}{2}\right)\,.
    \end{align}
\begin{proof}
    One first rewrites the expectation under consideration by expanding the product:
    \begin{align}
         \mathbf{E}^{p_1 = 0}\left[\prod_{2 \leq k \leq I}\left(1 - c\mathbf{1}_{p_k = 0}\right)\right] & = \mathbf{E}^{p_1 = 0}\left[\sum_{j \geq 0}\sum_{\substack{i_1, \ldots, i_j\\2 \leq i_1 < \ldots < i_j \leq I}}\prod_{1 \leq k \leq j}\left(-c\mathbf{1}_{p_{i_k} = 0}\right)\right]\,,
    \end{align}
    where we use the convention that for $j = 0$, the $\sum\prod$ to equal $1$.
    
    Then, note that for every fixed $2 \leq i_1 < \ldots < i_j$, one may use the Fourier expansion \ref{eq:random_walk_dft_tilde} of the probability distribution for one step of the random walk to obtain:
    \begin{align}
        & \mathbf{E}^{p_1 = 0}\left[\prod_{1 \leq k \leq j}\mathbf{1}_{p_{i_k} = 0}\right]\nonumber\\
        & = \prod_{1 \leq k \leq j}\left(\frac{1}{d}\sum_{-\frac{d - 1}{2} \leq s_k \leq \frac{d - 1}{2}}\widetilde{\theta}_3\left(\frac{s_k}{d}, \frac{2i\sigma_m^2}{d}\right)^{i_k - i_{k - 1}}\right)\\
        & = \frac{1}{d^j}\sum_{-\frac{d - 1}{2} \leq s_1, \ldots, s_j \leq \frac{d - 1}{2}}\prod_{1 \leq k \leq j}\widetilde{\theta}_3\left(\frac{s_k}{d}, \frac{2i\sigma_m^2}{d}\right)^{i_k - i_{k - 1}}\,,
    \end{align}
    where we have set $i_0 := 0$ for convenience. Therefore:
    \begin{align}
        & \mathbf{E}^{p_1 = 0}\left[\sum_{\substack{i_1, \ldots, i_j\\2 \leq i_1 < \ldots < i_j \leq I}}\prod_{1 \leq k \leq j}\left(-c\mathbf{1}_{p_k = 0}\right)\right]\nonumber\\
        & = \left(-\frac{c}{d}\right)^j\sum_{\substack{-\frac{d - 1}{2} \leq s_1, \ldots, s_j \leq \frac{d - 1}{2}\\2 \leq i_1 < \ldots < i_j \leq I}}\prod_{1 \leq k \leq j}\widetilde{\theta}_3\left(\frac{s_k}{d}, \frac{2i\sigma_m^2}{d}\right)^{i_k - i_{k - 1}}\\
        & = \left(-\frac{c}{d}\right)^j\sum_{\substack{-\frac{d - 1}{2} \leq s_1, \ldots, s_j \leq \frac{d - 1}{2}\\n_1, \ldots, n_j \geq 1\\n_1 + \ldots + n_j \leq I}}\prod_{1 \leq k \leq j}\widetilde{\theta}_3\left(\frac{s_k}{d}, \frac{2i\sigma_m^2}{d}\right)^{n_k}\,,
    \end{align}
    but for fixed $s_1, \ldots, s_j$, the sum over $n_1, \ldots, n_j$ shown above is generated by the sum of the coefficients up to degree $I$ of the generating series (in $z$):
    \begin{align}
        \sum_{n_1, \ldots, n_j \geq 1}z^{n_1 + \ldots + n_j}\prod_{1 \leq k \leq j}\widetilde{\theta}_3\left(\frac{s_k}{d}, \frac{2i\sigma_m^2}{d}\right)^{n_k} & = \prod_{1 \leq k \leq j}\frac{z\widetilde{\theta}_3\left(\frac{s_k}{d}, \frac{2i\sigma_m^2}{d}\right)}{1 - z\widetilde{\theta}_3\left(\frac{s_k}{d}, \frac{2i\sigma_m^2}{d}\right)}\,.
    \end{align}
    Summing over all $s_1, \ldots, s_j$, $\mathbf{E}^{p_1 = 0}\left[\sum_{\substack{i_1, \ldots, i_j\\2 \leq i_1 < \ldots < i_j \leq I}}\prod_{1 \leq k \leq j}\left(-c\mathbf{1}_{p_k = 0}\right)\right]$ is therefore generated by the sum of the coefficients up to $I$ in the generating series:
    \begin{align}
        & \left(-\frac{c}{d}\sum_{-\frac{d - 1}{2} \leq s \leq \frac{d - 1}{2}}\frac{z\widetilde{\theta}_3\left(\frac{s}{d}, \frac{2i\sigma_m^2}{d}\right)}{1 - z\widetilde{\theta}_3\left(\frac{s}{d}, \frac{2i\sigma_m^2}{d}\right)}\right)^j\,.
    \end{align}
    Note that this is consistent with our convention that for $j = 0$, the sum over $i_1, \ldots, i_j$ of the products is taken equal to $1$. Finally, summing over all $j \geq 0$, one finds that $\mathbf{E}^{p_1 = 0}\left[\sum_{j \geq 0}\sum_{\substack{i_1, \ldots, i_j\\2 \leq i_1 < \ldots < i_j \leq I}}\prod_{1 \leq k \leq j}\left(-c\mathbf{1}_{p_{i_k} = 0}\right)\right]$ is generated by the sum of the coefficients up to degree $I$ of the series:
    \begin{align}
        \frac{1}{1 + \frac{c}{d}\sum_{-\frac{d - 1}{2} \leq s \leq \frac{d - 1}{2}}\frac{z\widetilde{\theta}_3\left(\frac{s}{d}, \frac{2i\sigma_m^2}{d}\right)}{1 - z\widetilde{\theta}_3\left(\frac{s}{d}, \frac{2i\sigma_m^2}{d}\right)}}  & = \frac{1 - z}{1 - z\left(1 - \frac{c}{d}\sum_{-\frac{d - 1}{2} \leq s \leq \frac{d - 1}{2}}\frac{1 - z}{1 - z\widetilde{\theta}_3\left(\frac{s}{d}, \frac{2i\sigma_m^2}{d}\right)}\widetilde{\theta}_3\left(\frac{s}{d}, \frac{2i\sigma_m^2}{d}\right)\right)}\,,
    \end{align}
    or in other words by the coefficient of degree $I$ of the series
    \begin{align}
        & \frac{1}{1 - z\left(1 - \frac{c}{d}\sum_{-\frac{d - 1}{2} \leq s \leq \frac{d - 1}{2}}\frac{1 - z}{1 - z\widetilde{\theta}_3\left(\frac{s}{d}, \frac{2i\sigma_m^2}{d}\right)}\widetilde{\theta}_3\left(\frac{s}{d}, \frac{2i\sigma_m^2}{d}\right)\right)}\,.
    \end{align}
    Note that the series would simplify nicely and the said coefficient would be trivial to determine if one had either $\widetilde{\theta}_3\left(\frac{s}{d}, \frac{2i\sigma_m^2}{d}\right) = \delta_{s0}$ for all $s$ (corresponding formally to $\sigma_m^2 \downarrow 0$) or $\widetilde{\theta}_3\left(\frac{s}{d}, \frac{2i\sigma_m^2}{d}\right) = 1$ for all $s$ (corresponding formally to $\sigma_m^2 \uparrow +\infty$). In the general case, one will not be able to give a simple formula for the coefficient; yet, the hypotheses of the lemma will suffice to produce a lower bound. In fact,
    \begin{align}
        & \frac{1}{1 - z\left(1 - \frac{c}{d}\sum_{-\frac{d - 1}{2} \leq s \leq \frac{d - 1}{2}}\frac{1 - z}{1 - z\widetilde{\theta}_3\left(\frac{s}{d}, \frac{2i\sigma_m^2}{d}\right)}\widetilde{\theta}_3\left(\frac{s}{d}, \frac{2i\sigma_m^2}{d}\right)\right)}\nonumber\\
        & = \sum_{l \geq 0}z^l\left(1 - \frac{c}{d}\sum_{-\frac{d - 1}{2} \leq s \leq \frac{d - 1}{2}}\frac{1 - z}{1 - z\widetilde{\theta}_3\left(\frac{s}{d}, \frac{2i\sigma_m^2}{d}\right)}\widetilde{\theta}_3\left(\frac{s}{d}, \frac{2i\sigma_m^2}{d}\right)\right)^l\\
        & = \sum_{l \geq 0}z^l\left\{1 - \frac{c}{d}\sum_{-\frac{d - 1}{2} \leq s \leq \frac{d - 1}{2}}\left[\widetilde{\theta}_3\left(\frac{s}{d}, \frac{2i\sigma_m^2}{d}\right) + \sum_{r \geq 1}z^r\widetilde{\theta}_3\left(\frac{s}{d}, \frac{2i\sigma_m^2}{d}\right)^{r - 1}\left(1 - \widetilde{\theta}_3\left(\frac{s}{d}, \frac{2i\sigma_m^2}{d}\right)\right)\right]\right\}^l\\
        & = \sum_{l \geq 0}z^l\left(1 - \frac{c}{d}\sum_{-\frac{d - 1}{2} \leq s \leq \frac{d - 1}{2}}\widetilde{\theta}_3\left(\frac{s}{d}, \frac{2i\sigma_m^2}{d}\right)\right)^l\nonumber\\
        & \hspace{0.15\textwidth} \times \left\{1 + \frac{\sum_{r \geq 1}z^r\sum_{-\frac{d - 1}{2} \leq s \leq \frac{d - 1}{2}}\widetilde{\theta}_3\left(\frac{s}{d}, \frac{2i\sigma_m^2}{d}\right)^{r - 1}\left(1 - \widetilde{\theta}_3\left(\frac{s}{d}, \frac{2i\sigma_m^2}{d}\right)\right)}{1 - \frac{c}{d}\sum_{-\frac{d - 1}{2} \leq s \leq \frac{d - 1}{2}}\widetilde{\theta}_3\left(\frac{s}{d}, \frac{2i\sigma_m^2}{d}\right)}\right\}^l\,.
    \end{align}
    Therefore, if one could show $1 - \frac{c}{d}\sum_{-\frac{d - 1}{2} \leq s \leq \frac{d - 1}{2}}\widetilde{\theta}_3\left(\frac{s}{d}, \frac{2i\sigma_m^2}{d}\right) > 0$, an immediate lower bound for the coefficient of $z^I$ would be:
    \begin{align}
        & \left(1 - \frac{c}{d}\sum_{-\frac{d - 1}{2} \leq s \leq \frac{d - 1}{2}}\widetilde{\theta}_3\left(\frac{s}{d}, \frac{2i\sigma_m^2}{d}\right)\right)^I\left\{1 + \frac{\sum_{r \geq 1}z^r\sum_{-\frac{d - 1}{2} \leq s \leq \frac{d - 1}{2}}\widetilde{\theta}_3\left(\frac{s}{d}, \frac{2i\sigma_m^2}{d}\right)^{r - 1}\left(1 - \widetilde{\theta}_3\left(\frac{s}{d}, \frac{2i\sigma_m^2}{d}\right)\right)}{1 - \frac{c}{d}\sum_{-\frac{d - 1}{2} \leq s \leq \frac{d - 1}{2}}\widetilde{\theta}_3\left(\frac{s}{d}, \frac{2i\sigma_m^2}{d}\right)}\right\}^I\nonumber\\
        & \hspace{0.05\textwidth} \geq \left(1 - \frac{c}{d}\sum_{-\frac{d - 1}{2} \leq s \leq \frac{d - 1}{2}}\widetilde{\theta}_3\left(\frac{s}{d}, \frac{2i\sigma_m^2}{d}\right)\right)^I\,.
    \end{align}
    But
    \begin{align}
        & 1 - \frac{c}{d}\sum_{-\frac{d - 1}{2} \leq s \leq \frac{d - 1}{2}}\widetilde{\theta}_3\left(\frac{s}{d}, \frac{2i\sigma_m^2}{d}\right)\nonumber\\
        & = 1 - \frac{c}{d}\sqrt{\frac{d}{2\sigma_m^2}}\frac{\theta_3\left(0, \frac{i}{2\sigma_m^2d}\right)}{\theta_3\left(0, \frac{2i\sigma_m^2}{d}\right)}\\
        & = 1 - c\sqrt{\frac{c_2}{2}}d^{-1/2 - \beta/2}\frac{\theta_3\left(0, \frac{id^{-\beta - 1}}{2c_2}\right)}{\theta_3\left(0, 2ic_2d^{\beta - 1}\right)}\,.
    \end{align}
    For $\beta < -1$, as $d^{-\beta - 1} \to \infty$, one writes (using equation \ref{eq:approx_theta_im_arg}):
    \begin{align}
        & 1 - c\sqrt{\frac{c_2}{2}}d^{-1/2 - \beta/2}\frac{\theta_3\left(0, \frac{id^{-\beta - 1}}{2c_2}\right)}{\theta_3\left(0, 2ic_2d^{\beta - 1}\right)}\nonumber\\
        & = 1 - cc_2d^{-1}\frac{\theta_3\left(0, \frac{id^{-\beta - 1}}{2c_2}\right)}{\theta_3\left(0, \frac{id^{1 - \beta}}{2c_2}\right)}\\
        & \geq 1 - \frac{cc_2}{d}\frac{1 + \frac{2e^{-\frac{\pi}{2c_2}d^{-\beta - 1}}}{1 - e^{-\frac{\pi}{c_2}d^{-\beta - 1}}}}{1 - \frac{2e^{-\frac{\pi}{2c_2}d^{1 - \beta}}}{1 - e^{-\frac{\pi}{c_2}d^{1 - \beta}}}}\,.
    \end{align}
    For $-1 < \beta < -\frac{1}{2}$, one may estimate it as:
    \begin{align}
        & 1 - c\sqrt{\frac{c_2}{2}}d^{-1/2 - \beta/2}\frac{\theta_3\left(0, \frac{id^{-\beta - 1}}{2c_2}\right)}{\theta_3\left(0, 2ic_2d^{\beta - 1}\right)}\nonumber\\
        & = 1 - \sqrt{2}cc_2^{3/2}d^{-1/2 + \beta/2}\frac{\theta_3\left(0, 2ic_2d^{\beta + 1}\right)}{\theta_3\left(0, \frac{id^{-\beta + 1}}{2c_2}\right)}\\
        & \geq 1 - \sqrt{2}cc_2^{3/2}d^{-1/2 + \beta/2}\frac{1 + \frac{2e^{-2\pi c_2d^{\beta + 1}}}{1 - e^{-4\pi c_2d^{\beta + 1}}}}{1 - \frac{2e^{-\frac{\pi}{2c_2}d^{-\beta + 1}}}{1 - e^{-\frac{\pi}{c_2}d^{-\beta + 1}}}} \,.
    \end{align}
    Therefore, the positivity condition $1 - \frac{c}{d}\sum_{-\frac{d - 1}{2} \leq s \leq \frac{d - 1}{2}}\widetilde{\theta}_3\left(\frac{s}{d}, \frac{2i\sigma_m^2}{d}\right) > 0$ is established in any case and the lemma is proved.
\end{proof}
\end{lemma}
This lemma being established, one is now ready to prove the main proposition concerning the behavior the probabilistic expectation when $\beta < -\frac{1}{2}$.
\begin{proposition}
    Let $t > 0$ and $I := \lceil t\sqrt{d} \rceil$. Assuming the usual scalings $\xi^2 = c_1d^{\alpha}$, $\sigma_m^2 = c_2d^{\beta}$ ($c_1, c_2 > 0$, $-1 < \alpha < 1$) with $\beta < -\frac{1}{2}$ for $\xi^2$ and $\sigma_m^2$, there exists some constant $c_5$ (depending on $c_2$ and $c_3$) such that for all $\varepsilon > 0$, the following holds as $d \to \infty$:
    \begin{align}
        & \mathbf{E}^{p_1}\left[\prod_{1 \leq k \leq I}\left(1 - c_3\mathbf{1}_{p_k = \frac{d - 1}{2}}\right)\right]\nonumber\\
        & \leq 1 - (1 - \varepsilon)\frac{t}{\sqrt{d}}\left(1 + \mathcal{O}\left(d^{3/2}e^{-\frac{\pi\varepsilon t}{2c_2}d^{-\frac{1}{2} - \beta}}\right) + \mathcal{O}\left(e^{-c_5d^{||\beta| - 1|}}\right)\right)\,.
    \end{align}
\begin{proof}
    Fix $I := \left\lceil t\sqrt{d} \right\rceil$ and any integer $p_1 \in \left[-\frac{d - 1}{2}, \frac{d - 1}{2}\right]$. One wants to consider:
    \begin{align}
        \mathbf{E}^{p_1}\left[\prod_{1 \leq k \leq I}\left(1 - c_3\mathbf{1}_{p_k = \frac{d - 1}{2}}\right)\right]\,.
    \end{align}
    Expand the expectand as:
    \begin{align}
        \prod_{1 \leq k \leq I}\left(1 - c_3\mathbf{1}_{p_k = \frac{d - 1}{2}}\right) & = 1 - c_3\sum_{1 \leq k \leq I}\mathbf{1}_{p_k = \frac{d - 1}{2}}\prod_{k < j \leq I}\left(1 - c_3\mathbf{1}_{p_j = \frac{d - 1}{2}}\right)\,.
    \end{align}
    Now, consider the expectation of one term of the sum:
    \begin{align}
        & \mathbf{E}^{p_1}\left[\mathbf{1}_{p_k = \frac{d - 1}{2}}\prod_{k < j \leq I}\left(1 - c_3\mathbf{1}_{p_j = \frac{d - 1}{2}}\right)\right]\nonumber\\
        & = \mathbf{P}^{p_1}\left[p_k = \frac{d - 1}{2}\right]\mathbf{E}^{\frac{d - 1}{2}}\left[\prod_{2 \leq j \leq I - k + 1}\left(1 - c_3\mathbf{1}_{p_j = \frac{d - 1}{2}}\right)\right]\\
        & = \mathbf{P}^{p_1}\left[p_k = \frac{d - 1}{2}\right]\mathbf{E}^{0}\left[\prod_{2 \leq j \leq I - k + 1}\left(1 - c_3\mathbf{1}_{p_j = 0}\right)\right]\,.
    \end{align}
    where we used the strong Markov property and the spatial homogeneity of the process.
    
    Now, fix $\varepsilon \in (0, 1)$ and consider $k$ such that $k \geq \varepsilon t\sqrt{d}$ (such a $k$ exists for large enough $d$, e.g. $d \geq \frac{1}{(1 - \varepsilon)^2t^2}$). Then for such a $k$, lemma \ref{lemma:lower_bound_denominator_initial_distribution} implies that:
    \begin{align}
        \mathbf{P}^{p_1}\left[p_k = \frac{d - 1}{2}\right] & \geq \frac{1}{d} - \frac{\sqrt{d}}{2}e^{-\frac{\pi}{2c_2}d^{-\frac{1}{2} - \beta}\varepsilon t + \frac{\pi^3c_2}{12}d^{\beta - \frac{5}{2}}\varepsilon t}\,.
    \end{align}
    Taking into account $t \leq \sqrt{d}$ and $d \geq 2$, this can be weakened to:
    \begin{align}
        \mathbf{P}^{p_1}\left[p_k = \frac{d - 1}{2}\right] & \geq \frac{1}{d} - \frac{\sqrt{d}}{2}e^{-\frac{\pi}{2c_2}d^{-\frac{1}{2} - \beta}\varepsilon t + \frac{\pi^3c_2}{12}d^{-\frac{1}{2} - \frac{5}{2}}\varepsilon\sqrt{d}}\\
        & \geq \frac{1}{d}\left(1 - \frac{d^{3/2}}{2}e^{-\frac{\pi\varepsilon t   }{2c_2}d^{-\frac{1}{2} - \beta} + \frac{c_2}{2}\varepsilon}\right)\,.
    \end{align}
    Next, lemma \ref{lemma:bound_recurrence_random_walk} asserts the existence of constants $c_4, c_5 > 0$ and $\gamma \in \left[-1, -\frac{3}{4}\right)$ such that:
    \begin{align}
        c_4 & = \sqrt{2}c_3c_2^{3/2}\,,\\
        \mathbf{E}^0\left[\prod_{2 \leq j \leq I - k + 1}\left(1 - c_3\mathbf{1}_{p_j = 0}\right)\right] & \geq \left(1 - c_4d^{\gamma}\right)^{I - k}\left(1 + \mathcal{O}\left(e^{-c_5d^{||\beta| - 1|}}\right)\right)\,,
    \end{align}
    which this time holds for any $k$. In particular, this implies that the expectations that we find in our original sum are all positive. More precisely, for $-1 < \beta < -\frac{1}{2}$:
    \begin{align}
        c_4 & = \sqrt{2}c_3c_2^{3/2}\,,\\
        c_5 & = 2\pi c_2\,,\\
        \gamma & = -\frac{1}{2} + \frac{\beta}{2}\,,
    \end{align}
    while for $\beta < -1$:
    \begin{align}
        c_4 & = c_3c_2\,,\\
        c_5 & = \frac{\pi}{2c_2}\,,\\
        \gamma & = -1\,.
    \end{align}
    This allows one to conclude:
    \begin{align}
        & \mathbf{E}^{p_1}\left[\prod_{1 \leq k \leq I}\left(1 - c_3\mathbf{1}_{p_k = \frac{d - 1}{2}}\right)\right]\nonumber\\
        & = \mathbf{E}^{p_1}\left[1 - c_3\sum_{1 \leq k \leq I}\mathbf{1}_{p_k = \frac{d - 1}{2}}\prod_{k < j \leq I}\left(1 - c_3\mathbf{1}_{p_j = \frac{d - 1}{2}}\right)\right]\\
        & = 1 - c_3\sum_{1 \leq k \leq I}\mathbf{P}^{p_1}\left[p_k = \frac{d - 1}{2}\right]\mathbf{E}^0\left[\prod_{2 \leq j \leq I - k + 1}\left(1 - c_3\mathbf{1}_{p_j = 0}\right)\right]\\
        & \leq 1 - c_3\sum_{1 \leq k \leq I}\mathbf{P}^{p_1}\left[p_k = \frac{d - 1}{2}\right]\mathbf{E}^0\left[\prod_{\varepsilon t\sqrt{d} \leq j \leq I - k + 1}\left(1 - c_3\mathbf{1}_{p_j = 0}\right)\right]\\
        & \leq 1 - \frac{c_3}{d}\left(1 + \mathcal{O}\left(d^{3/2}e^{-\frac{\pi\varepsilon t}{2c_2}d^{-\frac{1}{2} - \beta}}\right) + \mathcal{O}\left(e^{-c_5d^{||\beta| - 1|}}\right)\right)\sum_{\lceil\varepsilon t\sqrt{d}\rceil \leq k \leq I}\left(1 - c_4d^{\gamma}\right)^{I - k}\,.
    \end{align}
    It remains to lower-bound the sum, which can be done as follows:
    \begin{align}
        \sum_{\varepsilon t\sqrt{d} \leq k \leq I}\left(1 - c_4d^{\gamma}\right)^{I - k} & = \left(1 - c_4d^{\gamma}\right)^{I - \lceil\varepsilon t\sqrt{d}\rceil}\frac{1 - (1 - c_4d^{\gamma})^{I - \lceil\varepsilon t\sqrt{d}\rceil + 1}}{c_4d^{\gamma}}\,.
    \end{align}
    Now, from equation \ref{eq:log_1_minus_x_bounds} (for large enough $d$),
    \begin{align}
        \left(1 - c_4d^{\gamma}\right)^{I - \lceil\varepsilon t\sqrt{d}\rceil + 1} & = \exp\left((I - \lceil\varepsilon t\sqrt{d}\rceil + 1)\log\left(1 - c_4d^{\gamma}\right)\right)\\
        & \geq \exp\left((I - \lceil\varepsilon t\sqrt{d}\rceil + 1)\left(-c_4d^{\gamma} - c_4^2d^{2\gamma}\right)\right)\\
        & \geq 1 + (I - \lceil\varepsilon t\sqrt{d}\rceil + 1)\left(-c_4d^{\gamma} - c_4^2d^{2\gamma}\right)\\
        & \geq 1 - (1 - \varepsilon)t\sqrt{d}c_4d^{\gamma}\left(1 + c_4d^{\gamma}\right)\,.
    \end{align}
    Hence:
    \begin{align}
        & \mathbf{E}^{p_1}\left[\prod_{1 \leq k \leq I}\left(1 - c_3\mathbf{1}_{p_k = \frac{d - 1}{2}}\right)\right]\nonumber\\
        & \leq 1 - (1 - \varepsilon)c_3\frac{t}{\sqrt{d}}\left(1 + \mathcal{O}\left(d^{3/2}e^{-\frac{\pi\varepsilon t}{2c_2}d^{-\frac{1}{2} - \beta}}\right) + \mathcal{O}\left(e^{-c_5d^{||\beta| - 1|}}\right)\right)\,.
    \end{align}
\end{proof}
\end{proposition}

This yields the following corollary:
\begin{corollary}
\label{cor:quasi_ideal_clock_variance_bad_scaling}
    Let $t > 0$ and $I := \lceil t\sqrt{d} \rceil$. Assume the usual scalings $\xi^2 = c_1d^{\alpha}$, $\sigma_m^2 = c_2d^{\beta}$ ($c_1, c_2 > 0$, $-1 < \alpha < 1$) with $\beta < -\frac{1}{2}$ for $\xi^2$ and $\sigma_m^2$. Then as $d \to \infty$,
    \begin{align}
        \left|\left\langle\exp\left(\frac{2\pi i\widetilde{\xi}_I}{\sqrt{d}}\right)\right\rangle\right| & \leq 1 - \frac{t}{\sqrt{d}}(1 + o(1))\,.
    \end{align}
\end{corollary}

\subsubsection{Arbitrary correlations for arbitrary $\delta$ in the optimal case}
In the optimal case identified in corollary \ref{cor:quasi_ideal_clock_variance_good_scaling}, for which $\sigma_m^2 := c_2d^{\beta}, \beta \downarrow -\frac{1}{2}$, the probabilistic expectation contribution to the pseudo-variance \ref{eq:quasi_ideal_clock_pseudo_variance} is close to $1$ up to exponentially decaying terms. It is not difficult to see that the same would apply not only to the pseudo-variance but also to an arbitrary correlation function \ref{eq:quasi_ideal_clock_correlation} for arbitrary $\delta$ (provided the scaling $t\sqrt{d}$ of the measurement index $I$ is fixed). In other words, referring again to the formulae for the moments of order $1$ and $2$ of linear measurement, one may say that for $\beta > -\frac{1}{2}$, the correlations retain essentially the contributions of the spread of the wavefunction and that of the imprecision of the measurement while the backaction part is suppressed as the exponential of a power of $d$.

\subsubsection{Random measurement times of integer mean}
\label{sec:IId3}
One could now, in the original equation \ref{eq:quasi_ideal_clock_pseudo_variance} for the pseudo-variance, consider the case where $\delta_j = 1 + X_j$ where $X_j$ is a random variable symmetrically distributed about $0$. Suppose for definiteness that the $X_j$ are i.i.d with:
\begin{align}
    X_j & \sim \mathcal{N}(0, \sigma^2)\,.
\end{align}
Then taking the expectation of equation \ref{eq:quasi_ideal_clock_pseudo_variance} over the distribution of the $X_j$, the factor $\left(1 - e^{-2\pi i\delta_j}\right)$ is converted to $1 - e^{-2\pi\sigma^2}$. Therefore, one is led back to the case $\delta = \frac{1}{2}$, except that $c_3 = 1 - e^{-2\pi\sigma^2}$ instead of $2$. Fundamentally, such a setting would arise if the clock was coupled to a classical (random) waveform $x(\cdot)$ through the following Hamiltonian:
\begin{align}
    \hat{H}_d'(t) & = \left(1 + x(t)\right)\hat{H}_d\\
    & = \left(1 + x(t)\right)\frac{2\pi}{\sqrt{d}}\sum_{-\frac{d - 1}{2} \leq n \leq \frac{d - 1}{2}}n\ket{n}\bra{n}
\end{align}
and one would measure the clock at interval $\frac{1}{\sqrt{d}}$. The $X_j$ would correspond to:
\begin{align}
    X_j & = \sqrt{d}\int_{\frac{j - 1}{\sqrt{d}}}^{\frac{j}{\sqrt{d}}}\!\mathrm{d}t\,x(t)\,.
\end{align}
As a simple realization (to ensure that the $X_j$ are i.i.d), one can think of $x(\cdot)$ as white noise with variance $\sigma^2$.

\subsubsection{Application to waveform estimation}


Keeping the setting introduced in the last paragraph, one may now wonder what can be learnt about $x(\cdot)$ from monitoring the time $\widetilde{\xi}$ of the clock. An estimate for $X_j$ from the measurements of $\widetilde{\xi}_{j - 1}, \widetilde{\xi}_j$ is $\sqrt{d}\left(\widetilde{\xi}_j - \widetilde{\xi}_{j - 1}\right)$.

From what we identified as the optimal scaling in $d$ from corollary \ref{cor:quasi_ideal_clock_variance_good_scaling}, the variance on each measurement of $\frac{\widetilde{\xi}}{\sqrt{d}}$ cannot decrease to $0$ (with $d$) faster than $d^{-3/2}$. This implies an error increasing at least as $\left(d^2 \times d^{-3/2}\right)^{1/2} = d^{1/4}$ on the estimation of each $X_j$. Now, the typical magnitude of $X_j$ is given by $\sqrt{d}\left(\sigma^2\frac{1}{\sqrt{d}}\right)^{1/2} = \sigma d^{1/4}$. Therefore, $\sigma$ should grow faster than $1$ (that is, diverge) for the standard deviation on the estimation of $X_j$ to be negligible with respect to $X_j$.

Now, assume one no longer measures the clock at every $\frac{1}{\sqrt{d}}$ time steps, but approximately every $d^{-\varepsilon}$ time steps instead ($0 \leq \varepsilon < \frac{1}{2}$). The integral corresponding to white noise, rescaled by $\sqrt{d}$, along such an interval has typical magnitude $\sqrt{d}\left(\sigma^2d^{-\varepsilon}\right)^{1/2} = \sigma d^{(1 -\varepsilon)/2}$. For the standard deviation of the estimator of $X_j$ to be negligible with respect to $X_j$, one now needs $d^{1/4} \ll \sigma d^{(1 - \varepsilon)/2}$, that is $\sigma \gg d^{\varepsilon/2 - 1/4}$. If one considers the limit $\varepsilon \to 0$, this means that $\sigma$ needs to grow faster than $d^{-1/4}$.

\subsection{Existence of continuum limit}
In sections \ref{sec:measured_quasi_ideal_clock_derivation_formulae} and \ref{sec:measured_quasi_ideal_clock_scalings}, given a quasi-ideal clock of dimension $d$ measured with some precision $\sigma_m$, we constructed and characterized a discrete-time random process $\widetilde{\xi}_1^{(d)}, \widetilde{\xi}_2^{(d)}, \ldots$ Here, we added the superscript $d$ here to account for the dependence of this random process on the dimension; in the construction, the parameters $c_2, c_3, \alpha, \beta$ which specify $\xi^2$ and $\sigma_m^2$ for all $d$ are implicitly fixed; we also always take $\delta_j := \frac{1}{2}$ for all $j$. One may then wonder, loosely speaking, whether this discrete-time process admits a sensible ``continuum limit'' as $d \to \infty$. We will address this question in two interesting limiting cases: on the one hand, $\sigma_m^2 \propto d^{\beta}$ with $-\frac{1}{2} < \beta < 0$; on the other hand, measurement in the time eigenbasis as described section \ref{sec:quasi_ideal_clock_time_basis} ---formally $\sigma_m^2 = 0$.

Generally speaking, showing the convergence in law of a sequence of stochastic processes $\left\{(X^n_t)_{t \geq 0}\right\}_{n \geq 0}$ to some stochastic process $\left(X_t\right)$ implies showing the weak convergence of the finite-dimensional distribution as well as verifying a tightness condition; see \cite[theorems 7.5-13.5]{billingsley_convergence_1999} for specific criteria and \cite[theorems 8.1-8.2]{billingsley_convergence_1999} for a simple application (weak convergence of  a random walk of finite variance to Brownian motion).

We now give a sketch of how these results can be applied to the problem under consideration. First, the weak convergence of finite-dimensional distributions may be conveniently derived from the pointwise convergence of the characteristic functions. Indeed, the latter can be recovered --- say in the 2-dimensional case for definiteness --- from our expression for $\left\langle\exp\left(\frac{2\pi n\widetilde{\xi}^{(d)}_J}{\sqrt{d}}\right)\exp\left(\frac{2\pi m\widetilde{\xi}^{(d)}_I}{\sqrt{d}}\right)\right\rangle$ ($m, n \in \mathbf{Z}$) by taking $m := \lfloor \chi_0\sqrt{d}\rfloor, n := \lfloor\chi_1\sqrt{d}\rfloor$ where $\chi_0, \chi_1 \in \mathbf{R}$ are fixed. This differs from the case in which $m, n$ are constants treated above; fortunately enough, at least in the cases $\sigma_m^2 \gg d^{-1/2}$ and $\sigma_m^2 = 0$ (measurement in the time eigenbasis), one can show that our estimates are still robust against this scaling.

Let us now address the case $\sigma_m^2 \propto d^{\beta}, -\frac{1}{2} < \beta < 0$. To begin with, we state the following proposition which is a straightforward generalization of equation
\ref{eq:quasi_ideal_clock_correlation}:

\begin{proposition}
\label{prop:pseudo_correlation_functions}
    Let $I_1, \ldots, I_n$ denote positive integers such that $2 \leq I_1 < I_2 < \ldots < I_n$ and $m_1, \ldots, m_n$ denote integers satisfying $|m_k + \ldots + m_n| < d$ for all $1 \leq k \leq n$. Then the following holds (with the convention that $I_0 := 0$):
    \begin{align}
        & \left\langle\prod_{1 \leq k \leq n}\exp\left(\frac{2\pi im_k\widetilde{\xi}_{I_k}}{\sqrt{d}}\right)\right\rangle\nonumber\\
        & \hspace{0.05\textwidth} = \frac{\prod_{1 \leq k \leq n}\theta_3\left(\frac{i\sigma_m^2}{d}m_k, \frac{2i\sigma_m^2}{d}\right)}{\theta_3\left(0, \frac{2i\sigma_m^2}{d}\right)^n}e^{-\frac{\pi\sigma_m^2}{d}\sum_{1 \leq k \leq n}m_k^2 + \frac{2\pi i}{d}\sum_{1 \leq k \leq n}m_k\sum_{1 \leq j \leq I_k}\delta_j}\nonumber\\
        & \times \frac{\theta_3\left(\frac{\sum_{1 \leq k \leq n}m_k}{2d}, \frac{i}{2\xi^2d}\right)\theta_3\left(\frac{\sum_{1 \leq k \leq n}m_k}{2}, \frac{id}{2\xi^2}\right) + \theta_3\left(\frac{1}{2} - \frac{\sum_{1 \leq k \leq n}m_k}{2d}, \frac{i}{2\xi^2d}\right)\theta_3\left(\frac{1}{2} + \frac{\sum_{1 \leq k \leq n}m_k}{2}, \frac{id}{2\xi^2}\right)}{\theta_3\left(0, \frac{i}{2\xi^2d}\right)\theta_3\left(0, \frac{id}{2\xi^2}\right) + \theta_3\left(\frac{1}{2}, \frac{i}{2\xi^2d}\right)\theta_3\left(\frac{1}{2}, \frac{id}{2\xi^2}\right)}\nonumber\\
        & \times \sum_{-\frac{d - 1}{2} \leq p_1, \ldots, p_{I_n} \leq \frac{d - 1}{2}}\prod_{\substack{1 \leq j \leq I_n\\j \neq I_1, \ldots, I_n}}\frac{\theta_3\left(\frac{p_{j + 1} - p_j}{d}, \frac{i}{2\sigma_m^2d}\right)}{\sqrt{2\sigma_m^2d}\theta_3\left(0, \frac{2i\sigma_m^2}{d}\right)}\prod_{1 \leq k \leq n}\frac{\theta_3\left(\frac{p_{I_k + 1} - p_{I_k}}{d} + \frac{m_k}{2d}, \frac{i}{2\sigma_m^2d}\right)}{\sqrt{2\sigma_m^2d}e^{-\frac{\pi\sigma_m^2}{2d}m_k^2}\theta_3\left(\frac{i\sigma_m^2}{d}m_k, \frac{2i\sigma_m^2}{d}\right)}\nonumber\\
        & \hspace{0.02\textwidth}\times \prod_{1 \leq k \leq n}\prod_{I_{j - 1}< j \leq I_j}\left(1 - \left(1 - e^{-2\pi i\delta_j\sign\left(\sum_{k \leq l \leq n}m_l\right)}\right)\mathbf{1}_{|p_j - \sum_{k \leq l \leq n}m_l| > \frac{d - 1}{2}}\right)\nonumber\\
        &  \hspace{0.02\textwidth} \times \frac{\theta_3\left(\frac{p_1 - n_0 - \sum_{1 \leq k \leq n}m_k}{d}, \frac{i}{\xi^2d}\right)\theta_3\left(\frac{p_1 - n_0}{d}, \frac{i}{\xi^2d}\right)}{\frac{d}{2}\left(\theta_3\left(\frac{\sum_{1 \leq k \leq n}m_k}{2d}, \frac{i}{2\xi^2d}\right)\theta_3\left(\frac{\sum_{1 \leq k \leq n}m_k}{2}, \frac{id}{2\xi^2}\right) + \theta_3\left(\frac{1}{2} - \frac{\sum_{1 \leq k \leq n}m_k}{2d}, \frac{i}{2\xi^2d}\right)\theta_3\left(\frac{1}{2} + \frac{\sum_{1 \leq k \leq n}m_k}{2}, \frac{id}{2\xi^2}\right)\right)}
    \end{align}
\end{proposition}

We now need two lemmas; the first one allows to verify the \textit{tightness condition} in \cite[theorem 7.5]{billingsley_convergence_1999}, the second one shows the convergence of the finite-dimensional distributions.

\begin{lemma}
\label{lemma:tightness_uniform_motion}
    Let $\gamma$ satisfy $\alpha \vee \beta < \gamma < 0$. For all $d \geq 1$, define the piecewise constant continuous-time random process $\left(\widetilde{\Xi}^{(d)}_t\right)_{t \geq 0}$ by $\widetilde{\Xi}^{(d)}_0 := 0$, $\widetilde{\Xi}^{(d)}_{kd^{\gamma}} := \widetilde{\xi}^{(d)}_{\lfloor 2kd^{1/2 + \gamma} \rfloor}$ for all $k \geq 1$ and extend it piecewise linearly to $t \in \mathbf{R}_+ - \{ kd^{\gamma}\,;\, k \in \mathbf{Z}_+ \}$. Then for all $T > 0$, for all $\varepsilon > 0$, there exists $\delta > 0$ such that $\limsup_{d \to \infty}\mathbf{P}\left[\sup_{\substack{0 \leq s < t \leq T\\|s - t| \leq \delta}}\left|\exp\left(2\pi i\eta\widetilde{\Xi}^{(d)}_t\right) - \exp\left(2\pi i\eta\widetilde{\Xi}^{(d)}_s\right)\right| > \varepsilon\right] = 0$.
\begin{proof}
    To ease the notation, let us denote by ``\textrm{exp terms}'' an error term consisting of a polynomial in $d$ times a decreasing exponential of some power of $d$. 
    
    Fix $T > 0$ throughout the proof. Let $k, l$ denote integers satisfying $1 \leq k < l \leq Td^{-\gamma}$ and $\eta > 0$. By a straightforward generalization of the estimates in section \ref{sec:measured_quasi_ideal_clock_scalings} and at the end of section \ref{sec:measured_quasi_ideal_clock_derivation_formulae},
    \begin{align}
        & \left\langle\exp\left(\frac{2\pi i\lfloor\eta\sqrt{d}\rfloor}{\sqrt{d}}\left(\widetilde{\Xi}^{(d)}_{kd^{\gamma}} - \widetilde{\Xi}^{(d)}_{ld^{\gamma}}\right)\right)\right\rangle\nonumber\\
        & = e^{-\frac{\pi\sigma_m^2\lfloor\eta\sqrt{d}\rfloor^2}{2d} + \frac{2\pi i\lfloor\eta\sqrt{d}\rfloor}{d}\sum_{\lfloor 2kd^{\gamma + 1/2}\rfloor < r \leq \lfloor 2ld^{\gamma + 1/2} \rfloor}\frac{1}{2}}\left(1 + \textrm{exp terms}\right)\\
        & = e^{-\frac{\pi\sigma_m^2\lfloor\eta\sqrt{d}\rfloor^2}{2d} + \frac{\pi i\lfloor\eta\sqrt{d}\rfloor\left(\lfloor 2ld^{\gamma + 1/2}\rfloor - \lfloor 2kd^{\gamma + 1/2}\rfloor\right)}{d}}\left(1 + \textrm{exp terms}\right)
    \end{align}
    This implies:
    \begin{align}
        & \left\langle\left|\exp\left(\frac{2\pi i\lfloor\eta\sqrt{d}\rfloor}{\sqrt{d}}\widetilde{\Xi}^{(d)}_{kd^{\gamma}}\right) - \exp\left(\frac{2\pi i\lfloor\eta\sqrt{d}\rfloor}{\sqrt{d}}\widetilde{\Xi}^{(d)}_{ld^{\gamma}}\right)\right|^2\right\rangle\nonumber\\
        & \leq 2\left(1 - e^{-\frac{\pi\sigma_m^2\lfloor\eta\sqrt{d}\rfloor^2}{2d}}\right) + 2\left(1 - \cos\left(\frac{\pi\lfloor\eta\sqrt{d}\rfloor\left(\lfloor 2ld^{\gamma + 1/2}\rfloor - \lfloor 2kd^{\gamma + 1/2}\rfloor\right)}{d}\right)\right) + \textrm{exp terms}\\
        & \leq \frac{\pi\sigma_m^2\lfloor\eta\sqrt{d}\rfloor^2}{d} + \left(\frac{\pi\lfloor\eta\sqrt{d}\rfloor\left(\lfloor 2ld^{\gamma + 1/2}\rfloor - \lfloor 2kd^{\gamma + 1/2}\rfloor\right)}{d}\right)^2 + \textrm{exp terms}\\
        & \leq \pi\sigma_m^2\eta^2 + \left(2\pi\eta(l - k)d^{\gamma}\right)^{2}\left(1 + \frac{1}{\eta\sqrt{d}}\right)^2\left(1 + \frac{2}{\lfloor 2ld^{\gamma + 1/2}\rfloor - \lfloor 2kd^{\gamma + 1/2}\rfloor}\right)^2 + \textrm{exp terms}\\
        & \leq \pi\sigma_m^2\eta^2 + 4\left(1 + \frac{1}{\eta\sqrt{d}}\right)^2\left(2\pi\eta(l - k)d^{\gamma}\right)^2 + \textrm{exp terms}.
    \end{align}
    Let now $\zeta$ satisfy
    \begin{align}
    \label{eq:proof_tightness_uniform_motion_zeta}
        \frac{\gamma}{2} \vee \frac{\beta - \gamma}{2} & < \zeta < 0
    \end{align}
    From Chebyshev's inequality, it follows that
    \begin{align*}
        & \mathbf{P}\left[\left|\exp\left(\frac{2\pi i\lfloor\eta\sqrt{d}\rfloor}{\sqrt{d}}\widetilde{\Xi}^{(d)}_{kd^{\gamma}}\right) - \exp\left(\frac{2\pi i\lfloor\eta\sqrt{d}\rfloor}{\sqrt{d}}\widetilde{\Xi}^{(d)}_{(k + 1)d^{\gamma}}\right)\right| > d^{\zeta}\right]\nonumber\\
        & \leq \pi\eta^2\sigma_m^2d^{-2\zeta} + 16\left(1 + \frac{1}{\eta\sqrt{d}}\right)^2\pi^2\eta^2(l - k)^2d^{2\gamma - 2\zeta} + \textrm{exp terms}.
    \end{align*}
    We now need to formalize the idea that $\exp\left(\frac{2\pi i\lfloor\eta\sqrt{d}\rfloor}{\sqrt{d}}\widetilde{\Xi}^{(d)}_{kd^{\gamma}}\right)$ is ``close'' to $\exp\left(2\pi i\eta\widetilde{\Xi}^{(d)}_{kd^{\gamma}}\right)$. This will be verified if $\widetilde{\Xi}^{(d)}_{kd^{\gamma}}$ is small enough, which happens with high probability; more precisely:
    \begin{align}
        & \mathbf{P}\left[\left|\exp\left(\frac{2\pi i\lfloor\eta\sqrt{d}\rfloor}{\sqrt{d}}\widetilde{\Xi}^{(d)}_{kd^{\gamma}}\right) - \exp\left(2\pi i\eta\widetilde{\Xi}^{(d)}_{kd^{\gamma}}\right)\right| > d^{\zeta}\right]\nonumber\\
        & \leq \mathbf{P}\left[\left|\exp\left(\frac{2\pi i\widetilde{\Xi}^{(d)}_{kd^{\gamma}}}{\sqrt{d}}\right) - 1\right| > d^{\zeta}\right]\\
        & \leq d^{-2\zeta}\left\langle\left|\exp\left(\frac{2\pi i\widetilde{\Xi}^{(d)}_{kd^{\gamma}}}{\sqrt{d}}\right) - 1\right|^2\right\rangle\\
        & \leq d^{-2\zeta}\left[\left(\frac{\pi\lfloor kd^{\gamma}\rfloor}{d}\right)^2 + \frac{\pi\sigma_m^2}{d}\right]\\
        & \leq \left(\pi^2k^2d^{2\gamma - 2\zeta - 2} + \pi\sigma_m^2d^{-2\zeta - 1}\right)\\
        & \leq \mathcal{O}\left(T^2d^{-2\zeta - 2}\right) + \mathcal{O}\left(d^{\beta - 2\zeta - 1}\right).
    \end{align}
    It follows:
    \begin{align}
        \mathbf{P}\left[\left|\exp\left(2\pi i\eta\widetilde{\Xi}^{(d)}_{kd^{\gamma}}\right) - \exp\left(2\pi i\eta\widetilde{\Xi}^{(d)}_{(k + 1)d^{\gamma}}\right)\right| > 3d^{\zeta}\right] & = \mathcal{O}\left(d^{\beta - 2\zeta} + d^{2\gamma - 2\zeta}\right) + \textrm{exp terms}.
    \end{align}
    Finally, from a union bound:
    \begin{align}
        \mathbf{P}\left[\left|\exp\left(2\pi i\eta\widetilde{\Xi}^{(d)}_{kd^{\gamma}}\right) - \exp\left(2\pi i\eta\widetilde{\Xi}^{(d)}_{(k + 1)d^{\gamma}}\right)\right| > 3d^{\zeta} \textrm{ for some $k \in [1, Td^{-\gamma}]$}\right] & = \mathcal{O}\left(d^{-\gamma + \beta - 2\zeta} + d^{\gamma - 2\zeta}\right) + \textrm{exp terms}.
    \end{align}
    Since $\left(\widetilde{\Xi}^{(d)}_t\right)_{t \geq 0}$ is linear between points $kd^{\gamma}$ and $(k + 1)d^{\gamma}$, this implies:
    \begin{align}
        \mathbf{P}\left[\sup_{\substack{0 \leq s, t \leq T\\s \neq t}}\frac{\left|\exp\left(2\pi i\eta\widetilde{\Xi}^{(d)}_t\right) - \exp\left(2\pi i\eta\widetilde{\Xi}^{(d)}_s\right)\right|}{\left|t - s\right|^{\frac{\zeta}{\gamma}}} > 3\right] = \mathcal{O}\left(d^{-\gamma + \beta - 2\zeta} + d^{\gamma - 2\zeta}\right) + \textrm{exp terms},
    \end{align}
    which by equation \ref{eq:proof_tightness_uniform_motion_zeta} vanishes as $d \to \infty$. This proves the sought tightness condition.
\end{proof}
\end{lemma}

\begin{lemma}
\label{lemma:fd_distributions_uniform_motion}
    Let $n \geq 1$, $\theta_1, \ldots, \theta_n \in \mathbf{R}$ and $t_1, \ldots, t_k \in \mathbf{R}$. The following holds:
    \begin{align}
        \left\langle\prod_{1 \leq k \leq n}\exp\left(2\pi i\theta_k\widetilde{\Xi}^{(d)}_{t_k}\right)\right\rangle & \xrightarrow[d \to \infty]{} e^{2\pi i\sum_{1 \leq k \leq n}\theta_k t_k}\,.
    \end{align}
\begin{proof}
    This essentially follows from applying proposition \ref{prop:pseudo_correlation_functions} with $m_k := \lfloor \theta_k\sqrt{d} \rfloor$ and $I_k := \left\lfloor\left\lfloor t_kd^{-\gamma}\right\rfloor d^{\gamma + 1/2}\right\rfloor$, using continuity arguments as in the proof of lemma \ref{lemma:tightness_uniform_motion} and estimates similar to those of sections \ref{sec:measured_quasi_ideal_clock_derivation_formulae}, \ref{sec:measured_quasi_ideal_clock_scalings}.
\end{proof}
\end{lemma}

This leads us to the following result:
\begin{theorem}
\label{th:convergence_measured_quasi_ideal_clock_uniform_motion}
    As $d \to \infty$, $\left(\widetilde{\Xi}^{(d)}_t\right)_{t \geq 0}$ converges weakly to a uniform motion.
\begin{proof}
    This follows from applying \cite[theorem
    7.5]{billingsley_convergence_1999}. The convergence of finite-dimensional distributions is given by lemma \ref{lemma:fd_distributions_uniform_motion}, the tightness condition is verified in lemma \ref{lemma:tightness_uniform_motion}.
\end{proof}
\end{theorem}

We now address the case of the quasi-ideal clock measured in the time eigenbasis. In this setting, we define a family (indexed by the dimension $d$) piecewise-constant continuous-time processes $\left(\widetilde{\Xi}^{(d)}_t\right)_{t \geq 0}$ by $\widetilde{\Xi}^{(d)}_t := \widetilde{\xi}^{(d)}_{\lfloor 2t\sqrt{d} \rfloor}$. We start by stating the analog of proposition \ref{prop:pseudo_correlation_functions} for a quasi-ideal clock measured in the time eigenbasis; the proof also follows analogously.

\begin{proposition}
\label{prop:pseudo_correlation_functions_time_eigenbasis}
    Consider the quasi-ideal clock measured in the time eigenbasis described in section \ref{sec:quasi_ideal_clock_time_basis}. Let $n \geq 1$, $m_1, \ldots, m_n$ integers such that $\left|\sum_{1 \leq l \leq k}m_l\right| < d$ for all $1 \leq k \leq n$ and $I_1, \ldots, I_n$ positive integers such that $2 \leq I_1 < \ldots < I_n$. Then the following holds (with the convention $I_0 := 0$):
    \begin{align}
        & \left\langle\prod_{1 \leq k \leq n}\exp\left(\frac{2\pi im_k\widetilde{\xi}^{(d)}_{I_k}}{\sqrt{d}}\right)\right\rangle = e^{\frac{2\pi i\delta}{d}\sum_{1 \leq k \leq n}m_kI_k}\prod_{1 \leq k \leq n}\left(1 - \left(1 - e^{-2\pi i\delta\sign\left(\sum_{k \leq l \leq n}m_l\right)}\right)\frac{\left|\sum_{k \leq l \leq n}m_l\right|}{d}\right)^{I_k - I_{k - 1}}
    \end{align}
\end{proposition}

We now turn to prove the tightness condition required for the application of \cite[theorem 13.5]{billingsley_convergence_1999}.

\begin{lemma}
\label{lemma:tightness_measurement_time_eigenbasis}
Let $\left(\widetilde{\Xi}^{(d)}_t\right)_{t \geq 0}$ be defined as stated above. Fix $T > 0$ and $m \geq 1$. Then for all $r, s, t$ satisfying $0 \leq r < s < t \leq T$, all $\eta \in 2^{-m}\mathbf{Z}$ and all integer $p \geq m$.
\begin{align}
    \mathbf{E}\left[\left|\exp\left(2\pi i
    \widetilde{\Xi}^{(4^p)}_t\right) - \exp\left(2\pi i\widetilde{\Xi}^{(4^p)}_s\right)\right|^2\left|\exp\left(2\pi i\widetilde{\Xi}^{(4^p)}_s\right) - \exp\left(2\pi i\widetilde{\Xi}^{(4^p)}_r\right)\right|^2\right] & = \mathcal{O}\left(\eta^2(t - r)^2\right)
\end{align}
where the $\mathcal{O}$ is independent of $p$ and $r, s, t \in [0, T]$.
\begin{proof}
Expanding the argument of the expectation gives:
\begin{align*}
    & \left|\exp\left(2\pi i\eta\widetilde{\Xi}^{(4^p)}_t\right) - \exp\left(2\pi i\eta\widetilde{\Xi}^{(4^p)}_s\right)\right|^2\left|\exp\left(2\pi i\eta\widetilde{\Xi}^{(4^p)}_s\right) - \exp\left(2\pi i\eta\widetilde{\Xi}^{(4^p)}_r\right)\right|^2\\
    & = \left[2 - \exp\left(2\pi i\eta\left(\widetilde{\Xi}^{(4^p)}_t - \widetilde{\Xi}^{(4^p)}_s\right)\right) - \exp\left(2\pi i\eta\left(\widetilde{\Xi}^{(d)}_s - \widetilde{\Xi}^{(4^p)}_t\right)\right)\right]\\
    & \hspace{0.05\textwidth} \times \left[2 - \exp\left(2\pi i\eta\left(\widetilde{\Xi}^{(4^p)}_s - \widetilde{\Xi}^{(4^p)}_r\right)\right) - \exp\left(2\pi i\eta\left(\widetilde{\Xi}^{(4^p)}_r - \widetilde{\Xi}^{(4^p)}_s\right)\right)\right]\\
    & = 4 - 2\left[\exp\left(2\pi i\eta\left(\widetilde{\Xi}^{(4^p)}_t - \widetilde{\Xi}^{(4^p)}_s\right)\right) + \exp\left(2\pi i\eta\left(\widetilde{\Xi}^{(4^p)}_s - \widetilde{\Xi}^{(4^p)}_t\right)\right) + \exp\left(2\pi i\eta\left(\widetilde{\Xi}^{(4^p)}_s - \widetilde{\Xi}^{(4^p)}_r\right)\right) + \exp\left(2\pi i\eta\left(\widetilde{\Xi}^{(4^p)}_r - \widetilde{\Xi}^{(4^p)}_s\right)\right)\right]\\
    & \hspace{0.05\textwidth} + \exp\left(2\pi i\eta\left(\widetilde{\Xi}^{(4^p)}_t - \widetilde{\Xi}^{(4^p)}_r\right)\right) + \exp\left(2\pi i\eta\eta\left(\widetilde{\Xi}^{(4^p)}_r - \widetilde{\Xi}^{(4^p)}_t\right)\right)\\
    & \hspace{0.05\textwidth} + \exp\left(2\pi i\eta\left(\widetilde{\Xi}^{(4^p)}_t + \widetilde{\Xi}^{(4^p)}_r - 2\widetilde{\Xi}^{(4^p)}_s\right)\right) + \exp\left(2\pi i\eta\left(2\widetilde{\Xi}^{(4^p)}_s - \widetilde{\Xi}^{(4^p)}_t - \widetilde{\Xi}^{(4^p)}_r\right)\right)
\end{align*}
Now, recalling $2^{p}\eta = 2^{p - m}2^m\eta \in \mathbf{Z}$, the expectation of each terms may be computed using proposition \ref{prop:pseudo_correlation_functions_time_eigenbasis} (applied with $d := 4^p$). For example:
\begin{align*}
    \exp\left(2\pi i\eta\left(\widetilde{\Xi}^{(4^p)}_t - \widetilde{\Xi}^{(4^p)}_s\right)\right) & = \exp\left(\frac{2\pi i\eta 2^p}{\sqrt{4^p}}\left(\widetilde{\xi}^{(4^p)}_{\lfloor 2t\sqrt{4^p} \rfloor} - \widetilde{\xi}^{(4^p)}_{\lfloor 2s\sqrt{4^p}\rfloor}\right)\right)\\
    & = e^{\frac{i\pi\eta}{2^p}\left(\lfloor 2^{p + 1}t \rfloor - \lfloor 2^{p + 1}s \rfloor\right)}\left(1 - \frac{\eta}{2^{p - 1}}\right)^{\lfloor 2^{p + 1}t \rfloor - \lfloor 2^{p + 1}s \rfloor}
\end{align*}
The complete result is:
\begin{align*}
    & \mathbf{E}\left[\left|\exp\left(2\pi i
    \widetilde{\Xi}^{(4^p)}_t\right) - \exp\left(2\pi i\widetilde{\Xi}^{(4^p)}_s\right)\right|^2\left|\exp\left(2\pi i\widetilde{\Xi}^{(4^p)}_s\right) - \exp\left(2\pi i\widetilde{\Xi}^{(4^p)}_r\right)\right|^2\right]\\
    & = 4 - 4\cos\left(\frac{\pi\eta}{2^p}\left(\lfloor 2^{p + 1}t \rfloor - \lfloor 2^{p + 1}s \rfloor\right)\right)\left(1 - \frac{\eta}{2^{p - 1}}\right)^{\lfloor 2^{p + 1}t \rfloor - \lfloor 2^{p + 1}s \rfloor} - 4\cos\left(\frac{\pi\eta}{2^p}\left(\lfloor 2^{p + 1}s \rfloor - \lfloor 2^{p + 1}r \rfloor\right)\right)\left(1 - \frac{\eta}{2^{p - 1}}\right)^{\lfloor 2^{p + 1}s \rfloor - \lfloor 2^{p + 1}r \rfloor}\\
    & \hspace{0.05\textwidth} + 2\cos\left(\frac{\pi\eta}{2^p}\left(\lfloor 2^{p + 1}t \rfloor - \lfloor 2^{p + 1}r \rfloor\right)\right)\left(1 - \frac{\eta}{2^{p - 1}}\right)^{\lfloor 2^{p + 1}t \rfloor - \lfloor 2^{p + 1}r \rfloor}\\
    & \hspace{0.05\textwidth} + 2\cos\left(\frac{\pi\eta}{2^p}\left(\lfloor 2^{p + 1}t \rfloor + \lfloor 2^{p + 1}r \rfloor - 2\lfloor 2^{p + 1}s \rfloor\right)\right)\left(1 - \frac{\eta}{2^{p - 1}}\right)^{\lfloor 2^{p + 1}t \rfloor - \lfloor 2^{p + 1}r \rfloor}\\
    & = \mathcal{O}\left(\eta^2(t - r)^2\right) + 4 - 4\left(1 - \frac{\eta}{2^{p - 1}}\right)^{\lfloor 2^{p + 1} t \rfloor - \lfloor 2^{p + 1}s \rfloor} - 4\left(1 - \frac{\eta}{2^{p - 1}}\right)^{\lfloor 2^{p + 1}s \rfloor - \lfloor 2^{p + 1}r + 2 \rfloor} + 4\left(1 - \frac{\eta}{2^{p - 1}}\right)^{\lfloor 2^{p + 1}t\rfloor - \lfloor 2^{p + 1}r \rfloor}\\
    & = \mathcal{O}\left(\eta^2(t - r)^2\right) + 4\left[1 - \left(1 - \frac{\eta}{2^{p - 1}}\right)^{\lfloor 2^{p + 1}t \rfloor - \lfloor 2^{p + 1}s \rfloor}\right]\left[1 - \left(1 - \frac{\eta}{2^{p - 1}}\right)^{\lfloor 2^{p + 1}s \rfloor - \lfloor 2^{p + 1}r \rfloor}\right]\\
    & = \mathcal{O}\left(\eta^2(t - r)^2\right).
\end{align*}
\end{proof}
\end{lemma}

As for the convergence of the finite-dimensional characteristic functions, they follow straightforwardly from equation \ref{eq:pseudo_correlations_time_eigenbasis}:
\begin{lemma}
\label{lemma:fd_distributions_measuremtime_eigenbasis}
Let $n \geq 1$, $\theta_1, \ldots, \theta_n \in \mathbf{R}$ and $t_1, \ldots, t_n \in \mathbf{R}$. Then the following holds (with the convention $t_0 := 0$):
\begin{align}
    \left\langle \prod_{1 \leq k \leq n}\exp\left(2\pi i\theta_k\widetilde{\Xi}^{(d)}_{t_k}\right)\right\rangle & = e^{2\pi i\sum_{1 \leq k \leq n}\theta_kt_k}e^{-2\sum_{1 \leq k \leq n}\left|\sum_{k \leq l \leq n}\theta_l\right|(t_k - t_{k - 1})}\,.
\end{align}
\end{lemma}

\begin{theorem}
\label{th:convergence_measured_quasi_ideal_clock_cauchy_process}
As $p \to \infty$, $\left(\widetilde{\Xi}^{(4^p)}_t\right)_{t \geq 0}$ converges weakly to a Cauchy process with drift.
\begin{proof}
    This follows from \cite[theorem 13.5]{billingsley_convergence_1999}, using the previous two lemmas.
\end{proof}
\end{theorem}

\section{Some results from calculus}
\label{sec:calculus}

In this appendix, we review some results from calculus that are used in section \ref{sec:quasi_ideal_clock}. Most notably, is the review of some important properties of the Jacobi $\theta$ functions.

\subsection{Non-asymptotic Stirling's formula}
Stirling's formula allows to approximate the factorial (or more generally the $\Gamma$ function) by more elementary functions. In its best-known form, it states:
\begin{align}
    n! & \sim \left(\frac{n}{e}\right)^n\sqrt{2\pi n}\,, \qquad \textrm{ as }n \to \infty\,.
\end{align}
There exist many other versions of the results. Some of these are asymptotic expansions with a variable number of terms while others give explicit lower and upper bound for any finite $n$. One result of the latter kind, proved in \cite{robbins_1955}, will be of particular use:
\begin{theorem}
\label{th:stirling_bounds}
    For all integer $n \geq 1$,
    \begin{equation}
        \left(\frac{n}{e}\right)^n\sqrt{2\pi n}e^{\frac{1}{12n + 1}} \leq n! \leq \left(\frac{n}{e}\right)^n\sqrt{2\pi n}e^{\frac{1}{12n}}\,.
    \end{equation}
\end{theorem}
Actually, by convexity of $\mathbf{R}_+ \ni x \longmapsto \log(x!)$ (where $x!$ denotes $\Gamma(x + 1)$), $\mathbf{R}_+ \ni x \longmapsto \log\left(x^{x + 1/2}e^{-x + \frac{1}{12x + 1}}\right)$ and $\mathbf{R}_+ \ni x \longmapsto \log\left(x^{x + 1/2}e^{-x + \frac{1}{12x}}\right)$, this inequality holds even if $n$ is taken to be any positive real.

\subsection{Jacobi $\theta$ functions}
\label{sec:jacobi_theta_functions}

In this section, we review several useful properties of the $\theta$ functions which are of special importance in the study of the quasi-ideal clock. It will be sufficient to focus on the $\theta_3$ function.

The Jacobi $\theta_3$ function is defined for all $z \in \mathbf{C}$, $\tau \in \mathbf{H} = \{ it; t > 0 \}$\footnote{$\mathbf{H}$, the half-plane of complex numbers with positive imaginary part, is usually referred as the \textit{Poincar\'e half-plane}.}:
\begin{align}
\label{eq:theta3_definition}
    \theta_3(z, \tau) & := \sum_{m \in \mathbf{Z}}\exp\left(i\pi\tau m^2 + 2\pi imz\right)\,.
\end{align}
This function enjoys a nice quasi-periodicity property:
\begin{align}
\label{eq:theta3_quasi_periodicity}
    \theta_3(z + a + b\tau, \tau) & = \exp\left(-\pi ib^2\tau - 2\pi ibz\right)\theta_3(z, \tau)\,.
\end{align}
It is also clearly even in $z$:
\begin{align}
    \theta_3(z, \tau) & = \theta_3(-z, \tau)\,.
\end{align}
Using Poisson's summation formula, one can show the following useful transformation property:
\begin{align}
\label{eq:theta3_modular_transformation}
    \theta_3\left(z, -\frac{1}{\tau}\right) & = (-i\tau)^{1/2}\exp\left(i\pi z^2\tau\right)\theta_3(\tau z, \tau)\,.
\end{align}
Applying this identity to transform $\theta_3(z, i\sigma)$ ($\sigma > 0$), one obtains:
\begin{align}
    \theta_3(z, i\sigma) & = \frac{1}{\sigma^{1/2}}\sum_{m \in \mathbf{Z}}\exp\left(-\frac{\pi}{\sigma}(z + m)^2\right)\,.
\end{align}
Therefore, $\theta_3(z, i\sigma)$ can be interpreted as a Gaussian of square width $\frac{\sigma}{2\pi}$, periodized with period $1$. From this representation, combined with proposition \ref{prop:bound_tail_sum_gaussians}, follows the estimate for all $z \in (-1, 1)$:
\begin{align}
\label{eq:bound_tail_theta}
    \frac{1}{\sigma^{1/2}}e^{-\frac{\pi}{\sigma}z^2} \leq |\theta_3(z, i\sigma)| \leq \frac{1}{\sigma^{1/2}}\left(e^{-\frac{\pi}{\sigma}z^2} + \frac{e^{-\frac{\pi}{\sigma}(1 + z)^2}}{1 - e^{-\frac{2\pi}{\sigma}(1 + z)}} + \frac{e^{-\frac{\pi}{\sigma}(1 - z)^2}}{1 - e^{-\frac{2\pi}{\sigma}(1 - z)}}\right)\,.
\end{align}
Finally, another important expression for the $\theta_3$ function is given by the Jacobi triple product identity (see e.g \cite[page 49]{saltzer_1963} and \cite{andrews_1965} for a simple derivation):
\begin{align}
\label{eq:jacobi_triple_product}
    \theta_3\left(z, \tau\right) & = \prod_{n \geq 1}\left(1 - e^{2\pi i\tau}\right)\left(1 + 2\cos(2\pi z)e^{(2n - 1)i\tau} + e^{(4n - 2)i\tau}\right)\,.
\end{align}

The $\theta$ functions play for the discrete Fourier transform a similar role to the Gaussians for the continuous Fourier transform. Roughly speaking, the discrete Fourier transform of $\theta_3$ function of a given width is a $\theta_3$ of ``inversed'' width. More precisely, given a positive integer $N$ and $\xi > 0$, the following relations hold \cite{ruzzi_2006}:
\begin{align}
    \label{eq:theta3_inv_dft}
    \theta_3\left(z + \frac{k}{N}, \frac{i\xi^2}{N}\right) & = \frac{1}{\sqrt{N\xi^2}}\sum_{0 \leq j < N}\theta_3\left(\frac{iz}{\xi^2} - \frac{j}{N}, \frac{i}{N\xi^2}\right)\exp\left(-\frac{\pi N}{\xi^2}z^2 + \frac{2\pi ijk}{N}\right)\,,\\
    \label{eq:theta3_dft}
    \theta_3\left(\frac{iz}{\xi^2} - \frac{k}{N}, \frac{i}{N\xi^2}\right) & = \sqrt{\frac{\xi^2}{N}}\sum_{0 \leq j < N}\theta_3\left(z + \frac{j}{N}, \frac{i\xi^2}{N}\right)\exp\left(\frac{\pi N}{\xi^2}z^2 - \frac{2\pi ijk}{N}\right)\,.
\end{align}
In the sequel, we will be exclusively interested with $\theta_3$ where the second argument is purely imaginary (with positive imaginary part) and we will therefore frequently use the notation $\theta_3(z, i\tau)$ where $\tau > 0$ (instead as $\tau \in \mathbf{H}$ from the initial definitions).

An interesting property is that if one restricts the summation in the definition \ref{eq:theta3_definition} of the $\theta_3$ function to the integers that are congruent to some $r$ modulo $N$ ($N > 0$, $0 \leq r < N$), the result is still a $\theta$ function. Precisely:
\begin{align}
    & \sum_{p \in \mathbf{Z}}\exp\left(-\pi\tau(r + Np)^2 + 2\pi i(r + Np)z\right)\nonumber\\
    & = \exp\left(-\pi r^2\tau + 2\pi irz\right)\theta_3\left(zN + irN\tau, iN^2\tau\right)\label{eq:theta3_sum_congruence_class}\\
    & = \frac{1}{N\tau^{1/2}}\exp\left(-\frac{\pi z^2}{\tau}\right)\theta_3\left(\frac{iz}{N\tau} - \frac{r}{N}, \frac{i}{N^2\tau}\right)\,.
\end{align}

In our study of the quasi-ideal clock, we will use the $\theta_3$ function as wavefunction coefficients. It will therefore be frequently necessary (e.g. to compute scalar products) to rewrite products of $\theta$ functions as a linear combination of such functions. The following proposition gives a general formula for this purpose.
\begin{proposition}
\label{prop:theta3_multiplication}
    Let $a, b$ denote positive integers. Then for all $z, w \in \mathbf{C}, \tau > 0$:
    \begin{align}
        & \theta_3(z, iab\tau)\theta_3(w, i\tau)\nonumber\\
        & = \sum_{0 \leq r < a + b}\exp\left(-\pi r^2\tau + 2\pi irw\right)\theta_3\left(aw - z + ira\tau, ia(a + b)\tau\right)\theta_3\left(z + bw + irb\tau, ib(a + b)\tau\right)\\
        & = \frac{1}{a + b}\sum_{0 \leq r < a + b}\theta_3\left(\frac{aw + z + r}{a + b}, i\frac{a}{a + b}\tau\right)\theta_3\left(\frac{z - bw + r}{a + b}, i\frac{b}{a + b}\tau\right)\,.
    \end{align}
\begin{proof}
One starts with the basic definition of $\theta_3$:
\begin{align}
    \theta_3\left(z, iab\tau\right)\theta_3\left(w, i\tau\right) & = \sum_{\substack{m_1 \in \mathbf{Z}\\m_2 \in \mathbf{Z}}}\exp\left(-\pi\tau abm_1^2 + 2\pi im_1z - \pi\tau m_2^2 + 2\pi im_2w\right)\,.
\end{align}
One then remarks:
\begin{align}
    abm_1^2 + m_2^2 & = \frac{a}{a + b}(m_2 - bm_1)^2 + \frac{b}{a + b}(m_2 + am_1)^2\,.
\end{align}
Therefore:
\begin{align}
    & \theta_3\left(z, iab\tau\right)\theta_3\left(w, i\tau\right)\nonumber\\
    & = \sum_{m_1, m_2 \in \mathbf{Z}}\exp\left(-\pi\tau\frac{a}{a + b}(m_2 - bm_1)^2 - \pi\tau\frac{b}{a + b}(m_2 + am_1)^2 + 2\pi im_1z + 2\pi im_2w\right)\\
    & = \sum_{m_1, m_2 \in \mathbf{Z}}\exp\left(-\pi\tau\frac{a}{a + b}m_2^2 - \pi\tau\frac{b}{a + b}(m_2 + (a + b)m_1)^2 + 2\pi im_1(z + bw) + 2\pi im_2w\right)\\
    & = \sum_{m_1, m_2 \in \mathbf{Z}}\exp\left(-\pi\tau\frac{a}{a + b}m_2^2 - \pi\tau\frac{b}{a + b}(m_2 + (a + b)m_1)^2 + 2\pi i(m_2 + (a + b)m_1)\frac{z + bw}{a + b}\right.\nonumber\\
    & \left.\hspace{0.05\textwidth} + 2\pi im_2\frac{aw - z}{a + b}\right)\,.
\end{align}
We now break the sum according to the congruence $r$ of $m_2$ modulo $a + b$ (setting $m_2 := r + (a + b)p$) and for all fixed $m_2$, perform the sum first over $m_1$, then over $p$ using equation \ref{eq:theta3_sum_congruence_class}:
\begin{align}
    & \theta_3\left(z, iab\tau\right)\theta_3\left(w, i\tau\right)\nonumber\\
    & = \sum_{\substack{0 \leq r < a + b\\p \in \mathbf{Z}}}\exp\left(-\pi\tau\frac{a}{a + b}(r + (a + b)p)^2 + 2\pi i(r + (a + b)p)\frac{aw - z}{a + b}\right)\nonumber\\
    & \hspace{0.05\textwidth} \times \exp\left(-\pi r^2\frac{b}{a + b}\tau + 2\pi ir\frac{z + bw}{a + b}\right)\theta_3\left(z + bw + irb\tau, ib(a + b)\tau\right)\\
    & = \sum_{0 \leq r < a + b}\exp\left(-\pi r^2\frac{a}{a + b}\tau + 2\pi ir\frac{aw - z}{a + b}\right)\theta_3\left(aw - z + ira\tau, ia(a + b)\tau\right)\nonumber\\
    & \hspace{0.05\textwidth} \times \exp\left(-\pi r^2\frac{b}{a + b}\tau + 2\pi ir\frac{z + bw}{a + b}\right)\theta_3\left(z + bw + irb\tau, ib(a + b)\tau\right)\\
    & = \sum_{0 \leq r < a + b}\exp\left(-\pi r^2\tau + 2\pi irw\right)\theta_3\left(aw - z + ira\tau, ia(a + b)\tau\right)\theta_3\left(z + bw + irb\tau, ib(a + b)\tau\right)\,.
\end{align}
Now, observe that from the transformation property \ref{eq:theta3_modular_transformation},
\begin{align}
    \theta_3(z, iab\tau)\theta_3(w, i\tau) & = \frac{1}{(ab\tau)^{1/2}\tau^{1/2}}\exp\left(-\frac{\pi}{ab\tau}z^2 - \frac{\pi}{\tau}w^2\right)\theta_3\left(i\frac{z}{ab\tau}, \frac{i}{ab\tau}\right)\theta_3\left(i\frac{w}{\tau}, \frac{i}{\tau}\right)\,.
\end{align}
One may apply the identity just derived to the product $\theta_3\left(i\frac{w}{\tau}, \frac{i}{\tau}\right)\theta_3\left(i\frac{z}{ab\tau}, \frac{i}{ab\tau}\right)$, making the substitutions (in this order) $\tau \to \frac{1}{ab\tau}, z \to i\frac{w}{\tau}, w \to i\frac{z}{ab\tau}$. This gives:
\begin{align}
    \theta_3\left(z, iab\tau\right)\theta_3\left(w, i\tau\right) & = \frac{1}{(ab\tau)^{1/2}\tau^{1/2}}\exp\left(-\frac{\pi}{ab\tau}z^2 - \frac{\pi}{\tau}w^2\right)\nonumber\\
    & \hspace{0.05\textwidth} \times \sum_{0 \leq r < a + b}\exp\left(-\frac{\pi r^2}{ab\tau} - \frac{2\pi rz}{ab\tau}\right)\nonumber\\
    & \hspace{0.15\textwidth} \times \theta_3\left(i\frac{z + r - bw}{b\tau}, i\frac{a + b}{b\tau}\right)\theta_3\left(i\frac{aw + z + r}{a\tau}, i\frac{a + b}{a\tau}\right)\,.
\end{align}
Transforming again the $\theta_3$ functions above by property \ref{eq:theta3_modular_transformation}, the exponentials cancel and one obtains:
\begin{align}
    \theta_3\left(z, iab\tau\right)\theta_3\left(w, i\tau\right) & = \frac{1}{a + b}\sum_{0 \leq r < a + b}\theta_3\left(\frac{aw + z + r}{a + b}, i\frac{a}{a + b}\tau\right)\theta_3\left(\frac{z - bw + r}{a + b}, i\frac{b}{a + b}\tau\right)\,,
\end{align}
which is the stated result.
\end{proof}
\end{proposition}
This formula proves especially useful when $aw = -z$ or $bw = z$ since in this case only one of the $\theta$ functions in the sum depends on $z$. As a simple application of the previous results, one can compute (where $N > 0$ is an integer, $z \in \mathbf{R}$ and $\xi > 0$):
\begin{align}
    & \sum_{0 \leq j < N}\theta_3\left(z + \frac{j}{N}, \frac{i\xi^2}{N}\right)^2\nonumber\\
    & = \sum_{0 \leq j < N}\theta_3\left(z + \frac{j}{N}, \frac{i\xi^2}{N}\right)\theta_3\left(z + \frac{j}{N}, \frac{i\xi^2}{N}\right)\\
    & = \sum_{0 \leq j < N}\frac{1}{2}\left[\theta_3\left(z + \frac{j}{N}, \frac{i\xi^2}{2N}\right)\theta_3\left(0, \frac{i\xi^2}{2N}\right) + \theta_3\left(z + \frac{j}{N} + \frac{1}{2}, \frac{i\xi^2}{2N}\right)\theta_3\left(\frac{1}{2}, \frac{i\xi^2}{2N}\right)\right]\\
    & \hspace{0.1\textwidth} \qquad (\textrm{proposition \ref{prop:theta3_multiplication} with }
    a = b = 1)\nonumber\\
    & = \frac{1}{2}\sqrt{\frac{2N}{\xi^2}}\left[\exp\left(-\frac{2\pi N}{\xi^2}z^2\right)\theta_3\left(\frac{2iz}{\xi^2}, \frac{2i}{N\xi^2}\right)\theta_3\left(0, \frac{i\xi^2}{2N}\right)\right.\nonumber\\
    & \left.\hspace{0.15\textwidth} + \exp\left(-\frac{2\pi N}{\xi^2}\left(z + \frac{1}{2}\right)^2\right)\theta_3\left(\frac{2i}{\xi^2}\left(z + \frac{1}{2}\right), \frac{2i}{N\xi^2}\right)\theta_3\left(\frac{1}{2}, \frac{i\xi^2}{2N}\right)\right]\\
    & \hspace{0.1\textwidth} (\textrm{equation } \ref{eq:theta3_dft})\nonumber\\
    & = \frac{N}{2}\left[\theta_3\left(zN, \frac{iN\xi^2}{2}\right)\theta_3\left(0, \frac{i\xi^2}{2N}\right) + \theta_3\left(\left(z + \frac{1}{2}\right)N, \frac{iN\xi^2}{2}\right)\theta_3\left(\frac{1}{2}, \frac{i\xi^2}{2N}\right)\right]\,.
\end{align}

The following proposition is useful to approximate $\theta$ functions by Gaussians (which allows to bound the tail of a convolution of $\theta$ functions among other things):
\begin{proposition}
\label{prop:theta3_approx_1}
    Let $\tau > 0$ and $z \in \left(-\frac{1}{2}, \frac{1}{2}\right)$. Then the following holds:
    \begin{align}
        \label{eq:approx_theta_im_arg}
        \theta_3(i\tau z, i\tau) & = 1 + \varepsilon\,,\\
        |\varepsilon| & \leq 2\frac{e^{-\pi\tau(1 - 2|z|)}}{1 - e^{-2\pi\tau}}\,.
    \end{align}
\begin{proof}
The bound on $\varepsilon$ follows from proposition \ref{prop:bound_tail_sum_gaussians}.   
\begin{align}
    \varepsilon & := \sum_{m \neq 0}e^{-\pi\tau(m + z)^2 + \pi\tau z^2}\\
    & = \sum_{m \geq 1}e^{-\pi\tau(m + z)^2 + \pi\tau z^2} + \sum_{m \leq -1}e^{-\pi\tau(m + z)^2 + \pi\tau z^2}\\
    & \leq \frac{e^{-\pi\tau(z + 1)^2}e^{\pi\tau z^2}}{1 - e^{-\pi\tau(3 + 2z)}} + \frac{e^{-\pi\tau(1 - z)^2}e^{\pi\tau z^2}}{1 - e^{-\pi\tau(3 - 2z)}}\\
    & \leq 2\frac{e^{-\pi\tau(1 - 2|z|)}}{1 - e^{-2\pi\tau}}\,.&
\end{align}
\end{proof}
\end{proposition}

In the study of the measured quasi-ideal clock, one will consider random walks on a ring where the probability distributions for the jumps will be given by $\theta$ functions. Since we will be interested in estimating the marginal distribution of each step of this random walk, one needs an approximation for the iterated convolution of a $\theta$ function with itself. This is provided by the following proposition:
\begin{proposition}
\label{prop:estimate_convolution_theta}
    Given an odd integer $d > 0$, an integer $k \geq 1$ and a positive real $\tau > 0$, the following estimate holds for all integer $j \in \left[-\frac{d - 1}{2}, \frac{d - 1}{2}\right]$:
    \begin{align}
        & \sum_{-\frac{d - 1}{2} \leq j_1, \ldots, j_{k - 1} \leq \frac{d - 1}{2}}\left[\prod_{1 \leq l \leq k - 1}\theta_3\left(\frac{j_l}{d}, i\tau\right)\right]\theta_3\left(\frac{j - \sum_{1 \leq l \leq k - 1}j_l}{d}, i\tau\right)\nonumber\\
        & \hspace{0.05\textwidth} = d^{k - 1}\theta_3\left(\frac{j}{d}, ik\tau\right) + \varepsilon\,,\\
        |\varepsilon| & \leq d^{k - 1}(k - 1)\left(1 + \frac{2e^{-\pi\tau (k - 1)}}{1 - e^{-2\pi\tau (k - 1)}}\right)\frac{2e^{-\pi\tau\frac{(d + 1)^2}{4}}}{1 - e^{-\pi\tau d^2}}\left(1 + \frac{2e^{-\pi\tau d}}{1 - e^{-\pi\tau d^2}}\right)^{k - 2}\nonumber\\
        & \hspace{0.1\textwidth} + d^{k - 1}\frac{2e^{-\pi\tau\frac{(d + 1)^2}{4}}}{1 - e^{-\pi\tau d}}\left(1 + \frac{2e^{-\frac{\pi\tau d^2}{4}}}{1 - e^{-\pi\tau d^2}}\right)^{k - 1}\nonumber\\
        & \hspace{0.1\textwidth} + d^{k - 1}\frac{2e^{-\frac{\pi\tau kd^2}{2}}}{1 - e^{-\pi \tau k d}}\,.
    \end{align}
\begin{proof}
    \begin{align}
        & \sum_{-\frac{d - 1}{2} \leq j_1, \ldots, j_{k - 1} \leq \frac{d - 1}{2}}\left[\prod_{1 \leq l \leq k - 1}\theta_3\left(\frac{j_l}{d}, i\tau\right)\right]\theta_3\left(\frac{j - \sum_{1 \leq l \leq k - 1}j_l}{d}, i\tau\right)\nonumber\\
        & = \sum_{\substack{-\frac{d - 1}{2} \leq j_1, \ldots, j_{k - 1} \leq \frac{d - 1}{2}\\m_1, \ldots, m_k \in \mathbf{Z}}}\exp\left(-\pi\tau\sum_{1 \leq l \leq k}m_l^2 + \frac{2\pi i}{d}\sum_{1 \leq l \leq k - 1}m_lj_l + \frac{2\pi i}{d}m_k\left(j - \sum_{1 \leq l \leq k - 1}j_l\right)\right)\\
        & = d^{k - 1}\sum_{\substack{m_k \in \mathbf{Z}\\p_1, \ldots, p_{k - 1} \in \mathbf{Z}}}\exp\left(-\pi\tau m_k^2 - \pi\tau\sum_{1 \leq l \leq k - 1}(m_k + p_ld)^2 + \frac{2\pi i}{d}m_kj\right)\,.
    \end{align}
    Now, consider the summation over $p_1, \ldots, p_{k - 1}$ for some fixed $m_k, -\frac{d - 1}{2} \leq \frac{d - 1}{2}$:
    \begin{align}
        & \sum_{p_1, \ldots, p_{k - 1} \in \mathbf{Z}}\exp\left(-\pi\tau\sum_{1 \leq l \leq k - 1}(m_k + p_ld)^2\right)\nonumber\\
        & = \prod_{1 \leq l \leq k - 1}\sum_{p_l \in \mathbf{Z}}\exp\left(-\pi\tau d^2\left(\frac{m_k}{d} + p_l\right)^2\right)\\
        & = \prod_{1 \leq l \leq k - 1}\left(\exp(-\pi\tau m_k^2) + \sum_{p_l \in \mathbf{Z} - \{0\}}\exp\left(-\pi\tau d^2\left(\frac{m_k}{d} + p_l\right)^2\right)\right)\,.
    \end{align}
    For all $l$, one has:
    \begin{align}
        \sum_{p_l \in \mathbf{Z} - \{0\}}\exp\left(-\pi\tau d^2\left(\frac{m_k}{d} + p_l\right)\right) & \leq \frac{2e^{-\frac{\pi\tau (d + 1)^2}{4}}}{1 - e^{-\pi\tau d^2}}
    \end{align}
    by proposition \ref{prop:bound_tail_sum_gaussians}, so that
    \begin{align}
        & \left|\sum_{p_1, \ldots, p_{k - 1} \in \mathbf{Z}}\exp\left(-\pi\tau\sum_{1 \leq l \leq k - 1}(m_k + p_ld)^2\right) - e^{-\pi\tau(k  - 1)m_k^2}\right|\nonumber\\
        & \leq \left(e^{-\pi\tau m_k^2} + \frac{2e^{-\frac{\pi\tau(d + 1)^2}{4}}}{1 - e^{-\pi\tau d^2}}\right)^{k - 1} - e^{-\pi\tau(k - 1)m_k^2}\,,
    \end{align}
    and
    \begin{align}
        & \left|d^{k - 1}\sum_{p_1, \ldots, p_{k - 1} \in \mathbf{Z}}\exp\left(-\pi\tau m_k^2 - \pi\tau\sum_{1 \leq l \leq k - 1}(m_k + p_ld)^2 + \frac{2\pi im_kj}{d}\right) - d^{k - 1}e^{-\pi\tau km_k^2 + \frac{2\pi im_kj}{d}}\right|\nonumber\\
        & \leq d^{k - 1}e^{-\pi\tau m_k^2}\left[\left(e^{-\pi\tau m_k^2} + \frac{2e^{-\frac{\pi\tau(d + 1)^2}{4}}}{1 - e^{-\pi\tau d^2}}\right) - e^{-\pi\tau(k - 1)m_k^2}\right]\\
        & \leq d^{k - 1}e^{-\pi\tau m_k^2}(k - 1)\frac{2e^{-\pi\tau\frac{(d + 1)^2}{4}}}{1 - e^{-\pi\tau d^2}}\left(e^{-\pi\tau m_k^2} + \frac{2e^{-\pi\tau\frac{(d + 1)^2}{4}}}{1 - e^{-\pi\tau d^2}}\right)^{k - 2}\\
        & = d^{k - 1}(k - 1)e^{-\pi\tau(k - 1)m_k^2}\frac{2e^{-\pi\tau\frac{(d + 1)^2}{4}}}{1 - e^{-\pi\tau d^2}}\left(1 + \frac{2e^{-\pi\tau\frac{(d + 1)^2}{4} + \pi\tau m_k^2}}{1 - e^{-\pi\tau d^2}}\right)^{k - 2}\\
        & \leq d^{k - 1}(k - 1)e^{-\pi\tau(k - 1)m_k^2}\frac{2e^{-\pi\tau\frac{(d + 1)^2}{4}}}{1 - e^{-\pi\tau d^2}}\left(1 + \frac{2e^{-\pi\tau d}}{1 - e^{-\pi\tau d^2}}\right)^{k - 2}\,.
    \end{align}
    Summing this error over all $m_k, -\frac{d - 1}{2} \leq m_k \leq \frac{d - 1}{2}$, one obtains an error which is bounded by:
    \begin{align}
        \varepsilon_1 & = d^{k - 1}(k - 1)\left(1 + \frac{2e^{-\pi\tau (k - 1)}}{1 - e^{-2\pi\tau (k - 1)}}\right)\frac{2e^{-\pi\tau\frac{(d + 1)^2}{4}}}{1 - e^{-\pi\tau d^2}}\left(1 + \frac{2e^{-\pi\tau d}}{1 - e^{-\pi\tau d^2}}\right)^{k - 2}\,.
    \end{align}
    Next, for fixed $m_k$ such that $|m_k| \geq \frac{d + 1}{2}$, the sum
    \begin{align}
        & d^{k - 1}\sum_{p_1, \ldots, p_{k - 1} \in \mathbf{Z}}\exp\left(-\pi\tau m_k^2 -\pi\tau\sum_{1 \leq l \leq k - 1}(m_k + p_ld)^2 + \frac{2\pi im_kj}{d}\right)
    \end{align}
    can be upper bounded by
    \begin{align}
        & d^{k - 1}e^{-\pi\tau m_k^2}\left(1 + \frac{2e^{-\frac{\pi\tau d^2}{4}}}{1 - e^{-\pi\tau d^2}}\right)^{k - 1}
    \end{align}
    and summing over all such $m_k$, one obtains the following bound on the error
    \begin{align}
        \varepsilon_2 & = d^{k - 1}\frac{2e^{-\pi\tau\frac{(d + 1)^2}{4}}}{1 - e^{-\pi\tau d}}\left(1 + \frac{2e^{-\frac{\pi\tau d^2}{4}}}{1 - e^{-\pi\tau d^2}}\right)^{k - 1}\,.
    \end{align}
    Finally, one can easily bound:
    \begin{align}
        \left|d^{k - 1}\sum_{|m_k| \geq \frac{d + 1}{2}}e^{-\pi\tau km_k^2 + \frac{2\pi im_kj}{d}}\right| & \leq d^{k - 1}\frac{2e^{-\frac{\pi\tau kd^2}{2}}}{1 - e^{-\pi \tau k d}}\\
        & =: \varepsilon_3\,.
    \end{align}
    Putting all of this together, we have shown that up to an error $\varepsilon_1 + \varepsilon_2 + \varepsilon_3$
    \begin{align}
        & \sum_{-\frac{d - 1}{2} \leq j_1, \ldots, j_{k - 1} \leq \frac{d - 1}{2}}\left[\prod_{1 \leq l \leq k - 1}\theta_3\left(\frac{j_l}{d}, i\tau\right)\right]\theta_3\left(\frac{j - \sum_{1 \leq l \leq k - 1}j_l}{d}, i\tau\right)
    \end{align}
    amounts to
    \begin{align}
        & d^{k - 1}\sum_{m_k \in \mathbf{Z}}e^{-\pi\tau km_k^2 + \frac{2\pi ijm_k}{d}}\\
        & = d^{k - 1}\theta_3\left(\frac{j}{d}, ik\tau\right)\,.
    \end{align}
\end{proof}
\end{proposition}

\subsection{Miscellaneous}
In this section, we state various results that will be used either in section \ref{sec:quasi_ideal_clock} or to prove some properties of Jacobi $\theta$ functions in section \ref{sec:jacobi_theta_functions}.

\begin{proposition}[Bounds on the tails of sums of Gaussians]
\label{prop:bound_tail_sum_gaussians}
    Let $\alpha, \beta > 0$ and $z \in \mathbf{R}$. Then for all integer $a > (-z) \vee 0$,
    \begin{align}
        \sum_{n \geq a}n^{\beta}e^{-\alpha (z + n)^2} & \leq e^{\frac{\beta^2}{4a^2\alpha}}\frac{a^{\beta}e^{-\alpha (a + z)^2}}{1 - e^{-2\alpha (a + z)}}\,.
    \end{align}
    If furthermore $a > \frac{\beta}{\alpha}$,
    \begin{align}
        \sum_{n \geq a}n^{\beta}e^{-\alpha (z + n)^2} & \leq \frac{a^{\beta}e^{-\alpha (a + z)^2}}{1 - e^{-2\alpha (a + z)}}\,.
    \end{align}
\begin{proof}
    Letting $n := a + k, k \geq 0$,
    \begin{align}
        \log\left(n^{\beta}e^{-\alpha(z + n)^2}\right) & = \beta \log(n) - \alpha (z + n)^2\\
        & = \beta\log(a) + \beta\log\left(1 + \frac{k}{a}\right) - \alpha(a + z)^2 - 2\alpha(a + z)k - \alpha k^2\\
        & \leq \beta\log(a) - \alpha(a + z)^2 - 2\alpha(a + z)k + \beta\frac{k}{a} - \alpha k^2\\
        & = \beta\log(a) - \alpha(a + z)^2 - 2\alpha(a + z)k - \alpha k\left(k - \frac{\beta}{a\alpha}\right)\,.
    \end{align}
    Generally speaking, the last term is upper-bounded by $\frac{\beta^2}{4a^2\alpha}$. If $a \geq \frac{\beta}{\alpha}$, it is $\leq 0$. Summing the exponential of this over $k \geq 0$ gives the result.
\end{proof}
\end{proposition}

\begin{proposition}
    For $\alpha, \beta, c, \varepsilon > 0$, the following holds:
\begin{equation}
\label{eq:power_vs_exp_inequality}
    d^{\beta}\exp(-cd^{\alpha}) < \varepsilon \textrm{ if } d > \max\left\{0, \frac{\beta}{\alpha c}\left(1 + \frac{1}{a_0}\log\left(\frac{\beta}{\alpha c}\right)\right) - \frac{\log(\varepsilon)}{a_0c}\right\}^{\frac{1}{\alpha}}\textrm{ or } \varepsilon > \left(\frac{\beta}{\alpha ce}\right)^{\frac{\beta}{\alpha}}\,,
\end{equation}
where
\begin{equation}
    a_0 := 1 + W(-e^{-2}) \approx 0.841\,;
\end{equation}
with $W$ denoting the Lambert $W$ function.
\begin{proof}
    By rescaling, it suffices to treat the case $c = 1, \alpha = 0$. The function:
    \begin{align}
        x & \longmapsto \log\left(x^{\beta}\exp(-x)\right) = \beta\log(x) - x
    \end{align}
    is concave and attains its maximum at $x = \beta$. Therefore, if $\varepsilon > \left(\frac{\beta}{e}\right)^{\beta}$, the inequality
    \begin{align}
        d^{\beta}\exp(-d) & < \varepsilon
    \end{align}
    holds for all $d > 0$. Otherwise, let us set $x := \beta(1 + t), t \geq 0$ and look for a lower bound on $t$ such that for all $t$ greater than this lower bound: $\beta\log(x) - x < \log(\varepsilon)$ is satisfied. For any $t_0 > 0$, one may write (bouding the concave function $t \longmapsto \log(1 + t) - t - 1$ by its tangent at $t = t_0$):
    \begin{align}
        \beta\log(x) - x & = \beta\log(\beta) + \beta\left(\log(1 + t) - t - 1\right)\\
        & \leq \beta\log(\beta) + \beta\left(\log(1 + t_0) - \frac{1 + 2t_0}{1 + t_0} - \frac{t_0}{1 + t_0}t\right)\,.
    \end{align}
    Therefore, $\beta\log(x) - x$ is certainly smaller than $\log(\varepsilon)$ if the r.h.s of the inequality is, i.e. if
    \begin{align}
        t & > \frac{1 + t_0}{t_0}\left(\log(\beta) - \frac{\log(\varepsilon)}{\beta} + \log(1 + t_0) - \frac{1 + 2t_0}{1 + t_0}\right)\,.
    \end{align}
    It is convenient to choose $t_0 > 0$ such that $\log(1 + t_0) - \frac{1 + 2t_0}{1 + t_0}$. This is achieved for $t_0 := e^{2 + W\left(e^{-2}\right)}$, for which $\frac{t_0}{1 + t_0} = a_0 =  1 + W\left(e^{-2}\right)$. Therefore, for $x > \beta\left(1 + \frac{1}{a_0}\left(\log(\beta) - \frac{\log(\varepsilon)}{\beta}\right)\right)$
    we do have $\beta\log(x) - x < \log(\varepsilon)$ i.e. $x^{\beta}\exp(-x) < \varepsilon$.
\end{proof}
\end{proposition}

\begin{proposition}
    For all $x > 0$,
    \begin{equation}
    \label{eq:1_over_1_minus_exp_bounds}
        \frac{1}{x} + \frac{1}{2} < \frac{1}{1 - e^{-x}} < \frac{1}{x} + \frac{1}{2} + \frac{x}{12}\,.
    \end{equation}
\begin{proof}
    Let us first prove the lower bound. First, observe $\frac{\mathrm{d}}{\mathrm{d}x}\left(\frac{1}{1 - e^{-x}} - \frac{1}{x} - \frac{1}{2}\right) = \frac{1}{x^2} + \frac{1}{2 - 2\cosh(x)}$. Using the power series expansion of $\cosh$, this is straightforwardly $> 0$. Since furthermore $\frac{1}{1 - e^{-x}} - \frac{1}{x} - \frac{1}{2}$ is $0$ at $x = 0$, the lower bound follows.
    
    For the upper bound, $\frac{\mathrm{d}}{\mathrm{d}x}\left(\frac{1}{1 - e^{-x}} - \frac{1}{x} - \frac{1}{2} - \frac{x}{12}\right) = \frac{12 - x^2}{12x^2} + \frac{1}{2 - 2\cosh(x)}$. If $x > \sqrt{12}$, both terms in this sum are negative so the derivative is negative. For $x < \sqrt{12}$,
    \begin{align}
        \left(\frac{12 - x^2}{12x^2}\right)^{-1} & = \sum_{k \geq 0}\frac{x^{2(k + 1)}}{(12)^k}\\
        \left(\frac{1}{2 - 2\cosh(x)}\right)^{-1} & = -2\sum_{k \geq 1}\frac{x^{2k}}{(2k)!}
    \end{align}
\end{proof}
\end{proposition}
But 
\begin{align}
    2\sum_{k \geq 1}\frac{x^{2k}}{(2k)!} - \sum_{k \geq 0}\frac{x^{2(k + 1)}}{(12)^k} & = 2\sum_{k \geq 0}\frac{x^{2k + 2}}{(2k + 2)!} - \sum_{k \geq 0}\frac{x^{2(k + 1)}}{(12)^k}\\
    & \leq 2\sum_{k \geq 4}\frac{x^{2k + 2}}{(2k + 2)!} - \sum_{k \geq 4}\frac{x^{2(k + 1)}}{(12)^k}\\
    & \leq 2\sum_{k \geq 4}\frac{x^{2(k + 1)}}{e^{\frac{1}{24(k + 1) + 1}}\left(\frac{2k + 2}{e}\right)^{2k + 2}\sqrt{2\pi(2k + 2)}} - \sum_{k \geq 4}\frac{x^{2(k + 1)}}{(12)^k}\\
    & \leq \frac{2}{\sqrt{2\pi \times 10}\left(\frac{10}{e}\right)^2}\sum_{k \geq 4}\frac{x^{2(k + 1)}}{\left(\frac{10}{e}\right)^{2k}} - \sum_{k \geq 4}\frac{x^{2(k + 1)}}{(12)^k}\\
    & \leq 0
\end{align}
This shows that $\frac{\mathrm{d}}{\mathrm{d}x}\left(\frac{1}{1 - e^{-x}} - \frac{1}{x} - \frac{1}{2} - \frac{x}{12}\right) = \frac{12 - x^2}{12x^2} + \frac{1}{2 - 2\cosh(x)}$ also for $x > 12$. Therefore, $\frac{1}{1 - e^{-x}} - \frac{1}{x} - \frac{1}{2} - \frac{x}{12}$ decreases with $x$ for $x > 0$ and the upper bound follows.

\begin{proposition}
    For all $x > 0$,
    \begin{align}
    \label{eq:exp_over_1_minus_expr_bound}
        \frac{e^{-x}}{1 - e^{-2x}} & \geq \frac{1}{2x} - \frac{x}{12}\,.
    \end{align}
\begin{proof}
    For all $x > 0$,
    \begin{align}
        \left(\frac{1}{2x} - \frac{x}{12}\right)^{-1} & = \frac{12x}{6 - x^2}\\
        & = 2\sum_{k \geq 0}\frac{x^{2k + 1}}{6^k}
    \end{align}
    \begin{align}
        \left(\frac{e^{-x}}{1 - e^{-2x}}\right)^{-1} & = 2\sinh(x)\\
        & = 2\sum_{k \geq 0}\frac{x^{2k + 1}}{(2k + 1)!}
    \end{align}
    But for $k \geq 3$,
    \begin{align}
        (2k + 1)! & \geq (2k + 1)^{2k + 1}e^{-(2k + 1)}\sqrt{2\pi(2k + 1)}\\
        & \geq 7\sqrt{2\pi 7}\left(\frac{7}{e}\right)^{2k}\\
        & \geq 6^k
    \end{align}
    and this is initially seen to hold for $k \in \{ 0, 1, 2 \}$ as well. This shows that $\left(\frac{e^{-x}}{1 - e^{-2x}}\right)^{-1} \leq \left(\frac{1}{2x} - \frac{x}{12}\right)^{-1}$, which is the claim.
\end{proof}
\end{proposition}

\begin{proposition}
    For all $x \in \left(0, \frac{1}{2^{3/4}}\right)$ (for example), the following holds:
    \begin{align}
        \label{eq:log_1_minus_x_bounds}
        -x - x^2 & \leq \log(1 - x) \leq -x - \frac{x^2}{2}\,.
    \end{align}
    \begin{proof}
        This follows easily from the series expansion $\ln(1 - x) = -\sum_{k \geq 1}\frac{x^k}{k}$.
    \end{proof}
\end{proposition}

\begin{proposition}
    For all $x > \frac{\log(3)}{2}$,
    \begin{align}
    \label{eq:inequality_tanh}
        \tanh(x) < 1 - \frac{3}{2}e^{-2x}\,.
    \end{align}
    \begin{proof}
        This is immediate from $\tanh(x) = \frac{1 - e^{-2x}}{1 + e^{-2x}}$.
    \end{proof}
\end{proposition}

\bibliographystyle{apsrev}
\bibliography{bibliography}

\begin{thebibliography}{9}%
\makeatletter
\providecommand \@ifxundefined [1]{%
 \@ifx{#1\undefined}
}%
\providecommand \@ifnum [1]{%
 \ifnum #1\expandafter \@firstoftwo
 \else \expandafter \@secondoftwo
 \fi
}%
\providecommand \@ifx [1]{%
 \ifx #1\expandafter \@firstoftwo
 \else \expandafter \@secondoftwo
 \fi
}%
\providecommand \natexlab [1]{#1}%
\providecommand \enquote  [1]{``#1''}%
\providecommand \bibnamefont  [1]{#1}%
\providecommand \bibfnamefont [1]{#1}%
\providecommand \citenamefont [1]{#1}%
\providecommand \href@noop [0]{\@secondoftwo}%
\providecommand \href [0]{\begingroup \@sanitize@url \@href}%
\providecommand \@href[1]{\@@startlink{#1}\@@href}%
\providecommand \@@href[1]{\endgroup#1\@@endlink}%
\providecommand \@sanitize@url [0]{\catcode `\\12\catcode `\$12\catcode
  `\&12\catcode `\#12\catcode `\^12\catcode `\_12\catcode `\%12\relax}%
\providecommand \@@startlink[1]{}%
\providecommand \@@endlink[0]{}%
\providecommand \url  [0]{\begingroup\@sanitize@url \@url }%
\providecommand \@url [1]{\endgroup\@href {#1}{\urlprefix }}%
\providecommand \urlprefix  [0]{URL }%
\providecommand \Eprint [0]{\href }%
\providecommand \doibase [0]{http://dx.doi.org/}%
\providecommand \selectlanguage [0]{\@gobble}%
\providecommand \bibinfo  [0]{\@secondoftwo}%
\providecommand \bibfield  [0]{\@secondoftwo}%
\providecommand \translation [1]{[#1]}%
\providecommand \BibitemOpen [0]{}%
\providecommand \bibitemStop [0]{}%
\providecommand \bibitemNoStop [0]{.\EOS\space}%
\providecommand \EOS [0]{\spacefactor3000\relax}%
\providecommand \BibitemShut  [1]{\csname bibitem#1\endcsname}%
\let\auto@bib@innerbib\@empty
\bibitem [{\citenamefont {Woods}\ \emph {et~al.}(2019)\citenamefont {Woods},
  \citenamefont {Silva},\ and\ \citenamefont
  {Oppenheim}}]{woods_autonomous_2019}%
  \BibitemOpen
  \bibfield  {author} {\bibinfo {author} {\bibfnamefont {M.~P.}\ \bibnamefont
  {Woods}}, \bibinfo {author} {\bibfnamefont {R.}~\bibnamefont {Silva}}, \ and\
  \bibinfo {author} {\bibfnamefont {J.}~\bibnamefont {Oppenheim}},\ }\href
  {\doibase 10.1007/s00023-018-0736-9} {\bibfield  {journal} {\bibinfo
  {journal} {Annales Henri Poincar{\'e}}\ }\textbf {\bibinfo {volume} {20}},\
  \bibinfo {pages} {125} (\bibinfo {year} {2019})};\ \Eprint
  {http://arxiv.org/abs/1607.04591} {arXiv:1607.04591 [quant-ph]} \BibitemShut
  {NoStop}%
\bibitem [{\citenamefont {Woods}\ \emph {et~al.}(2018)\citenamefont {Woods},
  \citenamefont {Silva}, \citenamefont {P{\"u}tz}, \citenamefont {Stupar},\
  and\ \citenamefont {Renner}}]{woods_quantum_2018}%
  \BibitemOpen
  \bibfield  {author} {\bibinfo {author} {\bibfnamefont {M.~P.}\ \bibnamefont
  {Woods}}, \bibinfo {author} {\bibfnamefont {R.}~\bibnamefont {Silva}},
  \bibinfo {author} {\bibfnamefont {G.}~\bibnamefont {P{\"u}tz}}, \bibinfo
  {author} {\bibfnamefont {S.}~\bibnamefont {Stupar}}, \ and\ \bibinfo {author}
  {\bibfnamefont {R.}~\bibnamefont {Renner}},\ }\Eprint {http://arxiv.org/abs/1806.00491} {arXiv:1806.00491
  [quant-ph]} \BibitemShut {NoStop}%
\bibitem [{\citenamefont {Braginsky}\ and\ \citenamefont
  {Khalili}(1992)}]{braginsky_quantum_1992}%
  \BibitemOpen
  \bibfield  {author} {\bibinfo {author} {\bibfnamefont {V.~B.}\ \bibnamefont
  {Braginsky}}\ and\ \bibinfo {author} {\bibfnamefont {F.~Y.}\ \bibnamefont
  {Khalili}},\ }\href {\doibase 10.1017/CBO9780511622748} {\emph {\bibinfo
  {title} {Quantum {{Measurement}}}}}\ (\bibinfo  {publisher} {{Cambridge
  University Press}},\ \bibinfo {year} {1992})\BibitemShut {NoStop}%
\bibitem [{\citenamefont {Diaconis}(1988)}]{diaconis_group_1988}%
  \BibitemOpen
  \bibfield  {author} {\bibinfo {author} {\bibfnamefont {P.}~\bibnamefont
  {Diaconis}},\ }\href {\doibase 10.1214/lnms/1215467412} {\emph {\bibinfo
  {title} {Group Representations in Probability and Statistics}}},\ \bibinfo
  {series} {Lecture Notes-Monograph Series}, Vol.~\bibinfo {volume} {11}\
  (\bibinfo  {publisher} {{Institute of Mathematical Statistics}},\ \bibinfo
  {address} {{Hayward, Calif}},\ \bibinfo {year} {1988})\BibitemShut {NoStop}%
\bibitem [{\citenamefont {Billingsley}(1999)}]{billingsley_convergence_1999}%
  \BibitemOpen
  \bibfield  {author} {\bibinfo {author} {\bibfnamefont {P.}~\bibnamefont
  {Billingsley}},\ }\href {\doibase 10.1002/9780470316962} {\emph {\bibinfo
  {title} {Convergence of {{Probability Measures}}}}},\ Wiley {{Series}} in
  {{Probability}} and {{Statistics}}\ (\bibinfo  {publisher} {{John Wiley \&
  Sons, Inc.}},\ \bibinfo {address} {{Hoboken, NJ, USA}},\ \bibinfo {year}
  {1999})\BibitemShut {NoStop}%
\bibitem [{\citenamefont {Robbins}(1955)}]{robbins_1955}%
  \BibitemOpen
  \bibfield  {author} {\bibinfo {author} {\bibfnamefont {H.}~\bibnamefont
  {Robbins}},\ }\href {\doibase 10.2307/2308012} {\bibfield  {journal}
  {\bibinfo  {journal} {The American Mathematical Monthly}\ }\textbf {\bibinfo
  {volume} {62}},\ \bibinfo {pages} {26} (\bibinfo {year} {1955})}\BibitemShut
  {NoStop}%
\bibitem [{\citenamefont {Saltzer}(1963)}]{saltzer_1963}%
  \BibitemOpen
  \bibfield  {author} {\bibinfo {author} {\bibfnamefont {C.}~\bibnamefont
  {Saltzer}},\ }\href {\doibase 10.1137/1005047} {\bibfield  {journal}
  {\bibinfo  {journal} {SIAM Review}\ }\textbf {\bibinfo {volume} {5}},\
  \bibinfo {pages} {162–162} (\bibinfo {year} {1963})}\BibitemShut {NoStop}%
\bibitem [{\citenamefont {Andrews}(1965)}]{andrews_1965}%
  \BibitemOpen
  \bibfield  {author} {\bibinfo {author} {\bibfnamefont {G.~E.}\ \bibnamefont
  {Andrews}},\ }\href {\doibase 10.1090/s0002-9939-1965-0171725-x} {\bibfield
  {journal} {\bibinfo  {journal} {Proceedings of the American Mathematical
  Society}\ }\textbf {\bibinfo {volume} {16}},\ \bibinfo {pages} {333–333}
  (\bibinfo {year} {1965})}\BibitemShut {NoStop}%
\bibitem [{\citenamefont {Ruzzi}(2006)}]{ruzzi_2006}%
  \BibitemOpen
  \bibfield  {author} {\bibinfo {author} {\bibfnamefont {M.}~\bibnamefont
  {Ruzzi}},\ }\href {\doibase 10.1063/1.2209770} {\bibfield  {journal}
  {\bibinfo  {journal} {Journal of Mathematical Physics}\ }\textbf {\bibinfo
  {volume} {47}},\ \bibinfo {pages} {063507} (\bibinfo {year} {2006})};\
  \Eprint {http://arxiv.org/abs/math-ph/0604018} {arXiv:math-ph/0604018}
  \BibitemShut {NoStop}%
\end{thebibliography}%


\begin{thebibliography}{35}%
\makeatletter
\providecommand \@ifxundefined [1]{%
 \@ifx{#1\undefined}
}%
\providecommand \@ifnum [1]{%
 \ifnum #1\expandafter \@firstoftwo
 \else \expandafter \@secondoftwo
 \fi
}%
\providecommand \@ifx [1]{%
 \ifx #1\expandafter \@firstoftwo
 \else \expandafter \@secondoftwo
 \fi
}%
\providecommand \natexlab [1]{#1}%
\providecommand \enquote  [1]{``#1''}%
\providecommand \bibnamefont  [1]{#1}%
\providecommand \bibfnamefont [1]{#1}%
\providecommand \citenamefont [1]{#1}%
\providecommand \href@noop [0]{\@secondoftwo}%
\providecommand \href [0]{\begingroup \@sanitize@url \@href}%
\providecommand \@href[1]{\@@startlink{#1}\@@href}%
\providecommand \@@href[1]{\endgroup#1\@@endlink}%
\providecommand \@sanitize@url [0]{\catcode `\\12\catcode `\$12\catcode
  `\&12\catcode `\#12\catcode `\^12\catcode `\_12\catcode `\%12\relax}%
\providecommand \@@startlink[1]{}%
\providecommand \@@endlink[0]{}%
\providecommand \url  [0]{\begingroup\@sanitize@url \@url }%
\providecommand \@url [1]{\endgroup\@href {#1}{\urlprefix }}%
\providecommand \urlprefix  [0]{URL }%
\providecommand \Eprint [0]{\href }%
\providecommand \doibase [0]{http://dx.doi.org/}%
\providecommand \selectlanguage [0]{\@gobble}%
\providecommand \bibinfo  [0]{\@secondoftwo}%
\providecommand \bibfield  [0]{\@secondoftwo}%
\providecommand \translation [1]{[#1]}%
\providecommand \BibitemOpen [0]{}%
\providecommand \bibitemStop [0]{}%
\providecommand \bibitemNoStop [0]{.\EOS\space}%
\providecommand \EOS [0]{\spacefactor3000\relax}%
\providecommand \BibitemShut  [1]{\csname bibitem#1\endcsname}%
\let\auto@bib@innerbib\@empty
\bibitem [{\citenamefont {Braginsky}\ and\ \citenamefont
  {Khalili}(1992)}]{braginsky_quantum_1992}%
  \BibitemOpen
  \bibfield  {author} {\bibinfo {author} {\bibfnamefont {V.~B.}\ \bibnamefont
  {Braginsky}}\ and\ \bibinfo {author} {\bibfnamefont {F.~Y.}\ \bibnamefont
  {Khalili}},\ }\href {\doibase 10.1017/CBO9780511622748} {\emph {\bibinfo
  {title} {Quantum {{Measurement}}}}}\ (\bibinfo  {publisher} {{Cambridge
  University Press}},\ \bibinfo {year} {1992})\BibitemShut {NoStop}%
\bibitem [{\citenamefont {Kippenberg}\ and\ \citenamefont
  {Vahala}(2008)}]{kippenberg_cavity_2008}%
  \BibitemOpen
  \bibfield  {author} {\bibinfo {author} {\bibfnamefont {T.~J.}\ \bibnamefont
  {Kippenberg}}\ and\ \bibinfo {author} {\bibfnamefont {K.~J.}\ \bibnamefont
  {Vahala}},\ }\href {\doibase 10.1126/science.1156032} {\bibfield  {journal}
  {\bibinfo  {journal} {Science}\ }\textbf {\bibinfo {volume} {321}},\ \bibinfo
  {pages} {1172} (\bibinfo {year} {2008})}\BibitemShut {NoStop}%
\bibitem [{\citenamefont {Cripe}\ \emph {et~al.}(2019)\citenamefont {Cripe},
  \citenamefont {Aggarwal}, \citenamefont {Lanza}, \citenamefont {Libson},
  \citenamefont {Singh}, \citenamefont {Heu}, \citenamefont {Follman},
  \citenamefont {Cole}, \citenamefont {Mavalvala},\ and\ \citenamefont
  {Corbitt}}]{cripe_measurement_2019}%
  \BibitemOpen
  \bibfield  {author} {\bibinfo {author} {\bibfnamefont {J.}~\bibnamefont
  {Cripe}}, \bibinfo {author} {\bibfnamefont {N.}~\bibnamefont {Aggarwal}},
  \bibinfo {author} {\bibfnamefont {R.}~\bibnamefont {Lanza}}, \bibinfo
  {author} {\bibfnamefont {A.}~\bibnamefont {Libson}}, \bibinfo {author}
  {\bibfnamefont {R.}~\bibnamefont {Singh}}, \bibinfo {author} {\bibfnamefont
  {P.}~\bibnamefont {Heu}}, \bibinfo {author} {\bibfnamefont {D.}~\bibnamefont
  {Follman}}, \bibinfo {author} {\bibfnamefont {G.~D.}\ \bibnamefont {Cole}},
  \bibinfo {author} {\bibfnamefont {N.}~\bibnamefont {Mavalvala}}, \ and\
  \bibinfo {author} {\bibfnamefont {T.}~\bibnamefont {Corbitt}},\ }\href
  {\doibase 10.1038/s41586-019-1051-4} {\bibfield  {journal} {\bibinfo
  {journal} {Nature}\ }\textbf {\bibinfo {volume} {568}},\ \bibinfo {pages}
  {364} (\bibinfo {year} {2019})}\BibitemShut {NoStop}%
\bibitem [{\citenamefont {{LIGO Scientific
  Collaboration}}(2013)}]{ligo_scientific_collaboration_enhanced_2013}%
  \BibitemOpen
  \bibfield  {author} {\bibinfo {author} {\bibnamefont {{LIGO Scientific
  Collaboration}}},\ }\href {\doibase 10.1038/nphoton.2013.177} {\bibfield
  {journal} {\bibinfo  {journal} {Nature Photonics}\ }\textbf {\bibinfo
  {volume} {7}},\ \bibinfo {pages} {613} (\bibinfo {year} {2013})}\BibitemShut
  {NoStop}%
\bibitem [{lig(2019)}]{ligopr}%
  \BibitemOpen
  \href@noop {} {\enquote {\bibinfo {title} {{LIGO} and {Virgo} resume search
  for ripples in space and time},}\ }\bibinfo {howpublished}
  {\href{https://www.ligo.caltech.edu/news/ligo20190326}{www.ligo.caltech.edu/news/ligo20190326}}
  (\bibinfo {year} {2019})\BibitemShut {NoStop}%
\bibitem [{\citenamefont {Caves}(1981)}]{caves_quantum-mechanical_1981}%
  \BibitemOpen
  \bibfield  {author} {\bibinfo {author} {\bibfnamefont {C.~M.}\ \bibnamefont
  {Caves}},\ }\href {\doibase 10.1103/PhysRevD.23.1693} {\bibfield  {journal}
  {\bibinfo  {journal} {Physical Review D}\ }\textbf {\bibinfo {volume} {23}},\
  \bibinfo {pages} {1693} (\bibinfo {year} {1981})}\BibitemShut {NoStop}%
\bibitem [{\citenamefont {Braginski{\u \i}}\ and\ \citenamefont
  {Vorontsov}(1975)}]{braginskii_quantum-mechanical_1975}%
  \BibitemOpen
  \bibfield  {author} {\bibinfo {author} {\bibfnamefont {V.~B.}\ \bibnamefont
  {Braginski{\u \i}}}\ and\ \bibinfo {author} {\bibfnamefont {Y.~I.}\
  \bibnamefont {Vorontsov}},\ }\href {\doibase 10.1070/PU1975v017n05ABEH004362}
  {\bibfield  {journal} {\bibinfo  {journal} {Soviet Physics Uspekhi}\ }\textbf
  {\bibinfo {volume} {17}},\ \bibinfo {pages} {644} (\bibinfo {year}
  {1975})}\BibitemShut {NoStop}%
\bibitem [{\citenamefont {Braginsky}\ \emph {et~al.}(1980)\citenamefont
  {Braginsky}, \citenamefont {Vorontsov},\ and\ \citenamefont
  {Thorne}}]{braginsky_quantum_1980}%
  \BibitemOpen
  \bibfield  {author} {\bibinfo {author} {\bibfnamefont {V.~B.}\ \bibnamefont
  {Braginsky}}, \bibinfo {author} {\bibfnamefont {Y.~I.}\ \bibnamefont
  {Vorontsov}}, \ and\ \bibinfo {author} {\bibfnamefont {K.~S.}\ \bibnamefont
  {Thorne}},\ }\href {\doibase 10.1126/science.209.4456.547} {\bibfield
  {journal} {\bibinfo  {journal} {Science}\ }\textbf {\bibinfo {volume}
  {209}},\ \bibinfo {pages} {547} (\bibinfo {year} {1980})}\BibitemShut
  {NoStop}%
\bibitem [{\citenamefont {Caves}\ \emph {et~al.}(1980)\citenamefont {Caves},
  \citenamefont {Thorne}, \citenamefont {Drever}, \citenamefont {Sandberg},\
  and\ \citenamefont {Zimmermann}}]{caves_measurement_1980}%
  \BibitemOpen
  \bibfield  {author} {\bibinfo {author} {\bibfnamefont {C.~M.}\ \bibnamefont
  {Caves}}, \bibinfo {author} {\bibfnamefont {K.~S.}\ \bibnamefont {Thorne}},
  \bibinfo {author} {\bibfnamefont {R.~W.~P.}\ \bibnamefont {Drever}}, \bibinfo
  {author} {\bibfnamefont {V.~D.}\ \bibnamefont {Sandberg}}, \ and\ \bibinfo
  {author} {\bibfnamefont {M.}~\bibnamefont {Zimmermann}},\ }\href {\doibase
  10.1103/RevModPhys.52.341} {\bibfield  {journal} {\bibinfo  {journal}
  {Reviews of Modern Physics}\ }\textbf {\bibinfo {volume} {52}},\ \bibinfo
  {pages} {341} (\bibinfo {year} {1980})}\BibitemShut {NoStop}%
\bibitem [{\citenamefont {Thorne}\ \emph {et~al.}(1978)\citenamefont {Thorne},
  \citenamefont {Drever}, \citenamefont {Caves}, \citenamefont {Zimmermann},\
  and\ \citenamefont {Sandberg}}]{thorne_quantum_1978}%
  \BibitemOpen
  \bibfield  {author} {\bibinfo {author} {\bibfnamefont {K.~S.}\ \bibnamefont
  {Thorne}}, \bibinfo {author} {\bibfnamefont {R.~W.~P.}\ \bibnamefont
  {Drever}}, \bibinfo {author} {\bibfnamefont {C.~M.}\ \bibnamefont {Caves}},
  \bibinfo {author} {\bibfnamefont {M.}~\bibnamefont {Zimmermann}}, \ and\
  \bibinfo {author} {\bibfnamefont {V.~D.}\ \bibnamefont {Sandberg}},\ }\href
  {\doibase 10.1103/PhysRevLett.40.667} {\bibfield  {journal} {\bibinfo
  {journal} {Physical Review Letters}\ }\textbf {\bibinfo {volume} {40}},\
  \bibinfo {pages} {667} (\bibinfo {year} {1978})}\BibitemShut {NoStop}%
\bibitem [{\citenamefont {Koopman}(1931)}]{koopman_hamiltonian_1931}%
  \BibitemOpen
  \bibfield  {author} {\bibinfo {author} {\bibfnamefont {B.~O.}\ \bibnamefont
  {Koopman}},\ }\href {\doibase 10.1073/pnas.17.5.315} {\bibfield  {journal}
  {\bibinfo  {journal} {Proceedings of the National Academy of Sciences}\
  }\textbf {\bibinfo {volume} {17}},\ \bibinfo {pages} {315} (\bibinfo {year}
  {1931})}\BibitemShut {NoStop}%
\bibitem [{\citenamefont {Tsang}\ and\ \citenamefont
  {Caves}(2012)}]{tsang_evading_2012}%
  \BibitemOpen
  \bibfield  {author} {\bibinfo {author} {\bibfnamefont {M.}~\bibnamefont
  {Tsang}}\ and\ \bibinfo {author} {\bibfnamefont {C.~M.}\ \bibnamefont
  {Caves}},\ }\href {\doibase 10.1103/PhysRevX.2.031016} {\bibfield  {journal}
  {\bibinfo  {journal} {Physical Review X}\ }\textbf {\bibinfo {volume} {2}},\
  \bibinfo {pages} {031016} (\bibinfo {year} {2012})};\ \Eprint
  {http://arxiv.org/abs/1203.2317} {arXiv:1203.2317 [quant-ph]} \BibitemShut
  {NoStop}%
\bibitem [{\citenamefont {Tsang}\ and\ \citenamefont
  {Caves}(2010)}]{tsang_coherent_2010}%
  \BibitemOpen
  \bibfield  {author} {\bibinfo {author} {\bibfnamefont {M.}~\bibnamefont
  {Tsang}}\ and\ \bibinfo {author} {\bibfnamefont {C.~M.}\ \bibnamefont
  {Caves}},\ }\href {\doibase 10.1103/PhysRevLett.105.123601} {\bibfield
  {journal} {\bibinfo  {journal} {Physical Review Letters}\ }\textbf {\bibinfo
  {volume} {105}},\ \bibinfo {pages} {123601} (\bibinfo {year} {2010})};\
  \Eprint {http://arxiv.org/abs/1006.1005} {arXiv:1006.1005 [quant-ph]}
  \BibitemShut {NoStop}%
\bibitem [{\citenamefont {Boulebnane}\ \emph {et~al.}(2019)\citenamefont
  {Boulebnane}, \citenamefont {Renes},\ and\ \citenamefont {Woods}}]{sami}%
  \BibitemOpen
  \bibfield  {author} {\bibinfo {author} {\bibfnamefont {S.}~\bibnamefont
  {Boulebnane}}, \bibinfo {author} {\bibfnamefont {J.~M.}\ \bibnamefont
  {Renes}}, \ and\ \bibinfo {author} {\bibfnamefont {M.~P.}\ \bibnamefont
  {Woods}},\ }\href@noop {} {\bibfield  {journal} {\bibinfo  {journal} {In
  preparation.}\ } (\bibinfo {year} {2019})}\BibitemShut {NoStop}%
\bibitem [{\citenamefont {Pauli}(1926)}]{pauli_quantentheorie_1926}%
  \BibitemOpen
  \bibfield  {author} {\bibinfo {author} {\bibfnamefont {W.}~\bibnamefont
  {Pauli}},\ }in\ \href {\doibase 10.1007/978-3-642-99593-4_1} {\emph {\bibinfo
  {booktitle} {{Quanten}}}},\ \bibinfo {series and number} {{Handbuch der
  Physik}},\ \bibinfo {editor} {edited by\ \bibinfo {editor} {\bibfnamefont
  {W.}~\bibnamefont {Bothe}}, \bibinfo {editor} {\bibfnamefont
  {J.}~\bibnamefont {Franck}}, \bibinfo {editor} {\bibfnamefont
  {P.}~\bibnamefont {Jordan}}, \bibinfo {editor} {\bibfnamefont
  {H.}~\bibnamefont {Kulenkampff}}, \bibinfo {editor} {\bibfnamefont
  {R.}~\bibnamefont {Ladenburg}}, \bibinfo {editor} {\bibfnamefont
  {W.}~\bibnamefont {Noddack}}, \bibinfo {editor} {\bibfnamefont
  {W.}~\bibnamefont {Pauli}}, \bibinfo {editor} {\bibfnamefont
  {P.}~\bibnamefont {Pringsheim}}, \ and\ \bibinfo {editor} {\bibfnamefont
  {H.}~\bibnamefont {Geiger}}}\ (\bibinfo  {publisher} {{Springer Berlin
  Heidelberg}},\ \bibinfo {address} {{Berlin, Heidelberg}},\ \bibinfo {year}
  {1926})\ pp.\ \bibinfo {pages} {1--278}\BibitemShut {NoStop}%
\bibitem [{\citenamefont {Galapon}(2002)}]{galapon_paulis_2002}%
  \BibitemOpen
  \bibfield  {author} {\bibinfo {author} {\bibfnamefont {E.}~\bibnamefont
  {Galapon}},\ }\href {\doibase 10.1098/rspa.2001.0874} {\bibfield  {journal}
  {\bibinfo  {journal} {Proceedings of the Royal Society of London. Series A:
  Mathematical, Physical and Engineering Sciences}\ }\textbf {\bibinfo {volume}
  {458}},\ \bibinfo {pages} {451} (\bibinfo {year} {2002})};\ \Eprint
  {http://arxiv.org/abs/quant-ph/9908033} {arXiv:quant-ph/9908033} \BibitemShut
  {NoStop}%
\bibitem [{\citenamefont {Julsgaard}\ \emph {et~al.}(2001)\citenamefont
  {Julsgaard}, \citenamefont {Kozhekin},\ and\ \citenamefont
  {Polzik}}]{julsgaard_experimental_2001}%
  \BibitemOpen
  \bibfield  {author} {\bibinfo {author} {\bibfnamefont {B.}~\bibnamefont
  {Julsgaard}}, \bibinfo {author} {\bibfnamefont {A.}~\bibnamefont {Kozhekin}},
  \ and\ \bibinfo {author} {\bibfnamefont {E.~S.}\ \bibnamefont {Polzik}},\
  }\href {\doibase 10.1038/35096524} {\bibfield  {journal} {\bibinfo  {journal}
  {Nature}\ }\textbf {\bibinfo {volume} {413}},\ \bibinfo {pages} {400}
  (\bibinfo {year} {2001})};\ \Eprint {http://arxiv.org/abs/quant-ph/0106057}
  {arXiv:quant-ph/0106057} \BibitemShut {NoStop}%
\bibitem [{\citenamefont {Wasilewski}\ \emph {et~al.}(2010)\citenamefont
  {Wasilewski}, \citenamefont {Jensen}, \citenamefont {Krauter}, \citenamefont
  {Renema}, \citenamefont {Balabas},\ and\ \citenamefont
  {Polzik}}]{wasilewski_quantum_2010}%
  \BibitemOpen
  \bibfield  {author} {\bibinfo {author} {\bibfnamefont {W.}~\bibnamefont
  {Wasilewski}}, \bibinfo {author} {\bibfnamefont {K.}~\bibnamefont {Jensen}},
  \bibinfo {author} {\bibfnamefont {H.}~\bibnamefont {Krauter}}, \bibinfo
  {author} {\bibfnamefont {J.~J.}\ \bibnamefont {Renema}}, \bibinfo {author}
  {\bibfnamefont {M.~V.}\ \bibnamefont {Balabas}}, \ and\ \bibinfo {author}
  {\bibfnamefont {E.~S.}\ \bibnamefont {Polzik}},\ }\href {\doibase
  10.1103/PhysRevLett.104.133601} {\bibfield  {journal} {\bibinfo  {journal}
  {Physical Review Letters}\ }\textbf {\bibinfo {volume} {104}},\ \bibinfo
  {pages} {133601} (\bibinfo {year} {2010})};\ \Eprint
  {http://arxiv.org/abs/0907.2453} {arXiv:0907.2453 [quant-ph]} \BibitemShut
  {NoStop}%
\bibitem [{\citenamefont {M{\o}ller}\ \emph {et~al.}(2017)\citenamefont
  {M{\o}ller}, \citenamefont {Thomas}, \citenamefont {Vasilakis}, \citenamefont
  {Zeuthen}, \citenamefont {Tsaturyan}, \citenamefont {Balabas}, \citenamefont
  {Jensen}, \citenamefont {Schliesser}, \citenamefont {Hammerer},\ and\
  \citenamefont {Polzik}}]{moller_quantum_2017}%
  \BibitemOpen
  \bibfield  {author} {\bibinfo {author} {\bibfnamefont {C.~B.}\ \bibnamefont
  {M{\o}ller}}, \bibinfo {author} {\bibfnamefont {R.~A.}\ \bibnamefont
  {Thomas}}, \bibinfo {author} {\bibfnamefont {G.}~\bibnamefont {Vasilakis}},
  \bibinfo {author} {\bibfnamefont {E.}~\bibnamefont {Zeuthen}}, \bibinfo
  {author} {\bibfnamefont {Y.}~\bibnamefont {Tsaturyan}}, \bibinfo {author}
  {\bibfnamefont {M.}~\bibnamefont {Balabas}}, \bibinfo {author} {\bibfnamefont
  {K.}~\bibnamefont {Jensen}}, \bibinfo {author} {\bibfnamefont
  {A.}~\bibnamefont {Schliesser}}, \bibinfo {author} {\bibfnamefont
  {K.}~\bibnamefont {Hammerer}}, \ and\ \bibinfo {author} {\bibfnamefont
  {E.~S.}\ \bibnamefont {Polzik}},\ }\href {\doibase 10.1038/nature22980}
  {\bibfield  {journal} {\bibinfo  {journal} {Nature}\ }\textbf {\bibinfo
  {volume} {547}},\ \bibinfo {pages} {191} (\bibinfo {year} {2017})};\ \Eprint
  {http://arxiv.org/abs/1608.03613} {arXiv:1608.03613 [quant-ph]} \BibitemShut
  {NoStop}%
\bibitem [{\citenamefont {Khalili}\ and\ \citenamefont
  {Polzik}(2018)}]{khalili_overcoming_2018}%
  \BibitemOpen
  \bibfield  {author} {\bibinfo {author} {\bibfnamefont {F.~Y.}\ \bibnamefont
  {Khalili}}\ and\ \bibinfo {author} {\bibfnamefont {E.~S.}\ \bibnamefont
  {Polzik}},\ }\href {\doibase 10.1103/PhysRevLett.121.031101} {\bibfield
  {journal} {\bibinfo  {journal} {Physical Review Letters}\ }\textbf {\bibinfo
  {volume} {121}},\ \bibinfo {pages} {031101} (\bibinfo {year} {2018})};\
  \Eprint {http://arxiv.org/abs/1710.10405} {arXiv:1710.10405 [quant-ph]}
  \BibitemShut {NoStop}%
\bibitem [{\citenamefont {Woods}\ \emph {et~al.}(2019)\citenamefont {Woods},
  \citenamefont {Silva},\ and\ \citenamefont
  {Oppenheim}}]{woods_autonomous_2019}%
  \BibitemOpen
  \bibfield  {author} {\bibinfo {author} {\bibfnamefont {M.~P.}\ \bibnamefont
  {Woods}}, \bibinfo {author} {\bibfnamefont {R.}~\bibnamefont {Silva}}, \ and\
  \bibinfo {author} {\bibfnamefont {J.}~\bibnamefont {Oppenheim}},\ }\href
  {\doibase 10.1007/s00023-018-0736-9} {\bibfield  {journal} {\bibinfo
  {journal} {Annales Henri Poincar{\'e}}\ }\textbf {\bibinfo {volume} {20}},\
  \bibinfo {pages} {125} (\bibinfo {year} {2019})};\ \Eprint
  {http://arxiv.org/abs/1607.04591} {arXiv:1607.04591 [quant-ph]} \BibitemShut
  {NoStop}%
\bibitem [{\citenamefont {Salecker}\ and\ \citenamefont
  {Wigner}(1958)}]{salecker_quantum_1958}%
  \BibitemOpen
  \bibfield  {author} {\bibinfo {author} {\bibfnamefont {H.}~\bibnamefont
  {Salecker}}\ and\ \bibinfo {author} {\bibfnamefont {E.~P.}\ \bibnamefont
  {Wigner}},\ }\href {\doibase 10.1103/PhysRev.109.571} {\bibfield  {journal}
  {\bibinfo  {journal} {Physical Review}\ }\textbf {\bibinfo {volume} {109}},\
  \bibinfo {pages} {571} (\bibinfo {year} {1958})}\BibitemShut {NoStop}%
\bibitem [{\citenamefont {Peres}(1980)}]{peres_measurement_1980}%
  \BibitemOpen
  \bibfield  {author} {\bibinfo {author} {\bibfnamefont {A.}~\bibnamefont
  {Peres}},\ }\href {\doibase 10.1119/1.12061} {\bibfield  {journal} {\bibinfo
  {journal} {American Journal of Physics}\ }\textbf {\bibinfo {volume} {48}},\
  \bibinfo {pages} {552} (\bibinfo {year} {1980})}\BibitemShut {NoStop}%
\bibitem [{\citenamefont {Van~Trees}(2001)}]{van_trees_detection_2001}%
  \BibitemOpen
  \bibfield  {author} {\bibinfo {author} {\bibfnamefont {H.~L.}\ \bibnamefont
  {Van~Trees}},\ }\href {\doibase 10.1002/0471221082} {\emph {\bibinfo {title}
  {Detection, Estimation, and Modulation Theory, {{Part I}}: {{Detection}},
  {{Estimation}}, and {{Linear Modulation Theory}}}}}\ (\bibinfo  {publisher}
  {{Wiley}},\ \bibinfo {address} {{New York; Chichester}},\ \bibinfo {year}
  {2001})\BibitemShut {NoStop}%
\bibitem [{\citenamefont
  {Tsang}(2009{\natexlab{a}})}]{tsang_time-symmetric_2009}%
  \BibitemOpen
  \bibfield  {author} {\bibinfo {author} {\bibfnamefont {M.}~\bibnamefont
  {Tsang}},\ }\href {\doibase 10.1103/PhysRevLett.102.250403} {\bibfield
  {journal} {\bibinfo  {journal} {Physical Review Letters}\ }\textbf {\bibinfo
  {volume} {102}},\ \bibinfo {pages} {250403} (\bibinfo {year}
  {2009}{\natexlab{a}})};\ \Eprint {http://arxiv.org/abs/0904.1969}
  {arXiv:0904.1969 [quant-ph]} \BibitemShut {NoStop}%
\bibitem [{\citenamefont {Tsang}(2009{\natexlab{b}})}]{tsang_optimal_2009}%
  \BibitemOpen
  \bibfield  {author} {\bibinfo {author} {\bibfnamefont {M.}~\bibnamefont
  {Tsang}},\ }\href {\doibase 10.1103/PhysRevA.80.033840} {\bibfield  {journal}
  {\bibinfo  {journal} {Physical Review A}\ }\textbf {\bibinfo {volume} {80}},\
  \bibinfo {pages} {033840} (\bibinfo {year} {2009}{\natexlab{b}})};\ \Eprint
  {http://arxiv.org/abs/0906.4133} {arXiv:0906.4133 [quant-ph]} \BibitemShut
  {NoStop}%
\bibitem [{\citenamefont {Tsang}(2010)}]{tsang_optimal_2010}%
  \BibitemOpen
  \bibfield  {author} {\bibinfo {author} {\bibfnamefont {M.}~\bibnamefont
  {Tsang}},\ }\href {\doibase 10.1103/PhysRevA.81.013824} {\bibfield  {journal}
  {\bibinfo  {journal} {Physical Review A}\ }\textbf {\bibinfo {volume} {81}},\
  \bibinfo {pages} {013824} (\bibinfo {year} {2010})};\ \Eprint
  {http://arxiv.org/abs/0909.2432} {arXiv:0909.2432 [quant-ph]} \BibitemShut
  {NoStop}%
\bibitem [{\citenamefont {Barchielli}\ \emph {et~al.}(1983)\citenamefont
  {Barchielli}, \citenamefont {Lanz},\ and\ \citenamefont
  {Prosperi}}]{barchielli_statistics_1983}%
  \BibitemOpen
  \bibfield  {author} {\bibinfo {author} {\bibfnamefont {A.}~\bibnamefont
  {Barchielli}}, \bibinfo {author} {\bibfnamefont {L.}~\bibnamefont {Lanz}}, \
  and\ \bibinfo {author} {\bibfnamefont {G.~M.}\ \bibnamefont {Prosperi}},\
  }\href {\doibase 10.1007/BF01906270} {\bibfield  {journal} {\bibinfo
  {journal} {Foundations of Physics}\ }\textbf {\bibinfo {volume} {13}},\
  \bibinfo {pages} {779} (\bibinfo {year} {1983})}\BibitemShut {NoStop}%
\bibitem [{\citenamefont {Caves}\ and\ \citenamefont
  {Milburn}(1987)}]{caves_quantum-mechanical_1987}%
  \BibitemOpen
  \bibfield  {author} {\bibinfo {author} {\bibfnamefont {C.~M.}\ \bibnamefont
  {Caves}}\ and\ \bibinfo {author} {\bibfnamefont {G.~J.}\ \bibnamefont
  {Milburn}},\ }\href {\doibase 10.1103/PhysRevA.36.5543} {\bibfield  {journal}
  {\bibinfo  {journal} {Physical Review A}\ }\textbf {\bibinfo {volume} {36}},\
  \bibinfo {pages} {5543} (\bibinfo {year} {1987})}\BibitemShut {NoStop}%
\bibitem [{\citenamefont {Wiseman}\ and\ \citenamefont
  {Milburn}(2009)}]{wiseman_quantum_2009}%
  \BibitemOpen
  \bibfield  {author} {\bibinfo {author} {\bibfnamefont {H.~M.}\ \bibnamefont
  {Wiseman}}\ and\ \bibinfo {author} {\bibfnamefont {G.~J.}\ \bibnamefont
  {Milburn}},\ }\href {\doibase 10.1017/CBO9780511813948} {\emph {\bibinfo
  {title} {Quantum {{Measurement}} and {{Control}}}}}\ (\bibinfo  {publisher}
  {{Cambridge University Press}},\ \bibinfo {year} {2009})\BibitemShut
  {NoStop}%
\bibitem [{\citenamefont {Billingsley}(1999)}]{billingsley_convergence_1999}%
  \BibitemOpen
  \bibfield  {author} {\bibinfo {author} {\bibfnamefont {P.}~\bibnamefont
  {Billingsley}},\ }\href {\doibase 10.1002/9780470316962} {\emph {\bibinfo
  {title} {Convergence of {{Probability Measures}}}}},\ Wiley {{Series}} in
  {{Probability}} and {{Statistics}}\ (\bibinfo  {publisher} {{John Wiley \&
  Sons, Inc.}},\ \bibinfo {address} {{Hoboken, NJ, USA}},\ \bibinfo {year}
  {1999})\BibitemShut {NoStop}%
\bibitem [{\citenamefont {Tsang}\ \emph {et~al.}(2011)\citenamefont {Tsang},
  \citenamefont {Wiseman},\ and\ \citenamefont
  {Caves}}]{tsang_fundamental_2011}%
  \BibitemOpen
  \bibfield  {author} {\bibinfo {author} {\bibfnamefont {M.}~\bibnamefont
  {Tsang}}, \bibinfo {author} {\bibfnamefont {H.~M.}\ \bibnamefont {Wiseman}},
  \ and\ \bibinfo {author} {\bibfnamefont {C.~M.}\ \bibnamefont {Caves}},\
  }\href {\doibase 10.1103/PhysRevLett.106.090401} {\bibfield  {journal}
  {\bibinfo  {journal} {Physical Review Letters}\ }\textbf {\bibinfo {volume}
  {106}},\ \bibinfo {pages} {090401} (\bibinfo {year} {2011})};\ \Eprint
  {http://arxiv.org/abs/1006.5407} {arXiv:1006.5407 [quant-ph]} \BibitemShut
  {NoStop}%
\bibitem [{\citenamefont {Woods}\ \emph {et~al.}(2018)\citenamefont {Woods},
  \citenamefont {Silva}, \citenamefont {P{\"u}tz}, \citenamefont {Stupar},\
  and\ \citenamefont {Renner}}]{woods_quantum_2018}%
  \BibitemOpen
  \bibfield  {author} {\bibinfo {author} {\bibfnamefont {M.~P.}\ \bibnamefont
  {Woods}}, \bibinfo {author} {\bibfnamefont {R.}~\bibnamefont {Silva}},
  \bibinfo {author} {\bibfnamefont {G.}~\bibnamefont {P{\"u}tz}}, \bibinfo
  {author} {\bibfnamefont {S.}~\bibnamefont {Stupar}}, \ and\ \bibinfo {author}
  {\bibfnamefont {R.}~\bibnamefont {Renner}},\ }\Eprint {http://arxiv.org/abs/1806.00491} {arXiv:1806.00491
  [quant-ph]} \BibitemShut {NoStop}%
\bibitem [{\citenamefont {Woods}\ and\ \citenamefont
  {Alhambra}(2019)}]{woods_continuous_2019}%
  \BibitemOpen
  \bibfield  {author} {\bibinfo {author} {\bibfnamefont {M.~P.}\ \bibnamefont
  {Woods}}\ and\ \bibinfo {author} {\bibfnamefont {{\'A}.~M.}\ \bibnamefont
  {Alhambra}},\ }\Eprint
  {http://arxiv.org/abs/1902.07725} {arXiv:1902.07725 [quant-ph]} \BibitemShut
  {NoStop}%
\bibitem [{\citenamefont {Faist}\ \emph {et~al.}(2019)\citenamefont {Faist},
  \citenamefont {Nezami}, \citenamefont {Albert}, \citenamefont {Salton},
  \citenamefont {Pastawski}, \citenamefont {Hayden},\ and\ \citenamefont
  {Preskill}}]{faist_continuous_2019}%
  \BibitemOpen
  \bibfield  {author} {\bibinfo {author} {\bibfnamefont {P.}~\bibnamefont
  {Faist}}, \bibinfo {author} {\bibfnamefont {S.}~\bibnamefont {Nezami}},
  \bibinfo {author} {\bibfnamefont {V.~V.}\ \bibnamefont {Albert}}, \bibinfo
  {author} {\bibfnamefont {G.}~\bibnamefont {Salton}}, \bibinfo {author}
  {\bibfnamefont {F.}~\bibnamefont {Pastawski}}, \bibinfo {author}
  {\bibfnamefont {P.}~\bibnamefont {Hayden}}, \ and\ \bibinfo {author}
  {\bibfnamefont {J.}~\bibnamefont {Preskill}},\ }\Eprint
  {http://arxiv.org/abs/1902.07714} {arXiv:1902.07714 [quant-ph]} \BibitemShut {NoStop}%
\end{thebibliography}%


\end{document}